\documentclass[comsoc,anonymous]{IEEEtran}

\usepackage{ifthen}
\usepackage{xspace}

\newboolean{talk}
\setboolean{talk}{false}
\newboolean{paper}
\setboolean{paper}{false}

\newboolean{IEEEstyle}
\setboolean{IEEEstyle}{false}
\newboolean{lipicsstyle}
\setboolean{lipicsstyle}{false}
\newboolean{eptcsstyle}
\setboolean{eptcsstyle}{false}
\newboolean{entcsstyle}
\setboolean{entcsstyle}{false}
\newboolean{sigplanstyle}
\setboolean{sigplanstyle}{false}
\newboolean{easychairstyle}
\setboolean{easychairstyle}{false}
\newboolean{scrartclstyle}
\setboolean{scrartclstyle}{false}
\newboolean{lncsstyle}
\setboolean{lncsstyle}{false}

\newboolean{acmmart}
\setboolean{acmmart}{false}

\newboolean{needstheorems}
\setboolean{needstheorems}{false}
\newboolean{withimages}
\setboolean{withimages}{false}
\newboolean{withproofs}
\setboolean{withproofs}{true}

\newboolean{french}
\setboolean{french}{false}

\setboolean{IEEEstyle}{true}
\usepackage{booktabs}   
\usepackage{subcaption} 
\usepackage[all]{xy}

\usepackage{ebproof}
\ebproofset{label separation=0.3em}

\usepackage{booktabs}   
\usepackage{subcaption} 
\usepackage{graphicx}
\usepackage[normalem]{ulem}
\usepackage{multirow}

\usepackage[utf8]{inputenc}

\usepackage{multirow}
\usepackage{cmll}
\usepackage{listings}

\usepackage{amsmath, amsfonts, 
	amsthm, mathtools}
\usepackage{stmaryrd}
\usepackage{mathtools}
\usepackage{enumerate}
\usepackage{listings}

\usepackage{fancybox}
\usepackage{hhline}
\usepackage{marginnote}
\usepackage{soul}
\usepackage{xcolor}
\usepackage{cleveref}
\usepackage{complexity}

\newcommand{\ignore}[1]{}

\newcommand{\myinput}[1]{\ifthenelse{\boolean{withimages}}{\input{#1}}{}}

\newcommand{\reflemma}[1]{Lemma~\ref{l:#1}}
\newcommand{\reflemmap}[2]{Lemma~\ref{l:#1}.\ref{p:#1-#2}}

\newcommand{\refrmk}[1]{Remark~\ref{rmk:#1}} 

\newcommand{\ie}{\textit{i.e.}\xspace}

\newcommand{\ih}{\textit{i.h.}\xspace}
\newcommand{\resp}{\textnormal{resp.}\xspace}

\newcommand{\ES}{\text{ES}\xspace}

\newcommand{\full}{\text{full}\xspace}
\newcommand{\Full}{\text{Full}\xspace}
\renewcommand{\full}{\text{strong}\xspace}
\renewcommand{\Full}{\text{Strong}\xspace}

\newcommand{\defeq}{\coloneqq} 
\newcommand{\eqdef}{\eqqcolon} 
\newcommand{\grameq}{\Coloneqq} 
\newcommand{\set}[1]{\{#1\}}

\newcommand{\nat}{\mathbb{N}}
\newcommand{\size}[1]{|#1|}

\newcommand{\mulsym}{\msym} 

\renewcommand{\l}{\lambda}
\newcommand{\isub}[2]{\{#1/#2\}}
\newcommand{\replace}[2]{#1{\shortleftarrow}#2}
\renewcommand{\isub}[2]{\{\replace{#1}{#2}\}}
\newcommand{\esub}[2]{[\replace{#1}{#2}]}
\renewcommand{\esub}[2]{[#1{\shortleftarrow}#2]}

\newcommand{\fv}[1]{{\sf fv}(#1)}

\newcommand{\rootRew}[1]{\mapsto_{#1}}

\newcommand{\Rew}[1]{\rightarrow_{#1}}

\newcommand{\lRew}[1]{\; \mbox{}_{#1}{\leftarrow}\ }

\newcommand{\lto}{\lRew{}}

\newcommand{\rtom}{\rootRew{\msym}}
\newcommand{\rtoe}{\rootRew{\esym}}

\newcommand{\betaplot}{\beta_v}
 
\newcommand{\tobvplot}{\Rew{\betaplot}} 
\newcommand{\rtobvplot}{\rootRew{\betaplot}} 

\newcommand{\esym}{{\mathsf e}}

\newcommand{\msym}{\mathsf{m}}

\newcommand{\ssym}{{\mathsf s}}
\newcommand{\wsym}{{\mathsf{o}}} 
\newcommand{\wmsym}{{\wsym\msym}} 
\newcommand{\wesym}{{\wsym\esym}} 
\newcommand{\omsym}{{\wmsym}} 
\newcommand{\oesym}{{\wesym}} 

\newcommand{\vmsym}{\mathsf{shuf}} 
\newcommand{\shuf}{\vmsym} 

\newcommand{\shufeqext}{\shufeqext} 

 \newcommand{\tom}{\Rew{\msym}}
 \newcommand{\toe}{\Rew{\esym}}

\newcommand{\tovsubo}{\Rew{\wsym}}

\newcommand{\tomo}{\Rew{\omsym}}

\newcommand{\toeo}{\Rew{\wsym{\esym}}}

\newcommand{\msolvsym}{\solvredsym\msym}
\newcommand{\esolvsym}{\solvredsym\esym}

\newcommand{\tovsubsolv}{\tosolv}
\newcommand{\tomsolv}{\Rew{\msolvsym}}
\newcommand{\toesolv}{\Rew{\esolvsym}}

\newcommand{\tm}{t}
\newcommand{\tmtwo}{u}
\newcommand{\tmthree}{r}
\newcommand{\tmfour}{q}
\newcommand{\tmfive}{p}
\newcommand{\tmsix}{s}

\newcommand{\tmp}{\tm'}
\newcommand{\tmtwop}{\tmtwo'}

\newcommand{\tmfourp}{\tmfour'}
\newcommand{\tmfivep}{\tmfive'}

\newcommand{\var}{x}
\newcommand{\vartwo}{y}
\newcommand{\varthree}{z}

\newcommand{\val}{v}

\newcommand{\ctxholep}[1]{\langle #1\rangle}
\newcommand{\ctxhole}{\ctxholep{\cdot}}

\newcommand{\ctx}{C}

\newcommand{\ctxp}[1]{\ctx\ctxholep{#1}}

\newcommand{\sctx}{L}
\newcommand{\sctxtwo}{\sctx'}

\newcommand{\sctxp}[1]{\sctx\ctxholep{#1}}
\newcommand{\sctxtwop}[1]{\sctxtwo\!\ctxholep{#1}}

\newcommand{\arbctxp}[1]{\arbctxp{#1}}
\newcommand{\arbctxtwop}[1]{\arbctxtwop{#1}}

\newcommand{\fctx}{F}

\newcommand{\fctxp}[1]{\fctx\ctxholep{#1}}

\ifthenelse{\boolean{acmmart}
}{
  
  }{
  
  }

\newcommand{\deriv}{d}

\newcommand{\sizehole}[2]{|#2|_{#1}}

\newcommand{\sizem}[1]{\sizehole{\msym}{#1}} 
\newcommand{\sizes}[1]{\sizehole{\ssym}{#1}} 
\newcommand{\sizeo}[1]{\sizehole{\osym}{#1}} 

\ifthenelse{\boolean{IEEEstyle}}{
\usepackage{amssymb}
\usepackage{amsthm}
    \newtheorem{theorem}{Theorem}[section]
    \newtheorem{lemma}[theorem]{Lemma}

    \newtheorem{proposition}[theorem]{Proposition}
    \newtheorem{definition}[theorem]{Definition}
}{}

\ifthenelse{\boolean{lncsstyle}}{}{\ifthenelse{\boolean{lipicsstyle}}{}{\newtheorem{remark}{Remark}[section]}}

\newcommand{\itm}{i}
\newcommand{\itmtwo}{\itm'}

\newcommand{\propersym}{\mathsf{p}}
\newcommand{\pitm}{\itm_\propersym}
\newcommand{\pitmtwo}{\itmtwo_\propersym}

\newcommand{\fire}{f}
\newcommand{\firetwo}{\fire'}

\newcommand{\sfire}{\fire_\fullsym}

\newcommand{\vsub}{\mathsf{vsc}} 
\newcommand{\VSC}{\textnormal{VSC}\xspace}

\newcommand{\tovsub}{\Rew{\vsub}}

\newcommand{\tow}{\Rew{\wsym}} 
\newcommand{\towm}{\Rew{\wmsym}} 
\newcommand{\towe}{\Rew{\wsym{\esym}}} 

\newcommand{\osym}{{\mathsf o}}

\newcommand{\la}[1]{\lambda #1.}

\newcommand{\myproof}[1]{
\ifthenelse{\boolean{omitproofs}}{\begin{IEEEproof} Proof available but omitted for readability. \end{IEEEproof}}{#1}}

\newcommand{\gregoire}{Gr{\'{e}}goire\xspace}

\newcommand{\withproofs}[1]{\ifthenelse{\boolean{withproofs}}{#1}{}}

\newcommand{\withoutproofs}[1]{\ifthenelse{\boolean{withproofs}}{}{#1}}

\newcommand{\NoteProof}[1]{
	\marginnote{{\normalfont\scriptsize{Proof\,p.\,{\pageref{#1}}\,}}}}
\newcommand{\NoteState}[1]{
	\marginnote{{\normalfont\scriptsize{See p.\,{\pageref{#1}}}}}}
\renewcommand{\NoteProof}[1]{\marginnote{{Proof\,p.\,{\pageref{#1}}}}}
\renewcommand{\NoteState}[1]{\marginnote{{See\,p.\,{\pageref{#1}}\\\cref{#1}}}}

\crefname{proposition}{Prop.}{Props.}
\crefname{theorem}{Thm.}{Thms.}
\crefname{lemma}{Lemma}{Lemmas}
\crefname{corollary}{Cor.}{Cors.}
\crefname{section}{Sect.}{Sects.}
\Crefname{section}{Section}{Sections}

\newcommand{\vsubcalc}{\lambda_\vsub}

\newcommand{\shufcalc}{\lambda_\shuf}

\newcommand{\plotsym}{\mathsf{Plot}}
\newcommand{\plotcalc}{\lambda_{\plotsym}}

\newcommand{\doubt}[1]{}

\newcommand{\letexp}{\mathsf{let}}

\newcommand{\lambdamucalc}{\overline\lambda\mu\tilde{\mu}}

\newcommand{\Rule}{\mathsf{r}}

\newcounter{numberone}
\newcounter{numberoneroman}
\newcounter{numberonealph}

\newcommand{\cbn}{CbN\xspace}

\newcommand{\cbv}{CbV\xspace}

\newcommand{\ocbv}{Open \cbv}

\newcommand{\scbv}{Strong \cbv}

\newcommand{\mset}[1]{[#1]}
\newcommand{\emptymset}{\mset{\,}}
\renewcommand{\emptymset}{\zero}
\newcommand{\zero}{\mathbf{0}}

\newcommand{\ltype}{\typefont{A}}
\newcommand{\ltypetwo}{\typefont{B}}

\newcommand{\typctx}{\Gamma}
\newcommand{\typctxtwo}{\Delta}
\newcommand{\typctxthree}{\Sigma}

\newcommand{\tderthree}{\rho}

\newcommand{\hastype}{\!:\!}

\newcommand{\domain}[1]{\mathsf{dom}(#1)}

\newcommand{\mytr}[1]{\underline{#1}}
\newcommand{\auxtr}[1]{\overline{#1}}

\makeatletter
\newcommand\Copy[2]{
        \marginpar{\scriptsize \ \ \hyperlink{hl-appendix-#1}{Proof p.\,{\pageref*{appendix-#1}}}}
	\immediate\write\@auxout{\unexpanded{\global\long\@namedef{mytext@#1}{#2}
  }}%
	#2%
}

\newcommand\Paste[1]{%
        \hypertarget{hl-appendix-#1}{}\label{appendix-#1}
	\renewcommand{\inappendix}[1]{}
	\ifcsname mytext@#1\endcsname
	\@nameuse{mytext@#1}%
	\else
	``??''
	\fi
	\renewcommand{\inappendix}[1]{#1}
}
\makeatother

\newcommand{\inappendix}[1]{#1}

\newcommand{\weakctx}{O}
\newcommand{\weakctxtwo}{\weakctx'}
\newcommand{\weakctxp}[1]{\weakctx\ctxholep{#1}}
\newcommand{\weakctxtwop}[1]{\weakctxtwo\ctxholep{#1}}

\newcommand{\openctx}{O}

\newcommand{\subctx}{\sctx}

\newcommand{\solvctx}{S}
\newcommand{\solvctxtwo}{\solvctx'}

\newcommand{\solvctxp}[1]{\solvctx\ctxholep{#1}}
\newcommand{\solvctxtwop}[1]{\solvctxtwo\ctxholep{#1}}

\newcommand{\larrow}[2]{#1 \multimap #2}
\newcommand{\ground}{\typefont{X}}

\newcommand{\Ax}{\mathsf{ax}}
\newcommand{\Es}{\mathsf{es}}

\newcommand{\derive}[2]{#1 \vartriangleright #2}
\newcommand{\concl}[4]{\derive{#1}{#2 \vdash #3 \hastype #4}}

\newcommand{\subctxp}[1]{\subctx\ctxholep{#1}}

\newcommand{\fullsym}{{\mathsf{f}}}
\renewcommand{\fullsym}{{\mathsf{s}}}
\newcommand{\sitm}{\itm_\fullsym}
\newcommand{\sitmtwo}{\itmtwo_\fullsym}

\newcommand\Crumb\mytr
\newcommand\CrumbAux\auxtr

\makeatletter
\newcommand{\exder}{%
  \def\exderW[##1]{\triangleright_{##1}\ }%
  \def\exderWO{\triangleright\ }%
  \@ifnextchar[\exderW\exderWO%
  }
\makeatother

\newcommand{\tderiv}{\Phi}
\newcommand{\tderivtwo}{\Psi}
\newcommand{\tderivtwop}{\tderivtwo'} 
\newcommand{\tderivthree}{\Theta}

\newcommand{\typefont}[1]{{\mathsf{#1}}}
\newcommand{\mtype}{\typefont{M}}
\newcommand{\mtypetwo}{\typefont{N}}
\newcommand{\mtypethree}{\typefont{O}}

\newcommand\emptytype{\mathbf{0}}
\newcommand{\type}{\typefont{T}}
\newcommand{\imtype}{\inertop{\mtype}}
\newcommand{\imtypetwo}{\inertop{\mtypetwo}}
\newcommand{\inltype}{\inertop{\ltype}}
\newcommand\mplus{\uplus}

\newcommand{\tyjp}[4]{{#3} \vdash^{#1} #2 \hastype #4}
\newcommand{\namedtyjp}[5]{#1 \vartriangleright \tyjp{#2}{#3}{#4}{#5}}

\newcommand{\dom}[1]{\mathsf{dom}(#1)}

\newcommand{\ruleApp}{@}
\newcommand{\ruleFun}{\lambda}
\newcommand{\ruleES}{\mathsf{es}}
\newcommand{\ruleMany}{\mathsf{many}}
\newcommand{\ruleManyVar}{\ruleMany}
\newcommand{\ruleManyVal}{\ruleMany}
\newcommand{\ruleAp}{@}
\newcommand{\ruleAx}{\mathsf{ax}}

\newcommand{\I}{I}
\newcommand{\J}{J}
\renewcommand{\K}{K}
\newcommand{\iI}{{i \in \I}}
\newcommand{\jJ}{{j \in \J}}
\newcommand{\kK}{{k \in \K}}

\newcommand{\tarrow}[2]{#1 \multimap #2}
\newcommand{\ty}[2]{\tarrow{#1}{#2}}

\newcommand{\mult}[1]{[ #1 ] }

\newcommand{\bigmplus}{\biguplus}

\newcommand{\Id}{{\mathsf{I}}}

\newcommand{\solvredsym}{\mathsf{s}}
\newcommand{\tosolv}{\Rew{\solvredsym}}

\newcommand{\tosolvm}{\Rew{\solvredsym\mulsym}}
\newcommand{\tosolve}{\Rew{\solvredsym\esym}}

\newcommand{\sltype}{\ltype^{\solvsym}}

\newcommand{\smtype}{\mtype^{\solvsym}}
\newcommand{\smtypetwo}{\mtypetwo^{\solvsym}}

\newcommand{\unitarysym}{{\textsc{u}}}

\newcommand{\usltype}{\ltype^{\unitarysym\solvsym}}

\newcommand{\usmtype}{\mtype^{\unitarysym\solvsym}}

\newcommand{\inertsym}{{\textsc{i}}}
\newcommand{\inertop}[1]{#1^{\mathsf{i}}}
\newcommand{\isltype}{\ltype^{\inertsym\solvsym}}

\newcommand{\ismtype}{\mtype^{\inertsym\solvsym}}

\newcommand{\precisesym}{{\textsc{p}}}

\newcommand{\psmtype}{\mtype^{\precisesym\solvsym}}
\newcommand{\psmtypetwo}{\mtypetwo^{\precisesym\solvsym}}

\newcommand{\solvsym}{\textsc{s}}
\newcommand{\solvnf}{\fire_{\solvredsym}}
\newcommand{\solvnftwo}{\solvnf'}

\newcommand{\myparagraph}[1]{\emph{#1.}}
\newcommand{\myparagraphsp}[1]{\medskip

\myparagraph{#1}}

\newcommand\mydots{\hbox to .6em{.\hss.}}

\renewcommand{\Full}{Full\xspace}
\renewcommand{\full}{full\xspace}

\renewcommand{\tmthree}{s}

\author{\IEEEauthorblockN{Beniamino Accattoli,}
\IEEEauthorblockA{Inria \& LIX, \'Ecole Polytechnique \\}
\and
\IEEEauthorblockN{Giulio Guerrieri,}
\IEEEauthorblockA{Huawei Edinburgh Research Centre}
}

\begin{document}
\onecolumn
\title{Call-by-Value Solvability and Multi Types}
\maketitle

\begin{abstract}
This paper provides a characterization of call-by-value solvability using call-by-value multi types. Our work is based on Accattoli and Paolini's characterization of  call-by-value solvable terms as those terminating with respect to the \emph{solving strategy} of the \emph{value substitution calculus}, a refinement of Plotkin's call-by-value $\l$-calculus. Here we show that the solving strategy terminates on a term $\tm$ if and only if $\tm$ is typable in a certain way in the multi type system induced by \emph{Ehrhard's call-by-value relational semantics}. Moreover, we show how to extract from the type system exact bounds on the length of the solving evaluation and on the size of its normal form, adapting \emph{de Carvalho's technique} for call-by-name.

%
\end{abstract}


\section{Introduction}
\label{sect:intro}

Plotkin's call-by-value $\l$-calculus \cite{DBLP:journals/tcs/Plotkin75} is at 
the heart  of 
programming languages such as OCaml and proof assistants such as Coq. In the study of programming 
languages, call-by-value (\cbv) evaluation is usually \emph{weak}, that is, it does not reduce under 
abstractions, and terms are assumed to be \emph{closed}, \ie , without free variables.
These constraints give rise to an elegant framework---we call it \emph{Closed \cbv},  following 
\cite{DBLP:conf/aplas/AccattoliG16}.

It often happens, however, that one needs to go beyond Closed \cbv  by 
considering \emph{Strong \cbv}, which is the extended setting where reduction under abstractions is allowed 
and terms may be open, or the intermediate framework of \emph{Open \cbv}, where evaluation is weak 
but terms are not necessarily closed. 
The need arises, most notably, when describing the implementation model of Coq, as done by \gregoire and Leroy
\cite{DBLP:conf/icfp/GregoireL02}, to realize the essential conversion test for dependent types. Other 
motivations lie in the study of bisimulations by Lassen
\cite{DBLP:conf/lics/Lassen05}, partial evaluation \cite{Jones:1993:PEA:153676}, or various topics 
of a semantical or logical nature, recalled below.\smallskip

\myparagraphsp{Na\"ive Extension of \cbv} In call-by-name (\cbn) turning to open terms  or strong 
evaluation is harmless because \cbn does not impose any special form to the arguments of 
$\beta$-redexes. 
On the contrary, turning to Open  or Strong \cbv is delicate. While some fundamental properties such 
as confluence and standardization hold also in such cases, as showed by Plotkin's himself \cite{DBLP:journals/tcs/Plotkin75}, 
others---typically of a semantical nature---break as soon as one considers open terms. 

The problems of \scbv can be traced back to Plotkin's seminal paper,  
where he points out the 
incompleteness of \cbv with respect to CPS translations, an issue later solved with categorical 
tools by Moggi \cite{DBLP:conf/lics/Moggi89}. This 
led to a number of studies, among others
\cite{DBLP:journals/lisp/SabryF93,DBLP:journals/toplas/SabryW97,DBLP:journals/tcs/MaraistOTW99,
DBLP:conf/icfp/CurienH00,DBLP:journals/logcom/DyckhoffL07,DBLP:conf/tlca/HerbelinZ09}, that 
introduced many proposals of 
improved calculi~for~\cbv.

The relationship with denotational semantics is also problematic, as  first shown by Paolini and Ronchi della Rocca
\cite{DBLP:journals/ita/PaoliniR99,DBLP:conf/ictcs/Paolini01,parametricBook}. There are two subtle 
points: 
\begin{enumerate}
\item \emph{Solvability}: the adaptation of the notion of solvability---roughly, 
a form of 
meaningfulness for terms---to \cbv;
\item \emph{Adequacy}: denotational semantics that are 
\emph{adequate} for Closed \cbv \cite{DBLP:conf/csl/AbramskyM97,DBLP:journals/fuin/EgidiHR92,DBLP:journals/tcs/HondaY99,DBLP:journals/mscs/PravatoRR99} are no longer adequate for the extended settings. Roughly, 
there are terms that are semantically divergent, that is, with 
trivial semantics (or unsolvable), 
while they are normal forms with respect to Plotkin's rules, and 
so are expected to have non-trivial 
semantics (and be solvable). 
\end{enumerate}
 Both semantical issues of Open/Strong \cbv have been addressed in the literature, relying of linear logic tools. Adequacy is studied at length by Accattoli and Guerrieri in \cite{DBLP:conf/aplas/AccattoliG18,DBLP:journals/corr/abs-2104-13979}, thus we here focus on solvability.\smallskip

\paragraph*{Call-by-Value Solvability} About solvability, Accattoli and Paolini \cite{AccattoliPaolini12} 
characterize operationally solvable terms using a calculus isomorphic to the proof-net \cbv representation of $\l$-calculus, the \emph{value substitution calculus} (shortened to \VSC).
Namely, they introduce a \emph{solving evaluation strategy} (called \emph{stratified-weak} in \cite{AccattoliPaolini12}) in the \VSC that terminates if and only if the term is solvable---notably their strategy requires \emph{strong} evaluation. 
This is akin to what happens in \cbn, where solvable term are 
those for which head evaluation terminates. A key point is that  
such a characterization of \cbv solvability is impossible in Plotkin's original formulation of \cbv. 

We would like to mention that the literature contains also a detailed study of \cbv solvability due to Garcia-Perez and Nogueira \cite{DBLP:journals/corr/Garcia-PerezN16} based on a different approach, as they stick to Plotkin's call-by-value $\l$-calculus.


Here we continue Accattoli and Paolini's study of \cbv solvability, providing two contributions:
\begin{enumerate}
\item \emph{Types and solvability}: we characterize \cbv solvability using \emph{multi types}, a variant of intersection types surveyed below. Namely, we prove that a term is \cbv solvable if and only if it is typable in a certain way;
\item \emph{Bounds from types}: we show how to extract, from a certain class of typing derivations, the number of steps taken by the solving strategy on a solvable term, together with the size of the solving normal form. 
\end{enumerate}
Before giving more details about our results, we recall multi types and their use for extracting operational bounds.\smallskip

\paragraph*{Multi Types} Intersection types are one of the standard tools to study $\l$-calculi, mainly used to characterize termination properties---classical references are Coppo and Dezani \cite{DBLP:journals/aml/CoppoD78,DBLP:journals/ndjfl/CoppoD80}, Pottinger \cite{Pottinger80}, and Krivine \cite{Kri}.  In contrast to other type systems, they do not provide a logical interpretation, at least not as smoothly as for simple or polymorphic types---see Ronchi Della Rocca and Roversi's \cite{DBLP:conf/csl/RoccaR01} or Bono, Venneri, and Bettini's \cite{DBLP:journals/tcs/BonoVB08} for details. They are better understood, in fact, as syntactic presentations of denotational semantics: they are invariant under evaluation and type all and only the terminating terms, thus naturally providing an adequate denotational model.

Intersection types are a flexible tool that can
be formulated in various ways. A flavour that emerged 
in the last 10 years is that of \emph{non-idempotent} intersection
types, where the
intersection $A \cap A$ is not equivalent to $A$. They were first considered by Gardner \cite{DBLP:conf/tacs/Gardner94}, and then Kfoury \cite{DBLP:journals/logcom/Kfoury00}, Neergaard and Mairson \cite{DBLP:conf/icfp/NeergaardM04}, and de Carvalho \cite{Carvalho07,deCarvalho18} provided a first wave of works abut them---a survey can be
found in Bucciarelli, Kesner, and Ventura's~\cite{BKV17}. Non-idempotent intersections can be seen as multisets, which is why, to ease the
language, we prefer to call them \emph{multi types} rather than
\emph{non-idempotent intersection types}.

Multi types retain the denotational character of intersection types, and they actually refine it along two correlated lines. First, taking types with multiplicities gives rise to a \emph{quantitative} approach, that reflects 
resource consumption in the evaluation of terms. Second, such a quantitative feature turns out to coincide exactly with the one at work in linear logic. 
Some care is needed here: multi types do not correspond to linear logic formulas, rather to the relational denotational semantics of linear logic (two seminal references for such a semantic are Girard's~\cite{Girard88} and Bucciarelli and Ehrhard's \cite{DBLP:journals/apal/BucciarelliE01}; see also \cite{DBLP:conf/csl/Carvalho16,GuerrieriPellissierTortora16})---similarly to intersection types, they provide a denotational rather than a logical interpretation.\smallskip

\paragraph*{De Carvalho's Bounds from Multi Types} An insightful use of multi types is de Carvalho's connection between the size of types and the size of normal forms, and between the size of type derivations and evaluation lengths for the 
\cbn $\l$-calculus \cite{deCarvalho18}. He shows how from a certain class of typing derivations one can extract \emph{exact} bounds about the length of evaluations and the size of the normal form of a term, according to various notions of evaluation. In particular, he shows how to do it for \emph{head} reduction, which in \cbn is the strategy characterizing solvability, that is, for which terminating terms coincide with solvable terms.\smallskip

 \paragraph*{De Carvalho's Legacy} De Carvalho developed his results in his PhD defended in 2007 \cite{Carvalho07}, known by the community thanks to a technical report that was eventually published much later \cite{deCarvalho18}.
Soon after his PhD, he adapted his work to linear logic, 
with Pagani and Tortora de Falco \cite{DBLP:journals/tcs/CarvalhoPF11,DBLP:journals/iandc/CarvalhoF16}. 
A few years later, Bernadet and Graham-Lengrand adapted his work to measure the longest evaluation in the $\l$-calculus \cite{DBLP:journals/corr/BernadetL13}.

At the time, it was not known whether it would make sense to count the number of $\beta$-steps (or linear logic cut-elimination steps) as a reasonable measure of complexity. After this was clarified (in the positive, for $\beta$ steps) by Accattoli and Dal Lago \cite{DBLP:journals/corr/AccattoliL16}, de Carvalho's work has been revisited by Accattoli, Graham-Lengrand, and Kesner in 2018 \cite{DBLP:journals/pacmpl/AccattoliGK18}. The revisitation started a new wave of works adapting de Carvalho's study to many evaluation strategies and extensions of the $\l$-calculus, including call-by-value \cite{DBLP:conf/aplas/AccattoliG18}, call-by-need \cite{DBLP:conf/esop/AccattoliGL19}, fully lazy call-by-need \cite{DBLP:conf/fossacs/KesnerPV21}, a linear logic presentation of call-by-push-value \cite{DBLP:conf/flops/BucciarelliKRV20,DBLP:conf/csl/KesnerV22}, the $\lambda\mu$-calculus \cite{DBLP:conf/lics/KesnerV20}, the $\l$-calculus with pattern matching \cite{DBLP:conf/types/AlvesKV19}, the probabilistic $\l$-calculus \cite{DBLP:journals/pacmpl/LagoFR21}, and the abstract machine underlying the geometry of interaction \cite{DBLP:journals/pacmpl/AccattoliLV21,DBLP:conf/lics/AccattoliLV21}. \smallskip

\paragraph*{\ocbv and Multi Types} Of all these studies, the relevant one for this paper is Accattoli and Guerrieri's study of \ocbv  \cite{DBLP:conf/aplas/AccattoliG18}, which rests on Ehrhard's \cbv relational semantics \cite{DBLP:conf/csl/Ehrhard12} reformulated as multi types. They show that the evaluation of a term $\tm$ terminates in \ocbv  if and only if $\tm$ is typable with \cbv multi types, and that in that case $\tm$ is typable with the empty multi set, that they note $\zero$. Moreover, they show how to extract bounds to evaluation lengths and to the size of normal forms from type derivations, and characterize those type derivations that give exact bounds.\smallskip

\paragraph*{Contribution 1: Qualitative Characterization of Solvability} Here, we study the solving strategy which is an extension of evaluation in \ocbv. Namely, it is a strong strategy that iterates \ocbv under abstractions, but not all abstractions, only under \emph{head} abstractions. We study it via Ehrhard's \cbv multi type system, as in  \cite{DBLP:conf/aplas/AccattoliG18}. Since the strategy \emph{extends} open evaluation, the terms that are terminating for the solving strategy---that is, the solvable terms---cannot be characterized simply as the typable ones. 

The typical example is $\la\var\Omega$, which is normal for open evaluation but unsolvable (and thus diverging for the solving strategy). It is typable with $\zero$, and only with $\zero$. One is then tempted to characterize solvable terms as those ones typable with a type different from $\zero$. Unfortunately, things are slightly subtler, as $\la\vartwo\la\var\Omega$ is also normal for open evaluation but unsolvable, and it is typable with $\zero$ and, for instance, also with $\mset{\larrow\mtype\zero}$. We then define a notion of \emph{solvable} multi type, which is a multi type that is not $\zero$ and that contains types that (recursively) do not have $\zero$ on the right of $\multimap$. Our result is that a term is \cbv solvable if and only if it is typable with a solvable \cbv multi type.

Our notion of solvable type is not new, as in fact the idea appears in the literature in other papers using intersection types to study \cbv solvability \cite{DBLP:journals/ita/PaoliniR99,DBLP:conf/fscd/KerinecMR21}. Those type systems, however, are defective, and the characterizations that they claim are not in fact proved because their systems do not verify subject reduction  (for \cite{DBLP:journals/ita/PaoliniR99} subject expansion also fails), as we detail in Appendix \ref{sect:counter}. Therefore, a first contribution of this work is to recast their idea in a setting where it fully works.

\smallskip

\paragraph*{Contribution 2: Quantitative Characterization of Solvability} The second contribution of this work is that we exploit the quantitative character of multi types and extract bounds from type derivations, building over de Carvalho's technique. First, we show that every solvable derivations provides bounds to the length of the solving evaluation and to the size of the solving normal form. Second, we characterize the solvable derivations that provide exact bounds. This last part requires to introduce two refinements of solvable types, as solvable type derivations are in general too permissive with respect to quantitative bounds. 

On the one hand, we need to ensure that sub-terms evaluated by the solving strategy are typed \emph{exactly} once (with multi types a sub-term can naturally be typed many times). Such a constraint gives rise to \emph{unitary solvable types}. On the other hand, we need to ensure that sub-terms not touched by the solving strategy are \emph{not} typed. This is obtained by constraining types appearing in the typing context and on the left of $\multimap$, and gives rise to the notion of \emph{inertly solvable types}. We then prove that when a solvable type derivation is both unitary and inert---which we shall refer to as being \emph{precisely solvable}---then it captures \emph{exactly} the evaluation length and the size of the normal form for the solving strategy.

Such a quantitative study of \cbv solvability is the first one in the literature, and the described refinements of solvable types are also new.\smallskip

\paragraph*{Reasonable Cost Models} In \cite{DBLP:conf/lics/AccattoliCC21}, Accattoli \textit{et al}. prove that the number of $\beta$ steps of a strategy of the VSC, namely the \emph{external} one, is a reasonable time cost model for \cbv (where \emph{reasonable} means polynomially related to the time cost model of Turing machines). The solving strategy studied here is a sub-strategy of the external one, and so the result transfers. Thus, we obtain that multi types provide quantitative bounds that are meaningful from a computational complexity point of view.

\smallskip
\paragraph*{Syntactic Variants} We study \scbv via the VSC, to stress the linear logic background and foundation of our study. Everything could equivalently be  easily  reformulated inside (the intuitionistic and \cbv fragment of) Curien and Herbelin's
$\lambdamucalc$-calculus, following the isomorphism with the VSC developed in \cite{DBLP:conf/aplas/AccattoliG16} 
and preserving the number of steps/cost model. The solving strategy can also be reformulated using the \emph{shuffling calculus} $\shufcalc$ by Guerrieri and Carraro
\cite{DBLP:conf/fossacs/CarraroG14}, a calculus with commuting conversions used recently also by Manzonetto \textit{et al}.
\cite{DBLP:journals/fuin/ManzonettoPR19}.  Commuting conversion steps, however, cannot be counted via multi 
types (but they should), and the cost model of $\shufcalc$ is unclear (see 
\cite{DBLP:conf/aplas/AccattoliG16}). Therefore, $\shufcalc$ can be used for a qualitative study of \cbv solvability, but not for a
quantitative one as we do here. Calculi with $\letexp$-commutation rules such as the one by Herbelin and Zimmerman \cite{DBLP:conf/tlca/HerbelinZ09} simply 
can be seen as subcalculi of the VSC (up to structural equivalence, see \cite{AccattoliPaolini12}). Similar remarks apply to many other \cbv calculi
\cite{DBLP:conf/lics/Moggi89,DBLP:journals/lisp/SabryF93,DBLP:journals/toplas/SabryW97,DBLP:journals/tcs/MaraistOTW99, 
DBLP:journals/logcom/DyckhoffL07,DBLP:conf/csl/Santo20}. The key point of the VSC (valid also 
in $\lambdamucalc$) is that it does not need any commuting rule. 
\smallskip

\paragraph*{Proofs} Proofs are in the Appendix. \smallskip

\paragraph*{Historical Note} This paper has been uploaded to Arxiv in February 2022 but the results were developed by the authors in 2020. They were part of a larger and rejected submission to LICS 2021 that included also the results of \cite{DBLP:journals/corr/abs-2104-13979}.
\section{Value Substitution Calculus}
\label{sect:vsc}

Here we present the \emph{value substitution calculus} (\VSC for short) introduced by Accattoli and Paolini \cite{AccattoliPaolini12}, and we recall some properties.
The \VSC is a $\lambda$-calculus with let-expressions whose reduction rules mimic cut-elimination on proof-nets, via Girard's \cbv translation $(A \Rightarrow B)^v = \oc (A^v \multimap B^v)$ of intuitionistic logic into linear logic, as explained in \cite{DBLP:journals/tcs/Accattoli15}.  

In \VSC, $\beta$-redexes are decomposed via let-expressions, and the
\emph{by-value} restriction on evaluation is on the let-substitution rule, not on $\beta$-redexes, because only values
can be substituted.
A let-expression is formulated as an \emph{explicit substitution} or \emph{sharing} (\ES for short) $\tm\esub{\var}{\tmtwo}$ which binds $\var$ in $\tm$.
All along the
paper we use (many notions of) \emph{contexts}, \ie terms with a hole, noted $\ctxhole$. For now, we need \emph{substitution contexts} $\sctx$, which are simply lists of \ES. The grammars are:
\begin{center}
$\arraycolsep=3pt\begin{array}{rrl}
\textsc{Values } & \val & \grameq \var \mid  \la\var\tm 
\\
\textsc{Terms } & \tm,\tmtwo, \tmthree & \grameq \val \mid \tm\tmtwo 
\mid \tm \esub\var\tmtwo 
\\
\textsc{Substitution Ctxs } &\subctx & \grameq \ctxhole \mid \subctx \esub\var\tm
\end{array}$
\end{center}

The set of free variables of term $\tm$ is denoted by $\fv{\tm}$. 
Plugging a term $\tm$ in a context $\ctx$ is noted $\ctxp{\tm}$, possibly~capturing variables.
An \emph{answer} is a term of the shape $\sctxp\val$, where $\val$ is a value (\ie an abstraction) and $\sctx$ is a substitution context.
We use $\tm \isub{\var}{\tmtwo}$ for the capture-avoiding substitution of
$\tm$ for each free occurrence of $\var$ in $\tm$.
There are two kinds of rewrite rules, both work \emph{at a distance}, that is, up to a substitution context. 
\begin{center}
$\begin{array}{rr@{\ }l@{\ }l}
    \textsc{Multiplicative rule} & \subctxp{\la\var\tm}\tmtwo &  \rtom  & \subctxp{\tm\esub{\var}{\tmtwo}} \\
    \textsc{Exponential rule}  & \tm\esub\var{\subctxp{\val}} &  \rtoe  & \subctxp{\tm\isub{\var}{\val}} 
\end{array}$    
\end{center}

We shall consider three fragments of the  \VSC. 
They all contain the terms of \VSC, they differ only in the choice of evaluation contexts for the rewrite rules.\smallskip

\paragraph*{The Open VSC} We first focus on the open fragment of the \VSC, where rewriting is forbidden under abstraction and terms are possibly open (but not necessarily). This fragment has a nice inductive description of its normal forms, called fireballs, that is the starting point for many other definitions in the paper. Open contexts and rules are defined as follows.
\begin{center}
		$\begin{aligned}
		\textsc{Open ctxs} && \weakctx & \grameq \ctxhole \mid \weakctx \tm \mid \tm \weakctx \mid \weakctx \esub\var\tm \mid 
\tm\esub\var \weakctx 
		\end{aligned}$

		\begin{tabular}{ccc}
\textsc{Open rewrite rules}:
			&
			\multirow{2}{*}{\begin{prooftree}
					\hypo{\tm \rootRew{a} \tm'}		
					\infer1{\weakctxp{\tm} \Rew{\wsym a} \weakctxp{\tm'}}
			\end{prooftree}}
		\\
		($a \in \set{\msym,\esym}$)
	\end{tabular}\medskip

		$\begin{array}{cccc}
\textsc{Open reduction}:& \tovsubo  \, \defeq \, \tomo \cup \toeo
	\end{array}$
\end{center}

\begin{proposition}[Properties of the open reduction]
	\label{prop:ovsc-diamond}\label{prop:properties-open-reduction}
	\NoteProof{propappendix:properties-open-reduction}
	\begin{enumerate}
		\item \label{p:properties-open-reduction-diamond} $\tovsubo$ is diamond; $\tomo$ and $\toeo$ strongly commute.
		\item \label{p:properties-open-redction-harmony} A term is $\osym$-normal if and only if it is a fireball, where \emph{fireballs} (and \emph{proper inert terms}) are defined by:
	\end{enumerate}
\begin{alignat*}{5}
\textsc{Proper inert terms } \ & \pitm \grameq \var \fire \mid \pitm \fire \mid \pitm \esub{\var}{\pitmtwo}
\qquad \ &
\textsc{Fireballs } \ & \fire \grameq \val \mid \pitm \mid \fire \esub{\var}{\pitm}
\end{alignat*}
\end{proposition}

Diamond of $\tovsubo$ and strong commutation of $\tomo$ and $\toeo$ are technical facts (see  \Cref{sect:preliminaries} for definitions) with relevant consequences:
$\tovsubo$ is confluent and its non-determinism is only apparent, because if an $\osym$-evaluation from $\tm$ reaches
a $\osym$-normal form $\tmtwo$, then \emph{every} $\osym$-evaluation from $\tm$ eventually ends in $\tmtwo$;
and all these $\osym$-evaluations have same length and same number of $\msym$-steps and $\esym$-steps.  
Same properties hold for other strategies we shall~consider. 


%
	
The \emph{Closed \VSC} is the restriction of the Open \VSC to closed terms. Its normal forms are the closed values, that is, closed abstractions (the only fireballs that are closed).

\begin{remark}[Inert terms]\label{rmk:inert-def}
	Proper inert terms have the form of applications iterated $n > 0$ times, possibly interleaved with ES, starting with a head free variable. 
	It is natural to see a variable as a ``degenerate'' case, and to define \emph{inert terms} as variables or proper inert terms.
	\begin{equation*}
		\textsc{Inert terms } \ \itm \grameq \var \mid \pitm 
	\end{equation*}
	
	In our study, inert terms (including variables) play a crucial role. 
	For instance, for inert terms  we shall prove stronger claims about typability to have the right inductive hypothesis.
	In \cbv, variables have somehow a double nature as both inert terms (the stronger statements for inert terms also hold for variables) and values (the exponential step can fire $\tm\esub{\var}{\vartwo}$). 
\end{remark}
	 
\paragraph*{The Strong/Full VSC} To avoid notation clashes between the Solvable \VSC and the Strong \VSC (both would start with 's'), we refer to the \emph{Strong} \VSC as the \emph{\Full \VSC}. The \Full \VSC is obtained by allowing rewriting rules  to be applied everywhere in a term. 
\begin{center}
		$\begin{array}{l r rclccc}
		\textsc{Full ctxs} & \fctx  \grameq \ctxhole \mid \fctx \tm \mid \tm \fctx \mid \la{\var}{\fctx} \mid 
\fctx \esub\var\tm \mid \tm \esub \var \fctx 
		\end{array}$\medskip

		\begin{tabular}{ccc}
\textsc{Full rewrite rules:}
			&
			\multirow{2}{*}
			{\begin{prooftree}
    \hypo{\tm \rootRew{a} \tm'}		
    \infer1{\fctxp\tm \Rew{a} \fctxp{\tm'}}
			\end{prooftree}}
		\\
		($a \in \set{\msym,\esym}$)
	\end{tabular}\medskip

		$\begin{array}{r l@{\ } l@{\ } lllllll}
  \textsc{\Full reduction:} & \tovsub  & \defeq  &\tom \cup \toe
\end{array}$
\end{center}
Reduction $\tovsub$ is not diamond: 
see all the $\vsub$-evaluations of $(\var\var) \esub{\var}{\la{\vartwo} \Id \Id}$ with~$\Id \defeq \la{\varthree}\varthree$.  

\begin{theorem}[Confluence, \cite{AccattoliPaolini12}]
The reduction $\tovsub$ is confluent.
\end{theorem}

\paragraph*{Plotkin \textit{vs} VSC}

Plotkin's original \cbv $\lambda$-calculus $\plotcalc$ \cite{DBLP:journals/tcs/Plotkin75} can be easily simulated in the \VSC.
The syntax of $\plotcalc$ is simply the same as in the \VSC but without \ES.
The reduction $\tobvplot$ in $\plotcalc$ is the closure under \full contexts (without \ES) of the rule 
\[(\la{\var}\tm)\val \rtobvplot \tm\isub{\var}{\val}\]

\begin{proposition}[Simulation]
	\label{prop:plotkin-vsc}
	\NoteProof{propappendix:plotkin-vsc}
	Let $\tm$ be a term without \ES. 
	If $\tm \tobvplot \tm'$ then $\tm \tom \!\cdot \toe  \tm'$.
\end{proposition}

There is no sensible way to simulate \VSC into $\plotcalc$.
Indeed \VSC is a proper extension of $\plotcalc$: \VSC makes divergent terms such as $(\la{\var}\delta)(\vartwo\vartwo)\delta$ and $\delta ((\la{\var}\delta) (\vartwo\vartwo))$ that are $\betaplot$-normal. Despite being an extension of $\plotcalc$, VSC does not loose the \cbv essence, as $(\la\var\vartwo) \Omega$ is strongly divergent in both $\plotcalc$ and VSC, while in \cbn it normalizes in one step, erasing $\Omega$.

\section{Call-by-Value Solvability and The Solving Strategy}
\label{sect:solving-strat}

In the $\lambda$-calculus, the notion of solvability identifies somehow ``meaningful'' terms.
This notion is well studied in the \cbn $\lambda$-calculus, with an elegant theory, see Barendregt \cite{Barendregt84}.
In \cbv, an elegant theory of solvability is still missing because of issues first observed by Ronchi Della Rocca and Paolini \cite{DBLP:journals/ita/PaoliniR99,parametricBook} (a survey is in \cite{DBLP:journals/corr/Garcia-PerezN16}), amounting to intrinsic limitations of  Plotkin's original \cbv $\lambda$-calculus \cite{DBLP:journals/tcs/Plotkin75}. Accattoli and Paolini showed that instead the \VSC is a good framework for \cbv solvability \cite{AccattoliPaolini12}.

Since the definition of solvability depends on the calculus and its evaluation, we give a parametric definition.

\begin{definition}[Solvability]
	Let $X$ be a calculus.
	A term $\tm$ in $X$ is \emph{$X$-solvable} if there are terms $\tmtwo_1, \dots , \tmtwo_k$ and variables $\var_1, \dots , \var_h$, with $h, k \geq 0$, such that
	$(\lambda \var_1 \dots \la{\var_h}\tm) \tmtwo_1 \dots \tmtwo_k$ $X$-evaluates to $\Id \defeq \la{\var}\var$.
\end{definition}

\paragraph*{The Solving VSC} 
Accattoli and Paolini \cite{AccattoliPaolini12}  characterize operationally \VSC-solvability: a term $\tm$ is \VSC-solvable if and only if $\tm$ $\solvredsym$-normalizes, for a suitable definition of solving reduction $\tosolv$ in between the full one $\tovsub$ and the open one $\tovsubo$, that is, restricting $\tovsub$ but extending $\tovsubo$, given below.
The \emph{Solving \VSC} is the \VSC endowed with the \emph{solving reduction} $\tosolv$ that we now define. 
It is obtained by extending the open rewriting rules under head 
abstractions only, via the notion of \emph{solving~context}.
\begin{center}
	$\begin{aligned}
	\textsc{Solving ctxs} && \solvctx &\grameq \openctx \mid \la{\var}\solvctx \mid \solvctx \tm \mid \solvctx\esub{\var}{\tm}
	\end{aligned}$
	\medskip
	
	\begin{tabular}{ccc}
		\textsc{Solving rewrite rules:}
		&
		\multirow{2}{*}
		{\begin{prooftree}
				\hypo{\tm \Rew{\wsym a} \tm'}		
				\infer1{\solvctxp\tm \Rew{\solvredsym a} \solvctxp{\tm'}}
		\end{prooftree}}
		\\
		($a \in \set{\msym,\esym}$)
	\end{tabular}\medskip
	
	$\begin{array}{cccc}
	\textsc{Solving reduction}:& \tovsubsolv  \, \defeq \, \tomsolv \cup \toesolv
	\end{array}$
	%
\end{center}
For instance, because of the extension under head abstractions,
$\la\var (\Id \Id) \tomsolv \la\var(\varthree\esub\varthree\Id) \toesolv \la\var\Id$.
But reduction under non-head abstractions is forbidden:
$\vartwo (\la\var(\Id \Id)) \not\tomsolv  \vartwo (\la\var(\varthree\esub\varthree\Id))$.

The solving strategy captures the fact that $\tm \defeq \la\var\Omega$ is \cbv unsolvable, as $\tovsubsolv$ diverges on $\tm$, while $\tmtwo \defeq \var (\la\var\Omega)$ is \cbv solvable, and indeed $\tovsubsolv$ terminates on $\tmtwo$. Note also the difference between \cbv and \cbn solvability: a term such as  $\tmthree \defeq \var\Omega$ is \cbv unsolvable (and $\tovsubsolv$ indeed diverges) while it is \cbn solvable (it is head normal). Every \cbv solvable term is also \cbn solvable, as the solving strategy is an extension of the head strategy, because it reduces arguments  both out of abstractions and under head abstractions.

\begin{proposition}[Properties of the solving reduction]
	\label{prop:solvsc-diamond}\label{prop:properties-solvable-reduction}
	\NoteProof{propappendix:properties-solvable-reduction}
	\hfill
	\begin{enumerate}
		\item \label{p:properties-solvable-reduction-diamond} $\tosolv$ is diamond; $\tosolvm$ and $\tosolve$ strongly commute.
		\item \label{p:properties-solvable-reduction-harmony} A term is $\solvredsym$-normal if and only if it is a solvable fireball, where \emph{solvable fireballs} are defined by:
	\end{enumerate}
	\begin{center}
		$\textsc{Solvable fireballs} \qquad \solvnf \grameq  \itm \mid \la{\var}\solvnf \mid \solvnf \esub{\var}{\pitm}$
	\end{center}
\end{proposition}


\section{Multi Types by Value}
\label{sect:types}

We present a  \emph{multi type system} for \cbv.
For Plotkin's \cbv $\l$-calculus, it has been introduced by Ehrhard \cite{DBLP:conf/csl/Ehrhard12}, as the \cbv 
version of de Carvalho's System $\mathsf{R}$ for \cbn \cite{Carvalho07,deCarvalho18}. Both systems can be seen as the restrictions to the \cbn/\cbv translations of the $\l$-calculus of the relational semantics of linear logic. The \cbv system in particular is studied with respect to \ocbv by Accattoli and Guerrieri in \cite{DBLP:conf/aplas/AccattoliG18}. \smallskip

\paragraph*{Multi Types} There are two layers of types, \emph{linear} and \emph{multi types}, mutually defined by: 
\begin{align*}
	\textsc{Linear types \ } \ltype, \ltypetwo &\grameq \ground \mid \larrow{\mtype}{\mtypetwo} \quad
	&
	\textsc{Multi types \ } \mtype, \mtypetwo &\grameq \mset{\ltype_1, \dots, \ltype_n} \quad n \geq 0
\end{align*}
where $\ground$ is an unspecified ground type and $\mset{\ltype_1, \dots, \ltype_n}$ is our notation for finite 
multisets.
The \emph{empty} multi type $\mset{\,}$ (obtained taking $n = 0$) is also denoted by $\emptymset$. 
A generic (multi or linear) type is denoted by $\type$.
A multi type $\mset{\ltype_1, \dots, \ltype_n}$ has to be intended as a conjunction $\ltype_1 \land \dots \land 
\ltype_n$ of linear types $\ltype_1, \dots, \ltype_n$, for a commutative, associative, non-idempotent conjunction 
$\land$ (morally a tensor $\otimes$), whose neutral element~is~$\emptymset$.

The intuition is that a linear type corresponds to a single use of a term $\tm$, and that $\tm$ is typed with a 
multiset 
$\mtype$ of $n$ linear types if it is going to be used (at most) $n$ times. The  meaning of \emph{using a term} is not 
easy to define precisely. Roughly, it means that if $\tm$ is part of a larger term $\tmtwo$, then (at most) $n$ copies 
of $\tm$ shall end up in evaluation position during the evaluation of $\tmtwo$. More precisely, the $n$ copies shall 
end 
up in evaluation positions where \mbox{they are applied to some terms.}

The derivation rules for the multi types system are in \Cref{fig:cbvtypes} (explanation follows).  The rules are the same as in Ehrhard \cite{DBLP:conf/csl/Ehrhard12}, up to the fact that they are extended to \ES. 

A \emph{multi} (\resp \emph{linear}) \emph{judgment} has the shape $\typctx \vdash \tm \hastype \type$ where $\tm$ is a term, $\type$ is a 
multi (\resp linear) type and $\typctx$ is a \emph{type context}, that is, a total function from variables to multi types 
such that  the set $\domain{\typctx} \defeq \{\var \mid \typctx(\var) \neq \emptymset\}$ is finite. \smallskip 

\begin{figure*}[t!]
	\begin{center}
	\scalebox{0.9}{
	$
			{\begin{prooftree}[label separation = .1em]
					\infer0
					[\scriptsize$\ruleAx$]
					{\var \hastype [\ltype] \vdash \var \hastype \ltype}
			\end{prooftree}}
			\quad
			{\begin{prooftree}[separation=1em, label separation = .1em]
					\hypo{\typctx \vdash \tm \hastype \mset{ \larrow{\mtype\!}{\!\mtypetwo} }}
					\hypo{\typctxtwo \vdash \tmtwo \hastype \mtype}
					\infer2[\scriptsize$\ruleAp$]
					{\typctx \uplus \typctxtwo \vdash \tm\tmtwo \hastype \mtypetwo}
			\end{prooftree}}
			\quad
			{\begin{prooftree}[label separation = .1em]
					\hypo{\tyjp{}{\tm}{\typctx, \var \hastype \mtype}{\mtypetwo}}
					\infer1[\scriptsize$\ruleFun$]
					{\tyjp{}{\la{\var}{\tm}}{\typctx}{\ty{\mtype}{\mtypetwo}}}
			\end{prooftree}}
			\quad
			{\begin{prooftree}[separation=1em, label separation = .1em]
					\hypo{\typctx, \var \hastype \mtype \vdash \tm \hastype \mtypetwo}
					\hypo{\typctxtwo \vdash \tmtwo \hastype \mtype}
					\infer2
					[\scriptsize$\ruleES$]
					{\typctx \uplus \typctxtwo \vdash \tm \esub\var\tmtwo \hastype \mtypetwo}
			\end{prooftree}}
			\quad
			{\begin{prooftree}[separation=1em, label separation = .1em]
				\hypo{\left[\tyjp{}{\val}{\typctx_{i}}{\ltype_{i}}\right]_{\iI}}
				\infer1
				[\scriptsize$\ruleMany$]
				{\tyjp{}{\val}{\biguplus_{\iI}\typctx_{i} }{\mult{\ltype_{i}}_{\iI}}}
				\end{prooftree}}
$
	}
\end{center}
	\caption{Call-by-Value Multi Type System.}
	\label{fig:cbvtypes}
\end{figure*}	
\paragraph*{Technicalities about Types} The type context $\typctx$ is \emph{empty} if $\dom{\typctx} = \emptyset$.  
\emph{Multi-set sum} $\mplus$ is extended to type contexts point-wise,
\ie\  $(\typctx \mplus \typctxtwo)(\var) \defeq \typctx(\var) \mplus \typctxtwo(\var)$ for each variable $\var$.
This notion is extended to a finite family of type contexts as expected, 
in particular $\bigmplus_{i \in J\!} \typctx_i$ is the empty context  when $J = \emptyset$.
A type context $\typctx$ is denoted by $\var_1 \hastype \mtype_1, \dots, \var_n \hastype \mtype_n$ (for some $n \in 
\nat$) if $\dom{\typctx} \subseteq \{\var_1, \dots, \var_n\}$ and $\typctx(\var_i) = \mtype_{i}$ for all $1 \leq i \leq 
n$.
Given two type contexts $\typctx$ and $\typctxtwo$ such that $\dom{\typctx} \cap \dom{\typctxtwo} = \emptyset$, the 
type 
context $\typctx, \typctxtwo$ is defined by $(\typctx, \typctxtwo)(\var) \defeq \typctx(\var)$ if $\var \in 
\dom{\typctx}$, $(\typctx, \typctxtwo)(\var) \defeq \typctxtwo(\var)$ if $\var \in \dom{\typctxtwo}$, and $(\typctx, 
\typctxtwo)(\var) \defeq \emptymset$ otherwise.
Note that $\typctx, \var \hastype \emptymset = \typctx$, where we implicitly assume $\var \notin \dom{\typctx}$. 

We write $\concl{\tderiv}{\typctx}{\tm}{\mtype}$ if $\tderiv$ is a (\emph{type}) \emph{derivation} (\ie a tree built up from the rules in \Cref{fig:cbvtypes}) with conclusion the multi judgment $\typctx \vdash \tm \hastype \mtype$.
In particular,  we write $\concl{\tderiv}{\,}{\tm}{\mtype}$ when $\typctx$ is empty.
We write $\derive{\tderiv}{\tm}$ if $\concl{\tderiv}{\typctx}{\tm}{\mtype}$ for some type context $\typctx$ and multi type $\mtype$.  \smallskip

\paragraph*{Explanations About the Rules of the Type System}
All rules but $\ruleAx$ and $\ruleFun$ assign a multi type to the term on the
right-hand side of a judgment.
Variables and abstractions are the only terms that can be typed by a linear type, via $\ruleAx$ and $\ruleFun$.
Rule $\ruleMany$ can be applied to variables and abstractions only, turning linear types into multi types: they have as many premises as the elements in the (possibly
empty) set of indices $I$ (when $I = \emptyset$, the rule has no premises, and it gives an empty multi type $\emptymset$). 
The $\ruleMany$ rules say how many ``copies'' of one occurrence of abstraction or variable in a term $\tm$ are needed to evaluate~$\tm$.
Essentially, they correspond to the promotion rule of linear logic,
which, in the \cbv representation of the $\lambda$-calculus, is indeed used for typing
abstractions and variables.\smallskip

\paragraph*{The Sizes of Type Derivations} Our study being quantitative, we need a notion of size of type derivations. In fact, we shall use \emph{two} notions of size.

\begin{definition}[Derivation size(s)]
\label{def:two-sizes}
	Let $\tderiv$ be a derivation. 
	The \emph{(general) size} $\size{\tderiv}$ of $\tderiv$ is the number of 
	rule occurrences in $\tderiv$ except for the rule $\ruleMany$.
	The \emph{multiplicative size} $\sizem{\tderiv}$ of $\tderiv$ is the number of occurrences of the rules $\lambda$ and 
$\ruleAp$ in $\tderiv$.
\end{definition}
The two sizes for a derivation play different qualitative and quantitative roles. \emph{Qualitative}: to have a 
combinatorial proof of the characterization of solvable terms, we need a measure that 
decreases for some kinds of $\tovsub$ steps; 
this role is played by the general size $\size{\!\cdot\!}$.
	\emph{Quantitative}: to count the number of $\tom$ steps in 
	solving evaluations, \ie the cost model;
	his role is played by the multiplicative size $\sizem{\!\cdot\!}$. \smallskip

\paragraph*{Substitution and Removal Lemmas}
The two next lemmas establish a key feature of this type system: in a typed term $\tm$,
substituting a value for
a variable as in the exponential step, or dually removing a value, preserves the type of $\tm$ and \emph{consumes} (dually, \emph{adds}) the multi type of the variable.
Besides, it also provides  quantitative information about the type derivation for $\tm$ before and after the substitution/removal.

\begin{lemma}[Substitution]
	\label{l:substitution}	
	\NoteProof{lappendix:substitution}
	Let $\tm$ be a term, $\val$ be a value and $\namedtyjp{\tderiv}{}{\tm}{\typctx, \var \hastype 
		\mtypetwo}{\mtype}$ and $\namedtyjp{\tderivtwo}{}{\val}{\typctxtwo}{\mtypetwo}$ be derivations.
	Then there is a derivation $\namedtyjp{\tderivthree}{}{\tm \isub{\var}{\val}}{\typctx \mplus \typctxtwo}{\mtype}$ 
	with $\sizem{\tderivthree} = \sizem{\tderiv} + \sizem{\tderivtwo}$ and $\size{\tderivthree} \leq \size{\tderiv} + 
	\size{\tderivtwo}$. 
\end{lemma}

\begin{lemma}[Removal]
	\label{l:anti-substitution}
	\NoteProof{lappendix:anti-substitution}
	Let $\tm$ be a term, $\val$ be a value, and 
	$\namedtyjp{\tderiv}{}{\tm\isub{\var}{\val}}{\typctx}{\mtype}$
	be a derivation. 
	Then there are two  derivations $\namedtyjp{\tderivtwo}{}{\tm}{\typctxtwo, \var \hastype 
		\mtypetwo}{\mtype}$ 
	and $\namedtyjp{\tderivthree}{}{\val}{\typctxthree}{\mtypetwo}$ such that $\typctx = \typctxtwo \mplus \typctxthree$ 
	with $\sizem{\tderiv} = \sizem{\tderivtwo} + \sizem{\tderivthree}$ and $\size{\tderiv} \leq \size{\tderivtwo} + 
	\size{\tderivthree}$.
\end{lemma}

\paragraph*{The Special Role of Inert Terms} In the characterizations via multi types of the following two sections, inert terms play a crucial role. In statements about solvable normal forms, they usually satisfy stronger properties, essential for the induction to go through. Predicates shall also \emph{spread} on inert terms: if assumed on the type context, they transfer to the right-hand type, which is in turn the key step to propagate the predicate on sub-derivations. 

\section{Multi Types for Open \cbv}
\label{sect:open}

Here we recall the relationship between \cbv multi types and \ocbv developed by Accattoli and Guerrieri in 
\cite{DBLP:conf/aplas/AccattoliG18}. The reason is threefold: 
\begin{enumerate}
\item \emph{Building block}: the solvable case relies on the open one, because the solving strategy is an iteration of open evaluation under head abstractions. 
\item \emph{Blueprint}: the open case provides the blueprint for the solvable case. 
\item \emph{Adapting a few details}: the development in \cite{DBLP:conf/aplas/AccattoliG18} needs to be slightly adapted to our present framework. 
Namely, here we use the Open VSC 
instead of the \emph{split fireball calculus} used in \cite{DBLP:conf/aplas/AccattoliG18} (another formalism for \ocbv),
and we include a \emph{ground type} $\ground$, necessary to later deal with the solving strategy.
\end{enumerate}

\paragraph*{The Open Size of Terms} For our quantitative study, we need a notion of term size, introduced here. Now, we actually need a notion of size for each evaluation strategy that we aim at measuring via multi types. Essentially, the size counts the constructors of a term that can be traversed by the strategy. The \emph{open size} $\sizeo{\tm}$ of a term $\tm$, then, is its number of applications out of abstractions, \ie 
	\begin{align*}
	\sizeo{\var} &\defeq 0 
	&
	\sizeo{\la{\var}{\tm}} & \defeq  0
	& 
	\sizeo{\tm\tmtwo} & \defeq  \sizeo{\tm} + \sizeo{\tmtwo} + 1 
	&
	\sizeo{\tm \esub{\var}{\tmtwo}} & \defeq  \sizeo{\tm} + \sizeo{\tmtwo}.
	\end{align*}

\paragraph*{Overview of the Characterization}
Qualitatively, the open evaluation of $\tm$ terminates if and only if $\tm$ is typable.
Since $\tovsubo$ does not reduce under abstractions, every abstraction is $\osym$-normal and hence must be typable: for this reason, $\la{\var}\delta\delta$ is typable with $\emptytype$ (take the derivation only made of one rule $\ruleManyVal$ with $0$ premises), though $\delta\delta$ is not.

Quantitatively, the multiplicative size 
$\sizem{\tderiv}$ of every type derivation $\tderiv$ for $\tm$ provides bounds the sum of the length of the open evaluation of $\tm$ plus 
the open size of its open normal form. To obtain exact bounds, then, one has to consider only type derivations satisfying a further \emph{tight} predicate, defined below over the auxiliary \emph{inert} predicate for type derivations. 

\paragraph*{Inert and Tight Derivations} 
A multi type is \emph{ground} if it is of the form $n\mset{\ground} \defeq \mset{ \ground, \dots, \ground}$  ($n$  times $\ground$) for any $n \geq 0$ (so, $0\mset{\ground} = \emptytype$).
\emph{Inert types} are defined below, with $n \geq 0$.
\begin{align*}
	\textsc{Inert multi type }& \imtype \grameq 	\mset{ \inltype_1, \dots, \inltype_n}
	&
	\textsc{Inert linear type }& \inltype \grameq \ground \mid \larrow{n\mset{\ground}}{\imtype} 
\end{align*}
A type context $\var_1 \hastype \mtype_1, \dots, \var_n \hastype \mtype_n$ is \emph{inert} if $\mtype_1, \dots, 
\mtype_n$ are inert multi types. 
Note that 
every ground multi type is inert.

A derivation $\namedtyjp{\tderiv}{}{\tm}{\typctx}{\mtype}$ is \emph{inert} if $\typctx$ is an inert type context and $\mtype$ is an inert multi type. If, moreover, $\mtype$ is a ground multi type, then $\tderiv$ is \emph{tight}.
Note that the definitions of inert and tight derivation depend only on its final judgment.

Tight
derivations are those we are actually interested in, but often, for the induction to go through, we have to consider the wider class of inert derivations.
For instance, tight derivations may have a complex
structure, having sub-derivations for inert terms that might
not be tight, but only inert.
Additionally, inert types spread on inert terms---the first key property of inert terms.

\begin{lemma}
	[Spreading of inertness on judgments]
	\label{l:spread-inert}
	\NoteProof{lappendix:spread-inert}
	Let $\concl{\tderiv}{\typctx}{\itm}{\mtype}$ be a derivation and $\itm$ be an inert term. 
	If $\typctx$ is a inert type context, then $\mtype$ is a inert multi type (and so $\tderiv$ is inert).
\end{lemma}

\paragraph*{Correctness}
Open correctness establishes that all typable terms $\osym$-normalize and the multiplicative size of the derivation bounds the number of $\tomo$ steps plus the open size of the $\osym$-normal form; this bound is exact if the derivation is tight. 
Open correctness is proved following a standard scheme in two stages: $(a)$ \emph{quantitative} subject reduction states that every $\tovsubo$ step preserves types and decreases the general size of a derivation, and that any $\tomo$ step decreases by an exact quantity the multiplicative size of a derivation;
$(b)$ a lemma states that the multiplicative size of any derivation typing a $\osym$-normal form $\tm$ provides an \emph{upper bound} to the open size of $\tm$, and if moreover the derivation is tight then the bound~is~\emph{exact}. The lemma follows---its unusual statement puts forward the special role of inert terms.
\newcounter{l:size-fireballs}
\addtocounter{l:size-fireballs}{\value{theorem}}
\begin{lemma}[Size of fireballs]
	\label{l:size-fireballs}
	\NoteProof{lappendix:size-fireballs}
	Let $\fire$ be a fireball. 
	If $\namedtyjp{\tderiv}{}{\fire}{\typctx}{\mtype}$
	then $\sizem{\tderiv} \geq \sizeo{\fire}$.
	If, moreover, $\typctx$ is inert and ($\mtype$ is ground inert or $\fire$ is inert), then 
$\sizem{\tderiv} = \sizeo{\fire}$.
\end{lemma}

\emph{Example}. The non-inert fireball $\delta \defeq \la{\var}{\var\var}$ is typable and the last rule of any derivation $\concl{\tderiv}{\typctx}{\delta}{\mtype}$ is 
$\ruleManyVal$.
If $\mtype$ is not ground then $\ruleManyVal$ has at least one premise that types the subterm $\var\var$, so $\sizem{\tderiv} > \sizeo{\delta}$.
If $\mtype$ is ground, then $\mtype = \emptytype$ and $\ruleManyVal$ has no premises, with $\sizem{\tderiv} = 0 = \sizeo{\delta}$.
For inert terms, \Cref{l:size-fireballs} says that their open size is the multiplicative size of their \emph{inert} derivations; 
by spreading of inertness (\Cref{l:spread-inert}), it amounts to say that~the~type~context~is~inert.

\begin{proposition}[Open quantitative subject reduction]
	\label{prop:weak-subject-reduction}
	\NoteProof{propappendix:weak-subject-reduction}
	Let $\namedtyjp{\tderiv}{}{\tm}{\typctx}{\mtype}$ be a derivation.
	\begin{enumerate}
		\item\emph{Multiplicative step:} if $\tm \towm \tm'$ then there is a derivation 
$\namedtyjp{\tderiv'}{}{\tm'}{\typctx}{\mtype}$ with
		$\sizem{\tderiv'} = \sizem{\tderiv} - 2$ and $\size{\tderiv'} = \size{\tderiv} - 1$; 
		\item\emph{Exponential step:} if $\tm \towe \tm'$ then there is a derivation 
$\namedtyjp{\tderiv'}{}{\tm'}{\typctx}{\mtype}$ such that
		$\sizem{\tderiv'} = \sizem{\tderiv}$ and $\size{\tderiv'} < \size{\tderiv}$.
	\end{enumerate}
\end{proposition}

As the general size of derivations decreases after any $\tovsubo$ step, 
we have a combinatorial proof of open correctness.

\begin{theorem}[Open correctness]
	\label{thm:open-correctness}
	\NoteProof{thmappendix:open-correctness}
	Let 
	$\derive{\tderiv}{\tm}$ be a derivation.
	Then there is a $\osym$-normalizing evaluation $\deriv \colon \tm \tovsubo^* \tmtwo$ with $2\sizem{\deriv} + 
	\sizeo{\tmtwo} \leq \sizem{\tderiv}$.
	And if $\tderiv$ is tight, then $2\sizem{\deriv} + \sizeo{\tmtwo} = \sizem{\tderiv}$.
\end{theorem}

\paragraph*{Completeness}
Open completeness establishes that every $\osym$-normalizing term is typable, and with a tight derivation $\tderiv$ such that $\sizem\tderiv$ is exactly the number of $\tomo$ steps plus the open size of the $\osym$-normal form. 
The proof technique is standard: $(a)$ \emph{quantitative} subject expansion states that typability can be pulled back along $\tovsubo$ steps, increasing $\size\tderiv$;
$(b)$ a lemma states that every $\osym$-normal form is typable with a \emph{tight} derivation---inert terms verify a stronger statement. 

\newcounter{prop:precise-open-typability-nf}
\addtocounter{prop:precise-open-typability-nf}{\value{theorem}}
\begin{lemma}[Tight typability of open normal forms]
	\label{prop:precise-open-typability-nf} 
	\NoteProof{propappendix:precise-open-typability-nf}
	\begin{enumerate}
		\item \emph{Inert:}\label{p:precise-open-typability-nf-inert} if $\tm$ is an inert term then, for any inert multi 
type $\mtype$, there is an inert derivation $\concl{\tderiv}{\typctx}{\tm}{\mtype}$.
		\item \emph{Fireball:}\label{p:precise-open-typability-nf-fireball} if $\tm$ is a fireball then there is a tight derivation $\concl{\tderiv}{\typctx}{\tm}{\emptytype}$.		
	\end{enumerate}
\end{lemma}

\begin{proposition}[Open quantitative subject expansion]
	\label{prop:weak-subject-expansion}
	\NoteProof{propappendix:weak-subject-expansion}
	Let $\namedtyjp{\tderiv'}{}{\tm'}{\typctx}{\mtype}$ be a derivation.
	\begin{enumerate}
		\item\emph{Multiplicative step:} if $\tm \towm \tm'$ then there is a derivation 
$\namedtyjp{\tderiv}{}{\tm}{\typctx}{\mtype}$ with
		$\sizem{\tderiv'} = \sizem{\tderiv} - 2$ and $\size{\tderiv'} = \size{\tderiv} - 1$; 
		\item\emph{Exponential step:} if $\tm \towe \tm'$ then there is a derivation 
$\namedtyjp{\tderiv}{}{\tm}{\typctx}{\mtype}$ such that
		$\sizem{\tderiv'} = \sizem{\tderiv}$ and $\size{\tderiv'} < \size{\tderiv}$.
	\end{enumerate}
\end{proposition}

As the general size of derivations increases 
after a backward $\osym$-step,
a combinatorial proof of open completeness~follows.

\begin{theorem}[Open completeness]
	\label{thm:open-completeness}
	\NoteProof{thmappendix:open-completeness}
	Let $\deriv \colon \tm \tovsubo^* \tmtwo$ be an $\osym$-normalizing evaluation. 
	Then there is a tight derivation $\concl{\tderiv}{\typctx}{\tm}{\emptytype}$ such that $2\sizem{\deriv} + \sizes{\fire} = \sizem{\tderiv}$.
\end{theorem}

%
\section{Multi Types for \cbv Solvability}
\label{sect:solvable}
Here we provide the main results of the paper.  \smallskip

\paragraph*{Solvable size} First of all, we need a notion of size for normal forms of the solving strategy. The \emph{solvable size} $\sizes{\tm}$ of a term $\tm$ is its number of applications plus its number of abstractions 
not in argument position: 
\begin{align*}
\sizes{\var} &\defeq 0 
&
\sizes{\la{\var}{\tm}} & \defeq  \sizes{\tm} + 1 
& 
\sizes{\tm\tmtwo} & \defeq  \sizes{\tm} + \sizeo{\tmtwo} + 1 
&
\sizes{\tm \esub{\var}{\tmtwo}} & \defeq  \sizes{\tm} + \sizeo{\tmtwo}.
\end{align*}

\begin{figure*}[!t]
\begin{center}
	\scalebox{.9}{
$\begin{array}{rr@{\ }c@{\ }l@{\qquad} rr@{\ }c@{\ }l}
 \text{Solvable multi type } & \smtype &\grameq &	\mset{\sltype_1, \dots, \sltype_n} \ \ n > 0
 &
 \text{Solvable linear type} & \sltype &\grameq & \ground \mid  \larrow{\mtype}{\smtype}
 \\
 \text{Unitary s. multi type} & \usmtype &\grameq &	\mset{ \usltype}
 &
 \text{Unitary s. linear type} &\usltype &\grameq & \ground \mid \larrow{\mtype}{\usmtype}
 \\
 \text{Inertly s. multi type} & \ismtype &\grameq & \mset{\isltype_1, \dots, \isltype_n} \ \ n > 0
 &
 \text{Inertly s. linear type} & \isltype &\grameq & \ground \mid \larrow{\imtype}{\ismtype}
 \end{array}$
}
 \end{center}
 \caption{Kinds of solvable types. A 
 	type is \emph{precisely solvable} if it is unitary and inertly solvable.}
 \label{fig:solvable-types}
\end{figure*}

\paragraph*{Solvable Multi Types}
The qualitative characterization of solvable terms with multi types is fairly simple: they are those terms typable with a \emph{solvable multi type}, defined in \Cref{fig:solvable-types}. 
There are two ingredients: descending under abstraction, and selecting only head abstractions. They corresponds on types to, respectively, being typable with something that is not $\emptytype$, and having non-$\emptytype$ type on the right of the linear arrow---recursively.

 Being typable with something that is not $\emptytype$, ultimately requires 
a ground multi type $n\mset{\ground}$ different from $\emptytype$ in the type system (in contrast to the open case, where there is no need for $\ground$). 
The point is exemplified by the term $\la{\vartwo}{\delta\delta}$ (with $\delta = \la{\var}{\var\var}$): it is $\osym$-normal and it can only be typed by $\emptytype$, but it is not $\solvredsym$-normalizing and it cannot be typed by solvable types. \smallskip

\paragraph*{Precisely Solvable Multi Types} The quantitative properties of solvable terms rest on two orthogonal predicates (see \Cref{fig:solvable-types}). 

The \emph{unitary} one ensures that each solving multiplicative step is counted \emph{exactly} once. Solvable types guarantee that each such step is counted, but it might be counted more than once. The constraint amounts to 
asking that the topmost and right-hand multisets are singletons. This is the key requirement for obtaining that in the statement of subject reduction the general size of the derivation decreases by exactly one at each multiplicative~step. 

The \emph{inert} predicate instead ensures that the type derivation does not type sub-terms that are not accessible to the solving strategy. Typing such sub-terms is harmless for a qualitative study, but it is problematic for a quantitative one, as it does introduces a mismatch between the size of solving normal forms and the size of their type derivations. The constraint amounts to asking that the left-hand multisets are inert. 

Solvable types that are both unitary and inert are called \emph{precise}, and provide exact~bounds. \smallskip

\paragraph*{Correctness}
Solvable correctness claims that all terms typable with a solvable type $\mtype$ are $\solvredsym$-normalizing, and the multiplicative size of the derivation bounds the number of $\tosolvm$ steps plus the solvable size of the $\solvredsym$-normal form; this bound is exact if the type context is inert is $\mtype$ is precisely~solvable.

\newcounter{l:size-solvable-nf}
\addtocounter{l:size-solvable-nf}{\value{theorem}}
\begin{lemma}[Size of solvable fireballs]
	\label{l:size-solvable-nf}
	\NoteProof{lappendix:size-solvable-nf}
	Let $\solvnf$ be a solvable fireball.
	If $\namedtyjp{\tderiv}{}{\solvnf}{\typctx}{\mtype}$ with $\mtype$ solvable (\resp $\typctx$  inert and $\mtype$ precisely solvable),
	then $\sizem{\tderiv} \geq \sizes{\solvnf}$ (\resp $\sizem{\tderiv} = \sizes{\solvnf}$).
\end{lemma}

Let us see with an example why in \Cref{l:size-solvable-nf} we need the restrictions on the type and type contexts to have (exact) bounds. 
Consider the solvable fireball (but non-inert) $\delta = \la{\var}{\var\var}$, which is typable and the last rule of any derivation $\concl{\tderiv}{\typctx}{\delta}{\mtype}$ is $\ruleManyVal$.
The solvable size of $\delta$ is $\sizes{\delta} = \sizes{\var\var} + 1$, so in order that $\sizem{\tderiv} \geq \sizes{\delta}$, $\ruleManyVal$ needs to have at least one premise (otherwise $\sizem{\tderiv} = 0 < \sizes{\delta}$), 
which is exactly what the solvability of multi type $\mtype$ guarantees.
If moreover we want $\sizem{\tderiv} = \sizes{\delta}$, 
any premise of $\ruleMany$ typing $\var\var$ must be an \emph{inert} derivation $\concl{\tderivtwo}{\var \hastype \mtypetwo}{\var\var}{\mtypethree}$, thus $\sizem{\tderivtwo} =  \sizes{\var\var}$, and $\ruleManyVal$ must have only one premise. 
Summing up, $\mtype= \mset{\larrow{\mtypetwo}{\mtypethree}}$, which is both unitary and inertly (as $\mtypetwo$ is inert) solvable, \ie precisely~solvable.

\begin{proposition}[Solving quantitative subject reduction]
	\label{prop:solvable-subject-reduction}
	\NoteProof{propappendix:solvable-subject-reduction}
	Assume $\concl{\tderiv}{\typctx}{\tm}{\mtype}$, with $\mtype$ solvable (\resp unitary solvable).
	\begin{enumerate}
		\item \emph{Multiplicative step:} if $\tm \tosolvm \tm'$ then there is a derivation 
$\concl{\tderiv'}{\typctx}{\tm'}{\mtype}$ such that $\sizem{\tderiv'} \leq \sizem{\tderiv}-2$ and $\size{\tderiv'} < 
\size{\tderiv}$
		(\resp $\sizem{\tderiv'} = \sizem{\tderiv}-2$ and $\size{\tderiv'} = \size{\tderiv}-1$);
		
		\item \emph{Exponential step:} if $\tm \tosolve \tm'$ then there is a derivation 
$\concl{\tderiv'}{\typctx}{\tm'}{\mtype}$ such that
		$\sizem{\tderiv'} = \sizem{\tderiv}$ and $\size{\tderiv'} < \size{\tderiv}$.
	\end{enumerate}
\end{proposition}

Quantitative solving subject reduction (\Cref{prop:solvable-subject-reduction})---as well as expansion (\Cref{prop:solvable-subject-expansion} below)---holds only for a restricted set of multi types, the solvable ones. 
The reason is evident if we consider the term $\la{\vartwo}{\delta\delta}$: it is 
and typable only through a derivation $\concl{\tderiv}{\,}{\la{\vartwo}\delta\delta}{\emptytype}$ with $\sizem{\tderiv} = 0 = \size{\tderiv}$, but $\la{\vartwo}{\delta\delta} \tosolvm \la{\vartwo}\var\esub{\var}{\delta}$ and $\emptytype$ is not a solvable multi type.

\begin{theorem}[Solving correctness]
	\label{thm:solvable-correctness}
	\NoteProof{thmappendix:solvable-correctness}
	Let $\concl{\tderiv}{\typctx}{\tm}{\mtype}$ be a derivation with $\mtype$ solvable (\resp $\typctx$ inert and $\mtype$ precisely solvable).
	Then, 
	there is an $\solvredsym$-normalizing evaluation $\deriv \colon \tm \tosolv^* \tmtwo$ 
	with $2\sizem{\deriv} + \sizes{\tmtwo} \leq \sizem{\tderiv}$ (\resp $2\sizem{\deriv} + \sizes{\tmtwo} = \sizem{\tderiv}$).
\end{theorem}

\paragraph*{Completeness}
For the typability of solvable normal forms, the ground type $\ground$ plays a crucial role, since $\mset{\ground}$ is both a precisely solvable and an inert multi type, and hence we can apply \Cref{prop:precise-open-typability-nf} when $\solvnf$ is an inert term.

\newcounter{prop:solvable-typability-nf}
\addtocounter{prop:solvable-typability-nf}{\value{theorem}}
\begin{lemma}[Precisely solvable typability of solvable fireballs]
	\label{prop:precise-solvable-typability-nf}
	\NoteProof{propappendix:precise-solvable-typability-nf}
		If $\tm$ is a solvable fireball, then there is a derivation 
$\concl{\tderiv}{\typctx}{\tm}{\mtype}$ with $\typctx$ inert and $\mtype$ precisely solvable.
\end{lemma}

\newcounter{prop:solvable-subject-expansion}
\addtocounter{prop:solvable-subject-expansion}{\value{theorem}}
\begin{proposition}[Solving quantitative subject expansion]
	\label{prop:solvable-subject-expansion}
	\NoteProof{propappendix:solvable-subject-expansion}
	Assume $\concl{\tderiv'\!}{\typctx}{\tm'\!}{\mtypetwo}$ with $\mtypetwo$ solvable (\resp unitary~solvable).
	\begin{enumerate}
		\item \emph{Multiplicative step:} if $\tm \tosolvm \tm'$ then there is a derivation 
$\concl{\tderiv}{\typctx}{\tm}{\mtypetwo}$ with
		$\sizem{\tderiv'} \leq \sizem{\tderiv}-2$ and \mbox{$\size{\tderiv'} < \size{\tderiv}$}
		(\resp $\sizem{\tderiv'} = \sizem{\tderiv}-2$ and $\size{\tderiv'} = \size{\tderiv}-1$);
		\item \emph{Exponential step:} if $\tm \tosolve \tm'$ then there is a derivation 
$\concl{\tderiv}{\typctx}{\tm}{\mtypetwo}$ such that
		$\sizem{\tderiv'} = \sizem{\tderiv}$ and $\size{\tderiv'} < \size{\tderiv}$.
	\end{enumerate}
\end{proposition}

\begin{theorem}[Solving completeness]
	\label{thm:solvable-completeness}
	\NoteProof{thmappendix:solvable-completeness}
	Let $\deriv \colon \tm \tosolv^* \tmtwo$ be an $\solvsym$-normalizing evaluation. 
	Then there is a derivation $\concl{\tderiv}{\typctx}{\tm}{\mtypetwo}$ with $\typctx$ inert, $\mtypetwo$ precisely solvable 	and \mbox{$2\sizem{\deriv} + \sizes{\tmtwo} = \sizem{\tderiv}$}.
\end{theorem}

\IEEEpeerreviewmaketitle


\bibliographystyle{IEEEtranS}
\bibliography{main.bbl}

\clearpage
\appendices

\section{Counterexamples}
\label{sect:counter}

\subsection{Counterexample to subject reduction and expansion in Paolini and Ronchi Della Rocca \cite{DBLP:journals/ita/PaoliniR99}}

In \cite{DBLP:journals/ita/PaoliniR99}, the idempotent intersection type system introduced to characterize \cbv solvability is defined as follows.

\emph{Types} and \emph{intersection types} are defined by mutual induction according to the grammar below, where $\alpha$ and $\nu$ are two distinct constants, and $\{\sigma_1, \dots, \sigma_n\}$ is a non-empty finite set of types:
\begin{align*}
\text{types} \qquad \sigma, \tau &\Coloneqq \alpha \mid \nu \mid S \Rightarrow \tau 
& &&
\text{intersection types} \qquad S &\Coloneqq \{\sigma_1, \dots, \sigma_n\}  \qquad (n \geq 1)
\end{align*}

An \emph{environment} $B$ is a (total) function mapping variables to finite sets of types such that $\dom{B} = \{\var \mid B(\var) \neq \emptyset\}$ is finite.
We write $B = \var_1 : S_1, \dots, \var_n : S_n$ if $\dom{B} = \{\var_1, \dots, \var_n\}$ and $\var_1, \dots, \var_n$ are pairwise disjoint. 
Given two environments $B$ and $B'$, we write $B \cup B'$ for their pointwise union, \ie, $(B \cup B')(x) = B(x) \cup B'(x)$ for every variable $x$.

The inference rules of the type system are the following (see \cite[Definition 6.2]{DBLP:journals/ita/PaoliniR99}):\footnote{In \cite[Definition 6.2]{DBLP:journals/ita/PaoliniR99}, the rule $\Rightarrow_{\nu E}$ is not included, but it is needed otherwise the \cbv solvble term $(\la{\varthree}\var)\la{\vartwo}\Omega$ would not be typable.}
\begin{gather*}
\begin{prooftree}
	\infer0[$\text{var}$]{\var: \{\sigma\} \vdash \var : \sigma}
\end{prooftree}
\qquad
\begin{prooftree}
	\hypo{B \vdash \tm : \{\sigma_1, \dots, \sigma_n\} \Rightarrow \tau}
	\hypo{(B_i \vdash \tmtwo : \sigma_i)_{1 \leq i \leq n}}
	\hypo{n \geq 1}
	\infer3[$\Rightarrow_E$]{B \cup \bigcup_{i=1}^n B_i \vdash  \tm\tmtwo : \tau}
\end{prooftree}
\qquad
\begin{prooftree}
	\hypo{B \vdash \tm : \{\nu\} \Rightarrow \tau}
	\hypo{B' \vdash \tmtwo : \nu}
	\infer2[$\Rightarrow_{\nu E}$]{B \cup B' \vdash  \tm\tmtwo : \tau}
\end{prooftree}
\\
\begin{prooftree}
\infer0[$\nu$]{\vdash \la{\var}{\tm} : \nu}
\end{prooftree}
\qquad
\begin{prooftree}
\hypo{B \vdash \tm : \tau}
\hypo{\var \notin \dom{B}}
\infer2[$\Rightarrow_{\nu I}$]{B \vdash \la{\var}{\tm} : \{\nu\} \Rightarrow \tau}
\end{prooftree}
\qquad
\begin{prooftree}
\hypo{B \vdash \tm : \tau}
\hypo{\var \notin \dom{B}}
\infer2[$\Rightarrow_{0 I}$]{B \vdash \la{\var}{\tm} : \tau}
\end{prooftree}	
\qquad
\begin{prooftree}
\hypo{B, x : S \vdash \tm : \tau}
\infer1[$\Rightarrow_{I}$]{B \vdash \la{\var}{\tm} : S \Rightarrow \tau}
\end{prooftree}
\end{gather*}

Let $\tm \defeq \la{\var}(\la{\varthree}{\var})({\var\var)}$ and $\tm' \defeq \la{\var}\var$.
According to the reduction defined in \cite{DBLP:journals/ita/PaoliniR99} to characterize operationally \cbv solvability, we have $\tm \to \tm'$.

The only possible judgments for $\tm$ are $\vdash \tm : \nu$ and $\vdash \tm : \{\{\sigma\} \Rightarrow \nu, \sigma, \tau\} \Rightarrow \tau$, for every types $\sigma, \tau$.
Indeed, 
\begin{align*}
\begin{prooftree}
\infer0[$\nu$]{\vdash \la{\var}(\la{\varthree}\var)(\var\var) : \nu}
\end{prooftree}
&&
\begin{prooftree}
	\infer0[$\text{var}$]{x : \{\tau\} \vdash \var : \tau}
	\infer1[$\Rightarrow_{\nu I}$]{x : \{\tau\} \vdash \la{\varthree}\var : \{\nu\} \Rightarrow \tau}
	\infer0[$\text{var}$]{x : \{\{\sigma\} \Rightarrow \nu\} \vdash x : \{\sigma\} \Rightarrow \nu}
	\infer0[$\text{var}$]{x : \{\sigma\} \vdash x : \sigma}
	\infer2[$\Rightarrow_E$]{x : \{\{\sigma\} \Rightarrow \nu, \sigma\} \vdash \var\var : \nu}
	\infer2[$\Rightarrow_{\nu E}$]{x : \{\{\sigma\} \Rightarrow \nu, \sigma, \tau\} \vdash (\la{\varthree}\var)(\var\var) : \tau}
	\infer1[$\Rightarrow_I$]{\vdash \la{\var}(\la{\varthree}\var)(\var\var) : \{\{\sigma\} \Rightarrow \nu, \sigma, \tau\} \Rightarrow \tau}
\end{prooftree}
\end{align*}
and no other judgments are derivable for $\tm$.

The only possible judgments for $\tm'$ are $\vdash \tm' : \nu$ and $\vdash \tm' : \{\tau\} \Rightarrow \tau$, for every type $\tau$.
Indeed, 
\begin{align*}
\begin{prooftree}
\infer0[$\nu$]{\vdash \la{\var}(\la{\varthree}\var)(\var\var) : \nu}
\end{prooftree}
&&
\begin{prooftree}
\infer0[$\text{var}$]{x : \{\tau\} \vdash \var : \tau}
\infer1[$\Rightarrow_{\nu I}$]{\vdash \la{\var}\var : \{\tau\} \Rightarrow \tau}
\end{prooftree}
\end{align*}
and no other judgments are derivable for $\tm'$.

Summing up, subject reduction (as it is stated in \cite[Lemma 6.4]{DBLP:journals/ita/PaoliniR99}) does not hold because $\vdash \tm : \{\{\sigma\} \Rightarrow \nu, \sigma, \tau\} \Rightarrow \tau$ but there is no environment $B$ such that $B \vdash \tm' : \{\{\sigma\} \Rightarrow \nu, \sigma, \tau\} \Rightarrow \tau$.
Subject expansion (as it is stated in \cite[Lemma 6.6]{DBLP:journals/ita/PaoliniR99}) does not hold because $\vdash \tm' : \{\tau\} \Rightarrow \tau$ for every type $\tau$, but there is no environment $B$ and no type $\tau$ such that $\vdash \tm : \{\tau\} \Rightarrow \tau$.

\subsection{Counterexample to subject reduction in Kerinec et al. \cite{DBLP:conf/fscd/KerinecMR21} (due to Delia Kesner)}

In \cite{DBLP:conf/fscd/KerinecMR21}, the non-idempotent intersection type system introduced to characterize \cbv solvability is defined as follows.

Given a countable set of constants $a, b, c, \dots$, \emph{types} and \emph{multiset types} are defined by mutual induction according to the grammar below, where $\mset{\alpha_1, \dots, \alpha_n}$ with $n \geq 0$ is a (possibly empty) finite multiset of types:
\begin{align*}
\text{types} \qquad \alpha, \beta &\Coloneqq a \mid \mset{\,} \mid M \Rightarrow \alpha 
&&& &&
\text{multiset types} \qquad M &\Coloneqq \mset{\alpha_1, \dots, \alpha_n}  \qquad (n \geq 0, \ \alpha_i \neq \mset{\,} \text{ for all } 1 \leq i \leq n)
\end{align*}

An \emph{environment} $\typctx$ is a (total) function mapping variables to multiset types such that $\dom{\typctx} = \{\var \mid \typctx(\var) \neq \mset{\,}\}$ is finite.
We write $\typctx = \var_1 : S_1, \dots, \var_n : S_n$ if $\dom{\typctx} \subseteq \{\var_1, \dots, \var_n\}$ and $\var_1, \dots, \var_n$ are pairwise disjoint. 
Given two environments $\typctx$ and $\typctx'$, we write $\typctx + \typctx'$ for their pointwise multiset union, \ie, $(\typctx + \typctx')(x) = \typctx(x) + \typctx'(x)$ for every variable $x$.

The inference rules of the type system are the following (see \cite[Definition 10]{DBLP:conf/fscd/KerinecMR21}):
\begin{gather*}
	\begin{prooftree}
		\infer0[$\text{var}$]{\var: \mset{\alpha} \vdash \var : \alpha}
	\end{prooftree}
	\qquad
	\begin{prooftree}
		\hypo{\typctx \vdash \tm : M \Rightarrow \alpha}
		\hypo{\typctx' \vdash \tmtwo : M}
	\infer2[$\text{app}$]{\Gamma + \Gamma' \vdash  \tm\tmtwo : \alpha}
	\end{prooftree}
	\\
	\begin{prooftree}
		\hypo{\val \text{ variable or abstraction}}
		\infer1[$\text{val}_0$]{\vdash \val : \mset{\,}}
	\end{prooftree}
	\qquad
	\begin{prooftree}
		\hypo{\typctx_1 \vdash \tm : \alpha_1}
		\hypo{\cdots}
		\hypo{\typctx_n \vdash \tm : \alpha_n}
		\hypo{n > 0}
		\infer4[$\text{val}_{>0}$]{\sum_{i=1}^n \typctx_i \vdash \tm : \mset{\alpha_1, \dots, \alpha_n}}
	\end{prooftree}	
	\qquad
	\begin{prooftree}
	\hypo{\typctx, x : M \vdash \tm : \alpha}
	\infer1[$\text{lam}$]{\typctx \vdash \la{\var}{\tm} : M \Rightarrow \alpha}
	\end{prooftree}
\end{gather*}

Let $\tm \defeq w((\lambda x.w')(zy))$ and $\tm' \defeq (\lambda x.ww')(zy)$.
According to the reduction defined in \cite{DBLP:conf/fscd/KerinecMR21} to characterize operationally \cbv solvability, we have $\tm \to \tm'$ (via rule $\sigma_{3}$).

The judgment $w:[\mset{a_1,a_2} \Rightarrow \alpha], z:\mset{\mset{b_1} \Rightarrow \mset{}, \mset{b_2} \Rightarrow \mset{}}, y:\mset{b_1,b_2}, w': \mset{a_1,a_2} \vdash \tm : \alpha$ is derivable for $\tm$, for every type $\alpha$ and every pairwise distinct constants $a_1, a_2, b_1, b_2$.
Indeed, 
\begin{gather*}
	\begin{prooftree}[label separation=0.2em]
		\infer0[\footnotesize{$\text{var}$}]{w:\mset{\mset{a_1,a_2} \Rightarrow \alpha} \vdash w : \mset{a_1,a_2} \Rightarrow \alpha}
		\hypo{}
		\ellipsis{$\Pi_1$}{\typctx_1 \vdash (\lambda x.w')(zy) : a_1}
		\hypo{}
		\ellipsis{$\Pi_2$}{\typctx_2 \vdash (\lambda x.w')(zy) : a_2}
		\infer[separation=5em]2[\footnotesize{$\text{val}_{>0}$}]{z:[\mset{b_1} \Rightarrow \mset{}, \mset{b_2} \Rightarrow \mset{}], y:\mset{b_1,b_2}, w': \mset{a_1,a_2} \vdash (\lambda x.w')(zy) : \mset{a_1,a_2}}
		\infer2[\footnotesize{$\text{app}$}]{w:[\mset{a_1,a_2} \Rightarrow \alpha], z:\mset{\mset{b_1} \Rightarrow \mset{}, \mset{b_2} \Rightarrow \mset{}}, y:\mset{b_1,b_2}, w': \mset{a_1,a_2} \vdash w((\lambda x.w')(zy)) : \alpha}
	\end{prooftree}
\end{gather*}
\noindent where, for $i \in \{1,2\}$, we set  $\typctx_i \defeq z:[\mset{b_i} \Rightarrow \mset{}], y:\mset{b_i}, w': \mset{a_i}$ (hence $\typctx_1+ \typctx_2 = z:[\mset{b_1} \Rightarrow \mset{}, \mset{b_2} \Rightarrow \mset{}], y:\mset{b_1,b_2}, w': \mset{a_1,a_2}$)  and  
\begin{gather*}
	\Pi_i \defeq 
	\begin{prooftree}
	\infer0[$\text{var}$]{w': \mset{a_i} \vdash w' : a_i}
	\infer1[$\text{lam}$]{w': \mset{a_i} \vdash \lambda x.w' : \mset{} \Rightarrow a_i}
	\infer0[$\text{var}$]{z:\mset{\mset{b_i} \Rightarrow \mset{}} \vdash z : \mset{b_i} \Rightarrow \mset{}}
	\infer0[$\text{var}$]{y:\mset{b_i} \vdash y : b_i}
	\infer1[$\text{val}_{>0}$]{y:\mset{b_i} \vdash y : \mset{b_i}}
	\infer2[$\text{app}$]{z:\mset{\mset{b_i} \Rightarrow \mset{}}, y:\mset{b_i} \vdash zy : \mset{}}
	\infer2[$\text{app}$]{z:[\mset{b_i} \Rightarrow \mset{}], y:\mset{b_i}, w': \mset{a_i} \vdash (\lambda x.w')(zy) : a_i}
	\end{prooftree}
\end{gather*}

But the judgment $w:[\mset{a_1,a_2} \Rightarrow \alpha], z:\mset{\mset{b_1} \Rightarrow \mset{}, \mset{b_2} \Rightarrow \mset{}}, y:\mset{b_1,b_2}, w': \mset{a_1,a_2} \vdash \tm' : \alpha$ is not derivable for $\tm'$, essentially because it is impossible to derive the judgment $y : \mset{b_1, b_2}, z : \mset{\mset{b_1} \Rightarrow \mset{}, \mset{b_2} \Rightarrow \mset{}} \vdash zy : \mset{}$.
Therefore, subject reduction \cite[Proposition 14.i]{DBLP:conf/fscd/KerinecMR21} does not hold.

\section{Preliminaries and Notations in Rewriting}
\label{sect:preliminaries}
For a relation $R$ on a set of terms, $R^*$ is its reflexive-transitive closure. 
Given a relation $\Rew{\Rule}$, an $\Rule$-\emph{evaluation}
(or simply evaluation if unambiguous) $\deriv$ is a finite sequence of terms $(\tm_i)_{0 \leq i \leq n}$ (for some $n \geq 0$) such that $\tm_i \Rew{\Rule} \tm_{i+1}$ for all $1 \leq i < n$, and we write $\deriv \colon \tm \Rew{\Rule}^* \tmtwo$ if $\tm_0 = \tm$ and $\tm_n = \tmtwo$. The
\emph{length} $n$ of $\deriv$ is denoted by $\size{\deriv}$, and $\size{\deriv}_a$ is the number of $a$-\emph{steps} (\ie the number of $\tm_i \Rew{a} \tm_{i+1}$ for some $1 \leq i \leq n$) in $\deriv$, for a given subrelation $\Rew{a}$ of $\Rew{\Rule}$.

A term $\tm$ is $\Rule$-\emph{normal} if there is no $\tmtwo$ such that $\tm \Rew{\Rule} \tmtwo$.
An evaluation $\deriv \colon \tm \Rew{\Rule}^* \tmtwo$ is \emph{$\Rule$-normalizing} if $\tmtwo$ is $\Rule$-normal.
A term $\tm$ is \emph{weakly $\Rule$-normalizing} if there is a $\Rule$-normalizing evaluation $\deriv \colon \tm \Rew{\Rule}^* \tmtwo$; and $\tm$ is \emph{strongly $\Rule$-normalizing} if there no infinite sequence $(\tm_i)_{i \in \nat}$  such that $\tm_0  = \tm$ and $\tm_i \Rew{\Rule} \tm_{i+1}$ for all $i \in \nat$.
Clearly, strong $\Rule$-normalization implies weak $\Rule$-normalization.

A relation $\Rew{\Rule}$ is \emph{diamond} if $\tmtwo_1 \,{}_\Rule\!\!\lto \tm \Rew{\Rule} \tmtwo_2$ and $\tmtwo_1 \neq \tmtwo_2$ imply $\tmtwo_1 \Rew{\Rule} \tmthree \, {}_\Rule\!\!\lto \tmtwo_2$ for some $\tmthree$. 
As a consequence:
\begin{enumerate}
	\item $\Rew{\Rule}$ is confluent (\ie $\tmtwo_1 \,{}_\Rule^*\!\!\lto \tm \Rew{\Rule}^* \tmtwo_2$  implies $\tmtwo_1 \Rew{\Rule}^* \tmthree \, {}_\Rule^*\!\!\lto \tmtwo_2$ for some $\tmthree$); \item any term $\tm$ has at most one normal form (\ie if $\tm \Rew{\Rule}^* \tmtwo$  and $\tm \Rew{\Rule}^* \tmthree$ with $\tmtwo$ and $\tmthree$ $\Rule$-normal, then $\tmtwo = \tmthree$);
	\item all $\Rule$-evaluations with the same start and end terms have the same length (\ie if $\deriv \colon \tm \Rew{\Rule}^* \tmtwo$  and $\deriv' \colon \tm \Rew{\Rule}^* \tmtwo$ then $\size{\deriv} = \size{\deriv'}$);
	\item $\tm$ is weakly $\Rule$-normalizing iff it is strongly $\Rule$-normalizing.
\end{enumerate}

Two relations $\Rew{\Rule_1}$ and $\Rew{\Rule_2}$ \emph{strongly commute} if $\tmtwo_1 \,{}_{\Rule_1}\!\!\!\lto \tm \Rew{\Rule_2} \tmtwo_2$ implies $\tmtwo_1 \Rew{\Rule_2} \tmthree \, {}_{\Rule_2}\!\!\!\lto \tmtwo_2$ for some $\tmthree$. 
If $\Rew{\Rule_1}$ and $\Rew{\Rule_2}$ strongly commute and are diamond, then 
\begin{enumerate}
	\item $\Rew{\Rule} \, = \, \Rew{\Rule_1} \!\cup \Rew{\Rule_2}$ is diamond,
	\item all $\Rule$-evaluations with the same start and end terms have the same number of any kind of steps (\ie if $\deriv \colon \tm \Rew{\Rule}^* \tmtwo$  and $\deriv' \colon \tm \Rew{\Rule}^* \tmtwo$ then $\size{\deriv}_{\Rule_1} = \size{\deriv'}_{\Rule_1}$ and $\size{\deriv}_{\Rule_2} = \size{\deriv'}_{\Rule_2}$).
\end{enumerate}

It is a strong form of confluence and implies \emph{uniform normalization} (if there is a normalizing sequence from $\tm$ then there are no diverging sequences from $\tm$) and the 
\emph{random descent property} (all normalizing sequences from $\tm$ have the same length)

\section{Proofs of \Cref{sect:vsc} (Value Substitution Calculus)}

\begin{lemma}[Basic Properties of $\vsubcalc$]
	\label{l:basic-value-substitution}
	\hfill
	\begin{enumerate}
		\item\label{p:basic-value-substitution-tom-toe-terminates} $\tom$ and $\toe$ are strongly normalizing (separately).
		\item\label{p:basic-value-substitution-tom-toe-diamond-open} $\tomo$ and $\toeo$ are diamond  (separately).
		
		\item\label{p:basic-value-substitution-tom-toe-commute-open}  $\tomo$ and $\toeo$ strongly commute.
	\end{enumerate}
\end{lemma}

\begin{proof}
	The statements of \reflemma{basic-value-substitution} are a refinement of some results proved in \cite{AccattoliPaolini12}, where $\tovsubo$ is denoted by $\to_\mathsf{w}$.
	\begin{enumerate}
		\item See \cite[Lemma~3]{AccattoliPaolini12}.
		
		\item We prove that $\tomo$ is diamond, \ie if $\tmtwo \lRew{\wmsym} \tm \tomo \tmthree$ with $\tmtwo \neq \tmthree$ then there exists $\tmp \in \Lambda_\vsub$ such that $\tmtwo \tomo \tmp \lRew{\wmsym} \tmthree$.
		The proof is by induction on the definition of $\tomo$. 
		Since there $\tm \tomo \tmthree \neq \tmtwo$ and the reduction $\tomo$ is weak, there are only eight cases:
		\begin{itemize}
			\item \emph{Step at the Root for $\tm \!\tomo\! \tmtwo$ and Application Right for $\tm \!\tomo\! \tmthree$}, \ie $\tm \defeq \sctxp{\la\var\tmfive}\tmfour \rtom \sctxp{\tmfive\esub{\var}{\tmfour}} \eqdef \tmtwo$ and $\tm \!\rtom\! \sctxp{\la\var\tmfive}\tmfourp\! \eqdef \tmthree$ with $\tmfour \!\tomo\! \tmfourp$: then, $\tmtwo \!\tomo\! \sctxp{\tmfive\esub{\var}{\tmfourp}} \!\lRew{\wmsym}\! \tmthree$;
			
			\item \emph{Step at the Root for $\tm \tomo \tmtwo$ and Application Left for $\tm \tomo \tmthree$}, \ie, for some $n > 0$, 
			$$
			\tm \defeq (\la\var\tmfive)\esub{\var_1}{\tm_1}\dots\esub{\var_n}{\tm_n}\tmfour \allowbreak\rtom \tmfive\esub{\var}{\tmfour}\esub{\var_1}{\tm_1}\dots\esub{\var_n}{\tm_n} \eqdef \tmtwo
			$$
			whereas $\tm \tomo \allowbreak (\la\var\tmfive)\esub{\var_1}{\tm_1}\dots\esub{\var_j}{\tmp_j}\dots\esub{\var_n}{\tm_n}\tmfour \eqdef \tmthree$ with $\tm_j \tomo \tmp_j$ for some $1 \leq j \leq n$: then, 
			\begin{align*}
			\tmtwo \tomo \allowbreak \tmfive\esub{\var}{\tmfour}\esub{\var_1}{\tm_1}\dots\esub{\var_j}{\tmp_j}\dots\esub{\var_n}{\tm_n} \lRew{\wmsym} \tmthree;
			\end{align*}
			\item \emph{Application Left for $\tm \tomo \tmtwo$ and Application Right for $\tm \tomo \tmthree$}, \ie $\tm \defeq \tmfour\tmfive \tomo \tmfourp\tmfive \eqdef \tmtwo$ and $\tm \tomo \tmfour\tmfivep \eqdef \tmthree$ with $\tmfour \tomo \tmfourp$ and $\tmfive \tomo \tmfivep$: then, $\tmtwo \tomo \tmfourp\tmfivep\! \lRew{\wmsym} \tmthree$;
			\item \emph{Application Left for both $\tm \tomo \tmtwo$ and $\tm \tomo \tmthree$}, \ie $\tm \defeq \tmfour\tmfive \tomo \tmfourp\tmfive \eqdef \tmtwo$ and $\tm \tomo \tmfour''\tmfive \eqdef \tmthree$ with $\tmfourp \lRew{\wmsym} \tmfour \tomo \tmfour''$: by \ih, there exists $\tmfour_0 \in \Lambda_\vsub$ such that $\tmfourp \tomo \tmfour_0 \lRew{\msym} \tmfour''$, hence $\tmtwo \tomo \tmfour_0\tmfive \lRew{\msym} \tmthree$;
			\item \emph{Application Right for both $\tm \tomo \tmtwo$ and $\tm \tomo \tmthree$}, \ie $\tm \defeq \tmfive\tmfour \tomo \tmfive\tmfourp \eqdef \tmtwo$ and $\tm \tomo \tmfive\tmfour'' \eqdef \tmthree$ with $\tmfourp \lRew{\wmsym} \tmfour \tomo \tmfour''$: by \ih, there exists $\tmfour_0 \in \Lambda_\vsub$ such that $\tmfourp \tomo \tmfour_0 \lRew{\wmsym} \tmfour''$, hence $\tmtwo \tomo \tmfive\tmfour_0 \lRew{\wmsym} \tmthree$;
			\item \emph{$\mathsf{ES}$ left for $\tm \tomo \tmtwo$ and $\mathsf{ES}$ right for $\tm \tomo \tmthree$}, \ie $\tm \defeq \tmfour\esub\var\tmfive \tomo \tmfourp\esub\var\tmfive \eqdef \tmtwo$ and $\tm \tomo \tmfour\esub\var\tmfivep \eqdef \tmthree$ with $\tmfour \tomo \tmfourp$ and $\tmfive \tomo \tmfivep$: then, 
			$$
			\tmtwo \tomo \tmfourp\esub\var\tmfivep\! \lRew{\wmsym} \tmthree
			$$
			\item \emph{$\mathsf{ES}$ left for both $\tm \tomo \tmtwo$ and $\tm \tomo \tmthree$}, \ie $\tm \defeq \tmfour\esub\var\tmfive \tomo \tmfourp\esub\var\tmfive \eqdef \tmtwo$ and $\tm \tomo \tmfour''\esub\var\tmfive \eqdef \tmthree$ with $\tmfourp \lRew{\wmsym} \tmfour \tomo \tmfour''$: by \ih, there exists $\tmfour_0 \in \Lambda_\vsub$ such that $\tmfourp \tomo \tmfour_0 \lRew{\wmsym} \tmfour''$, hence $\tmtwo \tom \tmfour_0\esub\var\tmfive \lRew{\wmsym} \tmthree$;
			\item \emph{$\mathsf{ES}$ right for both $\tm \tomo \tmtwo$ and $\tm \tomo \tmthree$}, \ie $\tm \defeq \tmfive\esub\var\tmfour \tomo \tmfive\esub\var\tmfourp \eqdef \tmtwo$ and $\tm \tomo \tmfive\esub\var{\tmfour''} \eqdef \tmthree$ with $\tmfourp \lRew{\wmsym} \tmfour \tom \tmfour''$: by \ih, there exists $\tmfour_0 \in \Lambda_\vsub$ such that $\tmfourp \tomo \tmfour_0 \lRew{\wmsym} \tmfour''$, hence $\tmtwo \tom \tmfive\esub\var{\tmfour_0} \lRew{\wmsym} \tmthree$.
		\end{itemize}
		
		\smallskip
		We prove that $\toeo$ is diamond, \ie if $\tmtwo \lRew{\wesym} \tm \toeo \tmthree$ with $\tmtwo \neq \tmthree$ then there exists $\tmfour \in \Lambda_\vsub$ such that $\tmtwo \toeo \tmp \lRew{\esym} \tmthree$.
		The proof is by induction on the definition of $\toeo$. 
		Since there $\tm \toeo \tmthree \neq \tmtwo$ and the reduction $\toeo$ is weak, there are only eight cases:
		\begin{itemize}
			\item \emph{Step at the Root for $\tm \!\toeo\! \tmtwo$} and \emph{$\mathsf{ES}$ left for $\tm \!\toeo\! \tmthree$}, \ie $\tm \defeq \tmfour\esub\var{\sctxp{\val}} \rtoe \sctxp{\tmfour\isub{\var}{\val}} \eqdef \tmtwo$ and $\tm \!\rtoe\! \tmfourp\esub\var{\sctxp{\val}}\! \eqdef \tmthree$ with $\tmfour \!\toeo\! \tmfourp$: then, 
			$$
			\tmtwo \!\toeo\! \sctxp{\tmfourp\esub{\var}{\val}} \!\lRew{\wesym}\! \tmthree
			$$
			
			\item \emph{Step at the Root for $\tm \toeo \tmtwo$} and \emph{$\mathsf{ES}$ right for $\tm \toeo \tmthree$}, \ie, for some $n > 0$, $\tm \defeq \tmfour\esub\var{\val\esub{\var_1}{\tm_1}\dots\esub{\var_n}{\tm_n}} \allowbreak\rtoe \tmfour\isub{\var}{\val}\esub{\var_1}{\tm_1}\dots\esub{\var_n}{\tm_n} \eqdef \tmtwo$ whereas $\tm \toeo \allowbreak \tmfour\esub{\var}{\val\esub{\var_1}{\tm_1}\dots\esub{\var_j}{\tmp_j}\dots\esub{\var_n}{\tm_n}} \eqdef \tmthree$ with $\tm_j \toeo \tmp_j$ for some $1 \leq j \leq n$: then, 
			\begin{align*}
			\tmtwo \toeo \allowbreak \tmfour\isub{\var}{\val}\esub{\var_1}{\tm_1}\dots\esub{\var_j}{\tmp_j}\dots\esub{\var_n}{\tm_n} \lRew{\wesym} \tmthree;
			\end{align*}
			\item \emph{Application Left for $\tm \toeo \tmtwo$} and \emph{Application Right for $\tm \toeo \tmthree$}, \ie $\tm \defeq \tmfour\tmfive \toeo \tmfourp\tmfive \eqdef \tmtwo$ and $\tm \toeo \tmfour\tmfivep \eqdef \tmthree$ with $\tmfour \toeo \tmfourp$ and $\tmfive \toeo \tmfivep$: then, $\tmtwo \toeo \tmfourp\tmfivep\! \lRew{\wesym} \tmthree$;
			\item \emph{Application Left for both $\tm \toeo \tmtwo$ and $\tm \toeo \tmthree$}, \ie $\tm \defeq \tmfour\tmfive \toeo \tmfourp\tmfive \eqdef \tmtwo$ and $\tm \toeo \tmfour''\tmfive \eqdef \tmthree$ with $\tmfourp \lRew{\wesym} \tmfour \toeo \tmfour''$: by \ih, there exists $\tmfour_0 \in \Lambda_\vsub$ such that $\tmfourp \toeo \tmfour_0 \lRew{\wesym} \tmfour''$, hence $\tmtwo \toeo \tmfour_0\tmfive \lRew{\wesym} \tmthree$;
			\item \emph{Application Right for both $\tm \toeo \tmtwo$ and $\tm \toeo \tmthree$}, \ie $\tm \defeq \tmfive\tmfour \toeo \tmfive\tmfourp \eqdef \tmtwo$ and $\tm \toeo \tmfive\tmfour'' \eqdef \tmthree$ with $\tmfourp \lRew{\wesym} \tmfour \toeo \tmfour''$: by \ih, there exists $\tmfour_0 \in \Lambda_\vsub$ such that $\tmfourp \toeo \tmfour_0 \lRew{\wesym} \tmfour''$, hence $\tmtwo \toeo \tmfive\tmfour_0 \lRew{\wesym} \tmthree$;
			\item \emph{$\mathsf{ES}$ left for $\tm \toeo \tmtwo$} and \emph{$\mathsf{ES}$ right for $\tm \toeo \tmthree$}, \ie $\tm \defeq \tmfour\esub\var\tmfive \toeo \tmfourp\esub\var\tmfive \eqdef \tmtwo$ and $\tm \toeo \tmfour\esub\var\tmfivep \eqdef \tmthree$ with $\tmfour \toeo \tmfourp$ and $\tmfive \toeo \tmfivep$: then, $\tmtwo \toeo \tmfourp\esub\var\tmfivep\! \lRew{\wesym} \tmthree$;
			\item \emph{$\mathsf{ES}$ left for both $\tm \toe \tmtwo$ and $\tm \toe \tmthree$}, \ie $\tm \defeq \tmfour\esub\var\tmfive \toe \tmfourp\esub\var\tmfive \eqdef \tmtwo$ and $\tm \toe \tmfour''\esub\var\tmfive \eqdef \tmthree$ with $\tmfourp \lRew{\esym} \tmfour \toe \tmfour''$: by \ih, there exists $\tmfour_0 \in \Lambda_\vsub$ such that $\tmfourp \toe \tmfour_0 \lRew{\esym} \tmfour''$, hence $\tmtwo \toe \tmfour_0\esub\var\tmfive \lRew{\esym} \tmthree$;
			\item \emph{$\mathsf{ES}$ right for both $\tm \toe \tmtwo$ and $\tm \toe \tmthree$}, \ie $\tm \defeq \tmfive\esub\var\tmfour \toe \tmfive\esub\var\tmfourp \eqdef \tmtwo$ and $\tm \toe \tmfive\esub\var{\tmfour''} \eqdef \tmthree$ with $\tmfourp \lRew{\esym} \tmfour \toe \tmfour''$: by \ih, there exists $\tmfour_0 \in \Lambda_\vsub$ such that $\tmfourp \toe \tmfour_0 \lRew{\esym} \tmfour''$, hence $\tmtwo \toe \tmfive\esub\var{\tmfour_0} \lRew{\esym} \tmthree$.
		\end{itemize}
		
		\smallskip
		
		Note that in \cite[Lemma~11]{AccattoliPaolini12} it has just been proved the strong confluence of $\tovsub$, not of $\tom$ or $\toe$.
		
		\item We show that $\toeo$ and $\tomo$ strongly commute, \ie if $\tmtwo \lRew{\wesym} \tm \tomo \tmthree$, then $\tmtwo \neq \tmthree$ and there is $\tmp \in \Lambda_\vsub$ such that $\tmtwo \tomo \tmp \lRew{\wesym} \tmthree$. 
		The proof is by induction on the definition of $\tm \toeo \tmtwo$. 
		The proof that $\tmtwo \neq \tmthree$ is left to the reader.
		Since the $\toe$ and $\tom$ cannot reduce under $\l$'s, all values are $\omsym$-normal and $\oesym$-normal. So, there are the following cases.
		\begin{itemize}
			\item \emph{Step at the Root for $\tm \toeo \tmtwo$} and \emph{$\mathsf{ES}$ left for $\tm \tomo \tmthree$}, \ie $\tm \defeq \tmfour\esub\varthree{\sctxp\val} \toeo \sctxp{\tmfour\isub\varthree{\val}} \eqdef \tmtwo$ and $\tm \tomo \tmfourp\esub\varthree{\sctxp{\val}} \eqdef \tmthree$ with $\tmfour \tomo \tmfourp$: then 
			$$
			\tmtwo \tomo \sctxp{\tmfourp\isub\varthree{\val}} \lRew{\wesym} \tmtwo
			$$
			\item \emph{Step at the Root for $\tm \toeo \tmtwo$} and \emph{$\mathsf{ES}$ right for $\tm \tomo \tmthree$}, \ie 
			\begin{align*}
			\tm &\defeq \tmfour\esub\varthree{\val\esub{\var_1}{\tm_1}\dots\esub{\var_n}{\tm_n}} \\
			&\toeo \tmfour\isub\varthree{\val}\esub{\var_1}{\tm_1}\dots\esub{\var_n}{\tm_n} \eqdef \tmtwo
			\end{align*}
			and $\tm \tomo \tmfour\esub\varthree{\val\esub{\var_1}{\tm_1}\dots\esub{\var_j}{\tmp_j}\dots\esub{\var_n}{\tm_n}} \eqdef \tmthree$
			for some $n > 0$, and $\tm_j \tomo \tmp_j$ for some $1 \leq j \leq n$: then, $\tmtwo \tomo \tmfour\isub\varthree{\val}\esub{\var_1}{\tm_1}\dots\esub{\var_j}{\tmp_j}\dots\esub{\var_n}{\tm_n} \lRew{\wesym} \tmthree$;
			
			\item \emph{Application Left for $\tm \toe \tmtwo$} and \emph{Application Right for $\tm \tom \tmthree$}, \ie $\tm \defeq \tmfour\tmfive \toeo \tmfourp\tmfive \eqdef \tmtwo$ and $\tm \tomo \tmfour\tmfivep \eqdef \tmthree$ with $\tmfour \toe \tmfourp$ and $\tmfive \tomo \tmfivep$: then, $\tm \tomo \tmfourp\tmfivep \lRew{\wesym} \tmtwo$;
			\item \emph{Application Left for both $\tm \toeo \tmtwo$ and $\tm \tomo \tmthree$}, \ie $\tm \defeq \tmfour\tmfive \toeo \tmfourp\tmfive \eqdef \tmtwo$ and $\tm \tomo \tmfour''\tmfive \eqdef \tmthree$ with $\tmfourp \lRew{\wesym} \tmfour \tomo \tmfour''$: by \ih, there exists $\tmsix \in \Lambda_\vsub$ such that $\tmfourp \tomo \tmsix \lRew{\wesym} \tmfour''$, hence $\tmtwo \tomo \tmsix\tmfive \lRew{\wesym} \tmthree$;
			\item \emph{Application Left for $\tm \toeo \tmtwo$} and \emph{Step at the Root for $\tm \tomo \tmthree$}, \ie $\tm \defeq (\la\var\tmfive){\esub{\var_1}{\tm_1}\dots\esub{\var_n}{\tm_n}}\tmfour \toeo (\la\var\tmfive){\esub{\var_1}{\tm_1}\dots\esub{\var_j}{\tmp_j}\dots\esub{\var_n}{\tm_n}}\tmfour \eqdef \tmtwo$ with $n > 0$ and $\tm_j \toeo \tmp_j$ for some $1 \leq j \leq n$, and 
			$$
			\tm \tomo \allowbreak \tmfive\esub\var\tmfour{\esub{\var_1}{\tm_1}\dots\esub{\var_n}{\tm_n}} \eqdef \tmthree
			$$ 
			Then,
			\begin{equation*}
			\tmtwo \tomo \tmfive\esub\var\tmfour{\esub{\var_1}{\tm_1}\dots\esub{\var_j}{\tmp_j}\dots\esub{\var_n}{\tm_n}} \lRew{\wesym} \tmthree;
			\end{equation*}
			\item \emph{Application Right for $\tm \toeo \tmtwo$} and \emph{Application Left for $\tm \tom \tmthree$}, \ie $\tm \defeq \tmfive\tmfour \toeo \tmfive\tmfourp \eqdef \tmtwo$ and $\tm \tomo \tmfivep\tmfour \eqdef \tmthree$ with $\tmfour \toeo \tmfourp$ and $\tmfive \tomo \tmfivep$: then, $\tmtwo \tomo \tmfivep\tmfourp \lRew{\wesym} \tmthree$;
			\item \emph{Application Right for both $\tm \toeo \tmtwo$ and $\tm \tomo \tmthree$}, \ie $\tm \defeq \tmfive\tmfour \toeo \tmfive\tmfourp \eqdef \tmtwo$ and $\tm \tomo \tmfive\tmfour'' \eqdef \tmthree$ with $\tmfourp \lRew{\wesym} \tmfour \tomo \tmfour''$: by \ih, there exists $\tmsix \in \Lambda_\vsub$ such that $\tmfourp \tomo \tmsix \lRew{\wesym} \tmfour''$, hence $\tmtwo \tomo \tmfive\tmsix \lRew{\wesym} \tmthree$;
			\item \emph{Application Right for $\tm \toeo \tmtwo$} and \emph{Step at the Root for $\tm \tomo \tmthree$}, \ie $\tm \defeq \sctxp{\la\var\tmfive}\tmfour \toeo \sctxp{\la\var\tmfive}\tmfourp \eqdef \tmtwo$ with $\tmfour \toeo \tmfourp$\!, and $\tm \tomo \allowbreak \sctxp{\tmfive\esub\var\tmfour} \eqdef \tmthree$: then, $\tmtwo \tomo \sctxp{\tmfive\esub\var\tmfourp} \lRew{\wesym} \tmthree$;
			\item \emph{$\mathsf{ES}$ left for $\tm \toeo \tmtwo$} and \emph{$\mathsf{ES}$ right for $\tm \tomo \tmthree$}, \ie $\tm \defeq \tmfour\esub\var\tmfive \toeo \tmfourp\esub\var\tmfive \eqdef \tmtwo$ and $\tm \tomo \tmfour\esub\var\tmfivep \eqdef \tmthree$ with $\tmfour \toe \tmfourp$ and $\tmfive \tomo \tmfivep$: then, $\tmtwo \tomo \tmfourp\esub\var\tmfivep \lRew{\wesym} \tmthree$;
			\item \emph{$\mathsf{ES}$ left for both $\tm \toeo \tmtwo$ and $\tm \tomo \tmthree$}, \ie $\tm \defeq \tmfour\esub\var\tmfive \toe \tmfourp\esub\var\tmfive \eqdef \tmtwo$ and $\tm \tomo \tmfour''\esub\var\tmfive \eqdef \tmthree$ with $\tmfourp \lRew{\wesym} \tmfour \tomo \tmfour''$: by \ih,there is $\tmsix \in \Lambda_\vsub$ such that $\tmfourp \tomo \tmsix \lRew{\wesym} \tmfour''$, hence $\tmtwo \tomo \tmsix\esub\var\tmfive \lRew{\wesym} \tmthree$;
			\item \emph{$\mathsf{ES}$ right for $\tm \toeo \tmtwo$} and \emph{$\mathsf{ES}$ left for $\tm \tomo \tmthree$}, \ie $\tm \defeq \tmfive\esub\var\tmfour \toeo \tmfive\esub\var\tmfourp \eqdef \tmtwo$ and $\tm \tomo \tmfivep\esub\var\tmfour \eqdef \tmthree$ with $\tmfour \toeo \tmfourp$ and $\tmfive \tomo \tmfivep$: then, $\tmtwo \tomo \tmfivep\esub\var\tmfourp \lRew{\wesym} \tmthree$;
			\item \emph{$\mathsf{ES}$ right for both $\tm \toeo \tmtwo$ and $\tm \tomo \tmthree$}, \ie $\tm \defeq \tmfive\esub\var\tmfour \toeo \tmfive\esub\var{\tmfourp} \eqdef \tmtwo$ and $\tm \tomo \tmfive\esub\var{\tmfour''} \eqdef \tmthree$ with $\tmfour \lRew{\esym} \tmfourp \tomo \tmfour''$: by \ih, there is $\tmsix \in \Lambda_\vsub$ such that $\tmfour \tomo \tmsix \lRew{\wesym} \tmfour''$, so $\tmtwo \tomo \tmfive\esub\var\tmsix \lRew{\wesym} \tmthree$.
			\qedhere
		\end{itemize}
	\end{enumerate}
\end{proof}

\begin{proposition}[Properties of the open reduction]
	\label{propappendix:properties-open-reduction}
	\NoteState{prop:properties-open-reduction}
	\begin{enumerate}
		\item \label{pappendix:properties-open-reduction-diamond} The reduction $\tovsubo$ is diamond.
		\item \label{pappendix:properties-open-redction-harmony} A term is $\osym$-normal if and only if it is a fireball, where \emph{fireballs} (and \emph{inert terms}) are defined by:
	\end{enumerate}
	\begin{alignat*}{7}
	\textsc{Proper inert terms } \ & \pitm \grameq  \var\fire \mid \pitm \fire \mid \pitm \esub{\var}{\pitmtwo}
	& \qquad
	\textsc{Fireballs } \ & \fire \grameq \val \mid \pitm \mid \fire \esub{\var}{\pitm}
	\end{alignat*}
\end{proposition}

\begin{proof}
	\begin{enumerate}
		\item Strong commutation of $\tomo$ and $\toeo$ is proven in \reflemmap{basic-value-substitution}{tom-toe-commute-open}.
		Diamond of $\tovsubo$ follows from that, from diamond for $\tomo$ and $\toeo$ (\reflemmap{basic-value-substitution}{tom-toe-diamond-open}), 
		and from Hindley-Rosen lemma (\cite[Prop. 3.3.5]{Barendregt84}).
		
		\item See \cite[Lemma~5]{AccattoliPaolini12}, where $\tovsubo$ is denoted by $\to_\mathsf{w}$.
		\qedhere
	\end{enumerate}
	
\end{proof}

\begin{remark}[Alternative definition of inert terms]\label{rmk:inert-alternative}
	Given the definitions above of proper inert terms $\pitm$ and fireballs $\fire$, inert terms $\itm$  (\Cref{rmk:inert-def}) can equivalently be defined by:
	\begin{align}\label{eq:inert}
		\textsc{Inert term} \  \itm &\grameq \var \mid \itm \fire \mid \itm \esub{\var}{\pitm}
	\end{align}
	The proof of the equivalence of the two definitions is straightforward.
	From now on, whenever we prove a property by induction on inert terms, we refer to the definition in \eqref{eq:inert}. 
\end{remark}

\begin{lemma}[Sizes for inert terms]
	\label{l:sizes-inert}
	For every inert term $\itm$, $\sizeo{\itm} = \sizes{\itm}$.
\end{lemma}

\begin{proof}
	By induction of the inert term $\itm$ (see \Cref{rmk:inert-alternative}).
	Cases:
	\begin{itemize}
		\item \emph{Variable}, \ie $\itm = \var$. 
		Then, $\sizeo{\itm} = 0 = \sizes{\itm}$.
		
%
		\item \emph{Application}, \ie $\itm = \itmtwo \fire$. 
		By \ih, $\sizeo{\itmtwo} = \sizes{\itmtwo}$.
		So, $\sizeo{\itm} = 1 + \sizeo{\itmtwo} + \sizeo{\fire} = 1 + \sizes{\itmtwo} + \sizeo{\fire} = \sizes{\itm}$.

		\item \emph{Explicit substitution}, \ie $\itm = \itmtwo \esub{\var}{\pitm}$. 
		By \ih, $\sizeo{\itmtwo} = \sizes{\itmtwo}$.
		Thus, $\sizeo{\itm} = \sizeo{\itmtwo} + \sizeo{\pitm} = \sizes{\itmtwo} + \sizeo{\pitm} = \sizes{\itm}$.
		\qedhere
	\end{itemize}
\end{proof}

\begin{lemma}[Shape of strong fireballs]
	\label{l:shape-of-strong-fireballs}
	Let $\tm$ be a \full fireball. Then exactly one of the following holds:
	\begin{itemize}
		\item either $\tm$ is a \full inert term,
		
		\item or $\tm$ is a \full value.
		
	\end{itemize}
\end{lemma}

\begin{proof}
	Proving that at least one of the two holds is left to the reader. 
	We now prove that only one of them holds:
	
	\begin{itemize}
		\item Let $\tm$ be a \full inert term. We prove that $\tm$ is not a \full value by structural induction on $\tm$:
		\begin{itemize}
			\item \emph{Variable}: Trivial.
			
			\item \emph{Application}: Trivial.
			
			\item \emph{$\mathsf{ES}$}; \ie, $\tm = \sitm \esub{\var}{\sitmtwo}$: Then $\sitm$ is not a \full value---by \ih---, and so neither is $\tm$.
			
		\end{itemize}
		
		\item Let $\tm = \sctxp{\la{\var}{\sfire}}$, with $\sctx = \esub{\var_{1}}{{\sitm}_{1}} \dots \esub{\var_{n}}{{\sitm}_{n}}$, with $n \geq 0$. 
		We prove that $\tm$ is not a \full inert term by induction on $n$:
		\begin{itemize}
			\item \emph{Empty substitution context}; \ie, $\sctx = \ctxhole$: Trivial.
			
			\item \emph{Non-empty substitution context}; \ie, $\sctx = \sctxtwo \esub{\var}{{\sitm}_{n+1}}$: As $\sctxtwop{\la{\var}{\fire}}$ is not a \full inert term by \ih, then neither is $\sctxtwop{\la{\var}{\fire}} \esub{\var}{{\sitm}_{n+1}} = \tm$.
			\qedhere
		\end{itemize}
		
	\end{itemize}
\end{proof}

\begin{proposition}[Simulation]
	\label{propappendix:plotkin-vsc}
	\NoteState{prop:plotkin-vsc}
	Let $\tm$ be a term without \ES. 
	If $\tm \tobvplot \tm'$ then $\tm \tom \cdot \toe \tm'$.
\end{proposition}

\begin{proof}
	By induction on the context $\ctx$ such that $\tm = \ctxp{\tmthree} \tobvplot \ctxp{\tmthree'} = \tm'$ with $\tmthree \rtobvplot \tmthree'$.
			
	If $\ctx = \ctxhole$, then we have the root-step $\tm = (\la{\var}\tmthree)\val \rtobvplot \tmthree \isub{\var}{\val} = \tm'$.  
	Then, $\tm \rtom \tmthree \esub{\var}{\val} \rtoe \tmthree \isub{\var}{\val} = \tm'$.
			
	The cases where $\ctx \neq \ctxhole$ follow  from the \ih easily.
\end{proof}
\section{Proofs of \Cref{sect:solving-strat} (Call-by-Value Solvability and The Solving Strategy)}

\begin{proposition}[Properties of the solving strategy]
	\label{propappendix:properties-solvable-reduction}
	\NoteState{prop:properties-solvable-reduction}
	\hfill
	\begin{enumerate}
		\item \label{pappendix:properties-solvable-reduction-diamond} $\tosolv$ is diamond; $\tosolvm$ and $\tosolve$ strongly commute.
		\item \label{pappendix:properties-solvable-reduction-harmony} A term is $\solvredsym$-normal if and only if it is a solvable fireball, where \emph{solvable fireballs} are defined by:
	\end{enumerate}
	\begin{center}
		$\textsc{Solvable fireballs} \qquad \solvnf \grameq  \itm \mid \la{\var}\solvnf \mid \solvnf \esub{\var}{\pitm}$
	\end{center}
\end{proposition}

\begin{proof}
\hfill
	\begin{enumerate}
		\item For diamond, see \cite[Lemma 11]{AccattoliPaolini12}, where $\tovsubsolv$ is noted  $\to_\mathsf{sw}$. 
		For strong commutation, see \cite[Lemma 4]{AccattoliPaolini12}, where $\tomsolv$ is noted  $\to_\mathsf{db}$, and $\toesolv$ is noted  $\to_\mathsf{vs}$.
		\item See \cite[Lemma 8]{AccattoliPaolini12}, where $\tovsubsolv$ is noted $\to_\mathsf{sw}$.
		\qedhere
	\end{enumerate}
\end{proof}

\section{Proofs of \Cref{sect:types} (Multi Types by Value)}

First, we observe the following property: given a derivation for a term $\tm$, all variables associated with a non-empty multi type in the type context are free variables~of~ $\tm$.
\begin{remark}
	\label{rmk:free-variables}
	If $\namedtyjp{\tderiv}{}{\tm}{\typctx}{\mtype}$ then $\dom{\typctx} \subseteq \fv{\tm}$.
	The proof is by straightforward induction on the derivation $\tderiv$.
\end{remark}

\begin{lemma}[Typing of values: splitting]
	\label{l:typing-value-splitting}
	Let $\namedtyjp{\tderiv}{}{\val}{\typctx}{\mtype}$ (for $\val$ value).
	\begin{enumerate}
		\item \label{p:typing-value-splitting-one} If $\mtype = \emptytype$, then $\dom{\typctx} = \emptyset$ and 
		$\sizem{\tderiv} = 0 = \size{\tderiv}$. 
		
		\item \label{p:typing-value-splitting-two} For every splitting $\mtype = \mtype_{1} \mplus \mtype_{2}$, there exist 
		type derivations $\namedtyjp{\tderiv_{1}}{}{\val}{\typctx_{1}}{\mtype_{1}}$ and 
		$\namedtyjp{\tderiv_{2}}{}{\val}{\typctx_{2}}{\mtype_{2}}$ such that $\typctx = \typctx_{1} \mplus \typctx_{2}$, 
		$\sizem{\tderiv} = \sizem{\tderiv_{1}} + \sizem{\tderiv_{2}}$ and $\size{\tderiv} = \size{\tderiv_{1}} + 
		\size{\tderiv_{2}}$.
		
	\end{enumerate}
\end{lemma}

\begin{proof}\hfill
	\begin{enumerate}
		\item By a simple inspection of the typing rules, $\mtype = \emptytype$ and the fact that $\val$ is a value imply 
		that 
		$\tderiv$ is of the form
		\begin{prooftree}
			\hypo{}
			\infer1[\footnotesize$\ruleMany$]{\tyjp{}{\val}{}{\emptytype}}
		\end{prooftree}
		where $\dom{\typctx} = \emptyset$ and $\sizem{\tderiv} = 0 = \size{\tderiv}$.
		
		\item Let
		\begin{equation*}
		\tderiv = 
		\begin{prooftree}
		\hypo{}
		\ellipsis{$\tderiv_{i}$}{\tyjp{}{\val}{\typctx_{i}}{\ltype_{i}}}
		\delims{\left(}{\right)_{\iI}}
		\infer1[\footnotesize$\ruleMany$]{\tyjp{}{\val}{\bigmplus_{\iI} \typctx_{i}}{\mult{\ltype}_{\iI}}}
		\end{prooftree}
		\end{equation*}
		with $\bigmplus_{\iI} \typctx_{i} = \typctx$ and $\mult{\ltype}_{\iI} = \mtype = \mtype_{1} \mplus \mtype_{2}$. Let 
		$I_{1}$ and $I_{2}$ be sets of indices such that $I = I_{1} \cup I_{2}$, $\mtype_{1} = \mult{\ltype_{i}}_{i \in I_{1}}$ 
		and $\mtype_{2} = \mult{\ltype_{i}}_{i \in I_{2}}$. 
		As $\val$ is a value, 	we can then derive
		\begin{equation*}
		\tderiv_{1} = 
		\begin{prooftree}
		\hypo{}
		\ellipsis{$\tderiv_{i}$}{\tyjp{}{\val}{\typctx_{i}}{\ltype_{i}}}
		\delims{\left(}{\right)_{i \in I_1}}
		\infer1[\footnotesize$\ruleMany$]{\tyjp{}{\val}{\bigmplus_{i \in I_{1}} \typctx_{i}}{\mult{\ltype}_{i \in I_{1}}}}
		\end{prooftree}
		\end{equation*}
		and 
		\begin{equation*}
		\tderiv_{2} = 
		\begin{prooftree}
		\hypo{}
		\ellipsis{$\tderiv_{i}$}{\tyjp{}{\val}{\typctx_{i}}{\ltype_{i}}}
		\delims{\left(}{\right)_{i \in I_2}}
		\infer1[\footnotesize$\ruleMany$]{\tyjp{}{\val}{\bigmplus_{i \in I_{2}} \typctx_{i}}{\mult{\ltype}_{i \in I_{2}}}}
		\end{prooftree}
		\end{equation*}
		noting that 
		$$
		\typctx = \bigmplus_{\iI} \typctx_{i} = \left( \bigmplus_{i \in I_{1}} \typctx_{i} \right) \mplus 
		\left(\bigmplus_{i \in I_{2}} \typctx_{i} \right)
		$$
		with 
		$$
		\sizem{\tderiv} = \sum_{\iI} \sizem{\tderiv_{i}} = \left( \sum_{i \in I_{1}} \sizem{\tderiv_{i}} \right) + \left( 
		\sum_{i \in I_{2}} \sizem{\tderiv_{i}} \right) = \sizem{\tderiv_{1}} + \sizem{\tderiv_{2}}
		$$
		and 
		$$
		\size{\tderiv} = \sum_{\iI} \size{\tderiv_{i}} = \left( \sum_{i \in I_{1}} \size{\tderiv_{i}} \right) + \left( 
		\sum_{i \in I_{2}} \size{\tderiv_{i}} \right) = \size{\tderiv_{1}} + \size{\tderiv_{2}}
		$$
		
	\end{enumerate}
\end{proof}

\begin{lemma}[Substitution]
	\label{lappendix:substitution}	
	\NoteProof{l:substitution}
	Let $\tm$ be a term, $\val$ be a value and $\namedtyjp{\tderiv}{}{\tm}{\typctx, \var \hastype 
		\mtypetwo}{\mtype}$ and $\namedtyjp{\tderivtwo}{}{\val}{\typctxtwo}{\mtypetwo}$ be derivations.
	Then there is a derivation $\namedtyjp{\tderivthree}{}{\tm \isub{\var}{\val}}{\typctx \mplus \typctxtwo}{\mtype}$ 
	with $\sizem{\tderivthree} = \sizem{\tderiv} + \sizem{\tderivtwo}$ and $\size{\tderivthree} \leq \size{\tderiv} + 
	\size{\tderivtwo}$. 
\end{lemma}

\begin{proof}
	By induction on the term $\tm$.
	Cases:
	\begin{itemize}
		\item \emph{Variable}, then are two sub-cases:
		\begin{enumerate}
			\item $\tm = \var$, then $\tm \isub{\var}{\val} = \var \isub{\var}{\val} = \val$ and 			
			$\sizem{\tderiv} = 0$ and $\size{\tderiv} = 1$.
			
			the derivation $\tderiv$ has necessarily the form (for some $n \in \nat$)
			\begin{equation*}
			\tderiv = 
			\begin{prooftree}
			\infer0[\footnotesize$\Ax$]{\tyjp{}{\var}{\var \hastype \mset{\ltype_1}}{\ltype_1}}
			\hypo{\overset{n \in \nat}{\ldots}}
			\infer0[\footnotesize$\Ax$]{\tyjp{}{\var}{\var \hastype \mset{\ltype_n}}{\ltype_n}}
			\infer3[\footnotesize$\ruleManyVar$]{\tyjp{}{\var}{\var \hastype 
					\mset{\ltype_1,\dots,\ltype_n}}{\mset{\ltype_1,\dots,\ltype_n}}}
			\end{prooftree}
			\end{equation*}
			with $\mtype = \mset{\ltype_1, \dots, \ltype_n} = \mtypetwo$ and $\dom{\typctx} = \emptyset$.
			Thus, $\sizem{\tderiv} = 0$ and $\size{\tderiv} = n$.
			Let $\tderivthree = \tderivtwo$: so, $\namedtyjp{\tderivthree}{}{\tm \isub{\var}{\val}}{\typctx \mplus 
				\typctxtwo}{\mtype}$ (since $\typctx \mplus \typctxtwo = \typctxtwo$) with $\sizem{\tderivthree} = \sizem{\tderivtwo} = 
			\sizem{\tderivtwo} + \sizem{\tderiv}$ and $\size{\tderivthree} = \size{\tderivtwo} \leq \size{\tderivtwo} + 
			\size{\tderiv}$ (note that $\size{\tderivthree} = \size{\tderiv} + \size{\tderivtwo}$ if and only if $n=0$).
			
			\item $\tm = \varthree \neq \var$, then $\tm \isub{\var}{\val} = \varthree$ and 
			$\sizem{\tderiv} = 0$, $\size{\tderiv} = 1$, 
			the derivation $\tderiv$ has necessarily the form (for some $n~\in~\nat$)
			\begin{equation*}
			\tderiv = 
			\begin{prooftree}
			\infer0[\footnotesize$\Ax$]{\tyjp{}{\varthree}{\varthree \hastype \mset{\ltype_1}}{\ltype_1}}
			\hypo{\overset{n \in \nat}{\ldots}}
			\infer0[\footnotesize$\Ax$]{\tyjp{}{\varthree}{\varthree \hastype \mset{\ltype_n}}{\ltype_n}}
			\infer3[\footnotesize$\ruleManyVar$]{\tyjp{}{\varthree}{\varthree \hastype 
					\mset{\ltype_1,\dots,\ltype_n}}{\mset{\ltype_1,\dots,\ltype_n}}}
			\end{prooftree}
			\end{equation*}
			where $\mtype = \mset{\ltype_1, \dots, \ltype_n}$ and  $\mtypetwo = \emptymset$ and $\typctx = \varthree \hastype 
			\mtype$ (while $\typctx(\var) = \emptymset$).
			Thus, $\sizem{\tderiv} = 0$ and $\size{\tderiv} = n$.
			By \reflemmap{typing-value-splitting}{one}, from $\namedtyjp{\tderivtwo}{}{\val}{\typctxtwo}{\emptytype}$ it 
			follows that $\sizem{\tderivtwo} = 0 = \size{\tderivtwo}$ and $\dom{\typctxtwo} = \emptyset$.  
			Therefore, $\typctx \mplus \typctxtwo = \typctx$.
			Let $\tderivthree = \tderiv$: so, $\namedtyjp{\tderivthree}{}{\tm \isub{\var}{\val}}{\typctx \mplus 
				\typctxtwo}{\mtype}$  with $\sizem{\tderivthree} = \sizem{\tderiv} = \sizem{\tderiv} + \sizem{\tderivtwo}$ and 
			$\size{\tderivthree} = \size{\tderiv} = \size{\tderiv} + \size{\tderivtwo}$.
		\end{enumerate}
		
		\item \emph{Application}, \ie $\tm = \tmtwo\tmthree$. 
		Then $\tm \isub{\var}{\val} = \tmtwo \isub{\var}{\val} \tmthree \isub{\var}{\val}$ and necessarily
		\begin{equation*}
		\tderiv = 
		\begin{prooftree}
		\hypo{}
		\ellipsis{$\tderiv_{1}$}{\typctx_1, \var \hastype \mtypetwo_1 \vdash \tmtwo \hastype 
			\mset{\larrow{\mtypethree}{\mtype}}}
		\hypo{}
		\ellipsis{$\tderiv_{2}$}{\typctx_2, \var \hastype \mtypetwo_2 \vdash \tmthree \hastype \mtypethree}
		\infer2[\footnotesize$\ruleAp$]{\typctx, \var \hastype \mtypetwo \vdash \tmtwo \tmthree \hastype \mtype}
		\end{prooftree}
		\end{equation*}
		with $\sizem{\tderiv} = \sizem{\tderiv_{1}} + \sizem{\tderiv_{2}} + 1$, $\size{\tderiv} = \size{\tderiv_{1}} + 
		\size{\tderiv_{2}} + 1$, $\typctx = \typctx_1 \mplus \typctx_2$ and $\mtypetwo = \mtypetwo_2 \mplus \mtypetwo_2$. 
		According to \reflemmap{typing-value-splitting}{two} applied to $\tderivtwo$ and to the decomposition $\mtypetwo = 
		\mtypetwo_1 \mplus \mtypetwo_2$, there are contexts $\typctxtwo_1, \typctxtwo_2$ and derivations 
		$\namedtyjp{\tderivtwo_{1}}{}{\val}{\typctxtwo_1}{\mtypetwo_1}$ and 
		$\namedtyjp{\tderivtwo_{2}}{}{\val}{\typctxtwo_2}{\mtypetwo_2}$ such that $\typctxtwo = \typctxtwo_{1} \mplus 
		\typctxtwo_2$, $\sizem{\tderivtwo} = \sizem{\tderivtwo_1} + \sizem{\tderivtwo_2}$ and $\size{\tderivtwo} = 
		\size{\tderivtwo_{1}} + \size{\tderivtwo_{2}}$.
		
		By \ih, there are derivations $\namedtyjp{\tderivthree_1}{}{\tmtwo \isub{\var}{\val}}{\typctx_1 \mplus 
			\typctxtwo_1}{\mult{\larrow{\mtypethree}{\mtype}}}$ and $\namedtyjp{\tderivthree_2}{}{\tmthree 
			\isub{\var}{\val}}{\typctx_2 \mplus \typctxtwo_2}{\mtypethree}$ such that $\sizem{\tderivthree_{i}} = 
		\sizem{\tderiv_{i}} + \sizem{\tderivtwo_{i}}$ and $\size{\tderivthree_{i}} \leq \size{\tderiv_{i}} + 
		\size{\tderivtwo_{i}}$ for all $i \in \{1,2\}$.
		Since $\typctx \mplus \typctxtwo = \typctx_1 \mplus \typctxtwo_1 \mplus \typctx_2 \mplus \typctxtwo_2$, we can 
		build the derivation
		\begin{equation*}
		\tderivthree = 
		\begin{prooftree}
		\hypo{}
		\ellipsis{$\tderivthree_1$}{\typctx_1 \mplus \typctxtwo_1 \vdash \tmtwo\isub{\var}{\val} \hastype 
			\mset{\larrow{\mtypethree}{\mtype}}}
		\hypo{}
		\ellipsis{$\tderivthree_2$}{\typctx_2 \mplus \typctxtwo_2 \vdash \tmthree\isub{\var}{\val} \hastype \mtypethree}
		\infer2[\footnotesize$\ruleAp$]{\typctx \mplus \typctxtwo \vdash \tmtwo \isub{\var}{\val} \tmthree 
			\isub{\var}{\val} \hastype \mtype}
		\end{prooftree}
		\end{equation*}
		where $\sizem{\tderivthree} = \sizem{\tderivthree_{1}} + \size{\tderivthree_{2}} + 1 = \sizem{\tderiv_{1}} + 
		\sizem{\tderivtwo_{1}} + \sizem{\tderiv_{2}} + \sizem{\tderivtwo_{2}} + 1 = \sizem{\tderiv} + \sizem{\tderivtwo}$ and 
		$\size{\tderivthree} = \size{\tderivthree_{1}} + \size{\tderivthree_{2}} + 1 \leq \sizem{\tderiv_{1}} + 
		\sizem{\tderivtwo_{1}} + \sizem{\tderiv_{2}} + \size{\tderivtwo_{2}} + 1 = \size{\tderiv} + \size{\tderivtwo}$.
		
		\item \emph{Abstraction}, \ie $\tm = \la{\vartwo}{\tmtwo}$.
		We can suppose without loss of generality that $\vartwo \notin \fv{\val} \cup \{\var \}$, hence $\tm 
		\isub{\var}{\val} = \la{\vartwo}{\tmtwo\isub{\var}\val}$ and  $\tderiv$ is necessarily of the form (for some $n \in 
		\nat$) 
		\begin{equation*}
		\begin{prooftree}[separation=1em]
		\hypo{}
		\ellipsis{$\tderiv_{i}$}{\typctx_{i}, \vartwo \hastype \mtypethree_{i}, \var \hastype \mtypetwo_{i} \vdash \tmtwo 
			\hastype \mtype_{i}}
		\infer1[\footnotesize$\ruleFun$]{\tyjp{}{\la{\vartwo}{\tmtwo}}{\typctx_{i}, \var \hastype 
				\mtypetwo_{i}}{\ty{\mtypethree_{i}\!}{\!\mtype_{i}}}}
		\delims{\left(}{\right)_{1\leq i \leq n}}
		\infer1[\footnotesize$\ruleManyVal$]{\tyjp{}{\la{\vartwo}{\tmtwo}}{\bigmplus_{i=1}^{n} \typctx_{i} , \var \hastype 
				\bigmplus_{i=1}^{n} \mtypetwo_i}{\bigmplus_{i=1}^{n} \mset{\larrow{\mtypethree_i}{\mtype_i}}}}
		\end{prooftree}
		\end{equation*}
		with $\sizem{\tderiv} = \sum_{i=1}^{n} (\sizem{\tderiv_{i}} + 1)$ and $\size{\tderiv} = \sum_{i=1}^n 
		(\size{\tderiv_i} + 1)$.
		Since $\vartwo \notin \fv{\val}$, then $\vartwo \notin \domain{\typctxtwo}$ (\refrmk{free-variables}), and so 
		$\namedtyjp{\tderivtwo}{}{\val}{\typctxtwo, \vartwo \hastype \emptymset}{\mtypetwo}$. 
		Now, there are two subcases:
		\begin{itemize}
			\item \emph{Empty multi type}: If $n = 0$,  then $\mtypetwo = \emptymset = \mtype$ and $\dom{\typctx} = 
			\emptyset$, with $\sizem{\tderiv} = 0 = \size{\tderiv}$. 
			According to \reflemmap{typing-value-splitting}{one} applied to $\tderivtwo$, $\dom{\typctxtwo} = \emptyset$ with 
			$\sizem{\tderivtwo} = 0 = \size{\tderivtwo}$.
			We can then build the derivation 
			\begin{equation*}
			\tderivthree = 
			\begin{prooftree}
			\infer0[\footnotesize$\ruleManyVal$]{\tyjp{}{\la{\vartwo}(\tmtwo\isub{\var}{\val})}{}{\emptymset}}
			\end{prooftree}
			\end{equation*}
			where $\sizem{\tderivthree} = 0 = \sizem{\tderiv} + \sizem{\tderivtwo}$ and $\size{\tderivthree} = 0 = 
			\size{\tderiv} + \size{\tderivtwo}$, and $\concl{\tderivthree}{\typctx \mplus 
				\typctxtwo}{\tm\isub{\var}{\val}}{\mtype}$ since $\dom{\typctx \mplus \typctxtwo} = \emptyset$.
			
			\item\emph{Non-empty multi type}: If $n > 0$, then we can decompose $\tderivtwo$ according to the partitioning 
			$\mtypetwo = \biguplus_{i=1}^n \mtypetwo_i$ by repeatedly applying \reflemmap{typing-value-splitting}{two}, and hence 
			for all $1 \leq i \leq n$ there are context $\typctxtwo_{i}$ and a derivation 
			$\namedtyjp{\tderivtwo_i}{}{\val}{\typctxtwo_{i} ; \vartwo \hastype \emptytype}{\mtypetwo_i}$ such that 
			$\sizem{\tderivtwo} = \sum_{i=1}^n \sizem{\tderivtwo_{i}}$ and $\size{\tderivtwo} = \sum_{i=1}^{n} 
			\size{\tderivtwo_{i}}$.
			By \ih, for all $1 \leq i \leq n$, there is a derivation 
			$\namedtyjp{\tderivthree_{i}}{}{\tmtwo\isub{\var}{\val}}{\typctx_i \mplus \typctxtwo_i, \vartwo \hastype 
				\mtypethree_i}{\mtype_i}$ such that $\sizem{\tderivthree_{i}} = \sizem{\tderiv_{i}} + \sizem{\tderivtwo_{i}}$ and 
			$\size{\tderivthree_{i}} \leq \size{\tderiv_{i}} + \size{\tderivtwo_{i}}$.
			Since $\typctx \mplus \typctxtwo = \bigmplus_{i=1}^n (\typctx_i \mplus \typctxtwo_i)$, we can build 
			$\tderivthree$ as
			\begin{equation*}
			\begin{prooftree}[separation=1em]
			\hypo{}
			\ellipsis{$\tderivthree_i$}{\typctx_i \mplus \typctxtwo_i, \vartwo \hastype \mtypethree_i \vdash \tmtwo 
				\isub{\var}{\val} \hastype \mtype_i}
			\infer1[\footnotesize$\ruleFun$]{\tyjp{}{\la{\vartwo}{(\tmtwo \isub{\var}{\val})}}{\typctx_{i} \mplus 
					\typctxtwo_{i}}{\ty{\mtypethree_{i}\!}{\!\mtype_{i}}}}
			\delims{\left(}{\right)_{1\leq i \leq n}}
			\hypo{}
			\infer2[\footnotesize$\ruleManyVal$]{\typctx \mplus \typctxtwo \vdash \la{\vartwo}{(\tmtwo \isub{\var}{\val})} 
				\hastype \mtype}
			\end{prooftree}
			\end{equation*}
			noting that $\sizem{\tderivthree} = \sum_{i=1}^n (\sizem{\tderivthree_{i}} + 1) = \sum_{i=1}^n 
			(\sizem{\tderiv_{i}} + \sizem{\tderivtwo_{i}} + 1) = \sum_{i=1}^{n} (\sizem{\tderiv_{i}} + 1) + \sum_{i=1}^{n} 
			\sizem{\tderivtwo_{i}} = \sizem{\tderiv} + \sizem{\tderivtwo} $ and that $\size{\tderivthree} = \sum_{i=1}^{n} 
			(\size{\tderivthree_{i}} + 1) \leq \sum_{i=1}^{n} (\size{\tderiv_{i}} + \size{\tderivtwo_{i}} + 1) = \sum_{i=1}^{n} 
			(\size{\tderiv_{i}} + 1) + \sum_{i=1}^{n} \size{\tderivtwo_{i}} = \size{\tderiv} + \size{\tderivtwo}$.
		\end{itemize}
		
		\item \emph{Explicit substitution}, \ie $\tm = \tmtwo \esub{\vartwo}{\tmthree}$. 
		We can suppose without loss of generality that $\vartwo \notin \fv{\val} \cup \{\var \}$, hence $\tm 
		\isub{\var}{\val} = \tmtwo\isub{\var}{\val} \esub{\vartwo}{\tmthree\isub{\var}\val}$ and necessarily
		\begin{equation*}
		\tderiv = 
		\begin{prooftree}
		\hypo{}
		\ellipsis{$\tderiv_{1}$}{\typctx_1 , \var \hastype \mtypetwo_1 , \vartwo \hastype \mtypethree \vdash \tmtwo 
			\hastype \mtype}
		\hypo{}
		\ellipsis{$\tderiv_{2}$}{\typctx_2, \var \hastype \mtypetwo_2 \vdash \tmthree \hastype \mtypethree}
		\infer2[\footnotesize$\Es$]{\typctx, \var \hastype \mtypetwo \vdash \tmtwo \esub{\vartwo}{\tmthree} \hastype \mtype}
		\end{prooftree}
		\end{equation*}
		with $\sizem{\tderiv} = \sizem{\tderiv_{1}} + \sizem{\tderiv_{2}}$, $\size{\tderiv} = \size{\tderiv_{1}} + 
		\size{\tderiv_{2}} + 1$, $\typctx = \typctx_1 \mplus \typctx_2$ and $\mtypetwo = \mtypetwo_2 \mplus \mtypetwo_2$. 
		According to \reflemmap{typing-value-splitting}{two} applied to $\tderivtwo$ and to the decomposition $\mtypetwo = 
		\mtypetwo_1 \mplus \mtypetwo_2$, there are contexts $\typctxtwo_1, \typctxtwo_2$ and derivations 
		$\namedtyjp{\tderivtwo_{1}}{}{\val}{\typctxtwo_{1}}{\mtypetwo_{1}}$ and 
		$\namedtyjp{\tderivtwo_{2}}{}{\val}{\typctxtwo_{2}}{\mtypetwo_{2}}$ such that $\typctxtwo = \typctxtwo_{1} \mplus 
		\typctxtwo_{2}$, $\sizem{\tderivtwo} = \sizem{\tderivtwo_{1}} + \sizem{\tderivtwo_{2}}$ and $\size{\tderivtwo} = 
		\size{\tderivtwo_{1}} + \size{\tderivtwo_{2}}$.
		
		By \ih, there are derivations $\namedtyjp{\tderivthree_{1}}{}{\tmtwo\isub{\var}{\val}}{\typctx_{1} \mplus 
			\typctxtwo_{1}, \vartwo \hastype \mtypethree}{\mtype}$ and $\namedtyjp{\tderivthree_{2}}{}{\tmthree 
			\isub{\var}{\val}}{\typctx_{2} \mplus \typctxtwo_{2}}{\mtypethree}$ such that $\sizem{\tderivthree_{i}} = 
		\sizem{\tderiv_{i}} + \sizem{\tderivtwo_{i}}$ and $\size{\tderivthree_{i}} \leq \size{\tderiv_{i}} + 
		\size{\tderivtwo_{i}}$ for all $i \in \{1,2\}$.
		Since $\typctx \mplus \typctxtwo = \typctx_1 \mplus \typctxtwo_1 \mplus \typctx_2 \mplus \typctxtwo_2$, we can 
		build the derivation
		\begin{equation*}
		\tderivthree = 
		\begin{prooftree}
		\hypo{}
		\ellipsis{$\tderivthree_1$}{\tyjp{}{\tmtwo\isub{\var}{\val}}{\typctx_{1} \mplus \typctxtwo_{1}, \vartwo \hastype 
				\mtypethree}{\mtype}}
		\hypo{}
		\ellipsis{$\tderivthree_2$}{\typctx_2 \mplus \typctxtwo_2 \vdash \tmthree\isub{\var}{\val} \hastype \mtypethree}
		\infer2[\footnotesize$\Es$]{\typctx \mplus \typctxtwo \vdash \tmtwo \isub{\var}{\val} \esub{\vartwo} {\tmthree 
				\isub{\var}{\val}} \hastype \mtype}
		\end{prooftree}
		\end{equation*}
		verifying that $\sizem{\tderivthree} = \sizem{\tderivthree_{1}} + \sizem{\tderivthree_{2}} = \sizem{\tderiv_{1}} + 
		\sizem{\tderivtwo_{1}} + \sizem{\tderiv_{2}} + \sizem{\tderivtwo_{2}} = \sizem{\tderiv} + \sizem{\tderivtwo}$ and 
		$\size{\tderivthree} = 1 + \size{\tderivthree_{1}} + \size{\tderivthree_{2}} \leq 1 + (\size{\tderiv_{1}} + 
		\size{\tderivtwo_{1}}) + (\size{\tderiv_{2}} + \size{\tderivtwo_{2}}) = \size{\tderiv} + \size{\tderivtwo}$.
		\qedhere
	\end{itemize}	
\end{proof}

\begin{lemma}[Typing of values: merging]
	\label{l:typing-value-complete} 
	Let $\val$ be a value.
	\begin{enumerate}
		\item \label{p:typing-value-complete-empty} There is a derivation $\namedtyjp{\tderiv}{}{\val}{}{\zero}$ with 
		$\sizem{\tderiv} = 0 = \size{\tderiv}$.
		
		\item \label{p:typing-value-complete-merge} For every derivations 
		$\namedtyjp{\tderiv_{1}}{}{\val}{\typctx_{1}}{\mtype_{1}}$ and 
		$\namedtyjp{\tderiv_{2}}{}{\val}{\typctx_{2}}{\mtype_{2}}$, there exists a  derivation 
		$\namedtyjp{\tderiv}{}{\val}{\typctx_{1} \mplus \typctx_2}{\mtype_{1} \mplus \mtype_{2}}$ such that $\sizem{\tderiv} = 
		\sizem{\tderiv_{1}} + \sizem{\tderiv_{2}}$ and $\size{\tderiv} = \size{\tderiv_{1}} + \size{\tderiv_{2}}$.
		
	\end{enumerate}
\end{lemma}

\begin{proof}
	\begin{enumerate}
		\item 
		Let $\tderiv$ be the following type derivation (applying the rule $\ruleMany$ with $n = 0$) such that 
		$\sizem{\tderiv} = 0 = \size{\tderiv}$:
		\begin{equation*}
		\tderiv =
		\begin{prooftree}
		\hypo{}
		\infer1[\footnotesize$\ruleMany$]{\tyjp{}{\val}{}{\emptytype}}
		\end{prooftree}\ .
		\end{equation*}
		
		\item Let 
		\begin{equation*}
		\tderiv_{1} =
		\begin{prooftree}
		\hypo{}
		\ellipsis{$\tderiv_{i}$}{\tyjp{}{\val}{\typctx_{i}}{\ltype_{i}}}
		\delims{ \left( }{ \right)_{\iI} }
		\infer1[\footnotesize$\ruleMany$]{\tyjp{}{\val}{\bigmplus_{\iI} \typctx_{i}}{\bigmplus_{\iI}\mult{\ltype_{i}}}}
		\end{prooftree}
		\end{equation*}
		with $\typctx_{1} = \bigmplus_{\iI} \typctx_{i}$ and $\mtype_{1} = \bigmplus_{\iI}\mult{\ltype_{i}}$, and let
		\begin{equation*}
		\tderiv_{2} =
		\begin{prooftree}
		\hypo{}
		\ellipsis{$\tderiv_{j}$}{\tyjp{}{\val}{\typctx_{j}}{\ltype_{j}}}
		\delims{ \left( }{ \right)_{\jJ} }
		\infer1[\footnotesize$\ruleMany$]{\tyjp{}{\val}{\bigmplus_{\jJ} \typctx_{j}}{\bigmplus_{\jJ}\mult{\ltype_{j}}}}
		\end{prooftree}
		\end{equation*}
		with $\typctx_{2} = \bigmplus_{\jJ} \typctx_{j}$ and $\mtype_{2} = \bigmplus_{\jJ}\mult{\ltype_{j}}$.
		
		We can derive $\tderiv$ by setting $K = I \cup J$ and then
		\begin{equation*}
		\tderiv =
		\begin{prooftree}
		\hypo{}
		\ellipsis{$\tderiv_{k}$}{\tyjp{}{\val}{\typctx_{k}}{\ltype_{k}}}
		\delims{ \left( }{ \right)_{\kK} }
		\infer1[\footnotesize$\ruleMany$]{\tyjp{}{\val}{\bigmplus_{\kK} \typctx_{k}}{\bigmplus_{\kK}\mult{\ltype_{k}}}}
		\end{prooftree}
		\end{equation*}
		trivially verifying the statement.
		\qedhere
	\end{enumerate}
\end{proof}

\begin{lemma}[Removal]
	\label{lappendix:anti-substitution}
	\NoteProof{l:anti-substitution}
	Let $\tm$ be a term, $\val$ be a value, and 
	$\namedtyjp{\tderiv}{}{\tm\isub{\var}{\val}}{\typctx}{\mtype}$
	be a type derivation. Then there are two  derivations $\namedtyjp{\tderivtwo}{}{\tm}{\typctxtwo, \var \hastype 
		\mtypetwo}{\mtype}$ 
	and $\namedtyjp{\tderivthree}{}{\val}{\typctxthree}{\mtypetwo}$ such that $\typctx = \typctxtwo \mplus \typctxthree$ 
	with $\sizem{\tderiv} = \sizem{\tderivtwo} + \sizem{\tderivthree}$ and $\size{\tderiv} \leq \size{\tderivtwo} + 
	\size{\tderivthree}$.
\end{lemma}

\begin{proof}
	By induction on the term $\tm$.
	Cases:
	\begin{itemize}
		\item \emph{Variable}, then are two sub-cases (let $\mtype = \mset{\ltype_1, \dots, \ltype_n}$ for some $n \in 
		\nat$):
		\begin{enumerate}
			\item $\tm = \var$, then $\tm \isub{\var}{\val} = \val$. 
			Let $\typctxthree = \typctx$, let $\typctxtwo$ be the empty context (\ie $\dom{\typctxtwo} = \emptyset$), let 
			$\mtypetwo = \mtype$ and let $\tderivtwo$ be the derivation 
			\begin{equation*}
			\tderivtwo = 
			\begin{prooftree}
			\infer0[\footnotesize$\Ax$]{\tyjp{}{\var}{\var \hastype \mset{\ltype_1}}{\ltype_1}}
			\hypo{\overset{n \in \nat}{\ldots}}
			\infer0[\footnotesize$\Ax$]{\tyjp{}{\var}{\var \hastype \mset{\ltype_n}}{\ltype_n}}
			\infer3[\footnotesize$\ruleManyVar$]{\tyjp{}{\var}{\var \hastype 
					\mset{\ltype_1,\dots,\ltype_n}}{\mset{\ltype_1,\dots,\ltype_n}}}
			\end{prooftree}
			\end{equation*}
			Thus, $\concl{\tderivtwo}{\typctxtwo, \var \hastype \mtypetwo}{\tm}{\mtype}$ with $\sizem{\tderivtwo} = 0$ and 
			$\size{\tderivtwo} = n$.
			Let $\tderivthree = \tderiv$: so, $\namedtyjp{\tderivthree}{}{\val}{\typctxthree}{\mtypetwo}$ and $\typctxthree 
			\mplus \typctxtwo = \typctx$ with $\sizem{\tderiv} = \sizem{\tderivthree} = \sizem{\tderivtwo} + \sizem{\tderivthree}$ 
			and $\size{\tderiv} = \size{\tderivthree} \leq \size{\tderivtwo} + \size{\tderivthree}$.
			
			\item $\tm = \varthree \neq \var$, then $\tm \isub{\var}{\val} = \varthree$ and 
			the derivation $\tderiv$ has necessarily the form (for some $n~\in~\nat$)
			\begin{equation*}
			\tderiv = 
			\begin{prooftree}
			\infer0[\footnotesize$\Ax$]{\tyjp{}{\varthree}{\varthree \hastype \mset{\ltype_1}}{\ltype_1}}
			\hypo{\overset{n \in \nat}{\ldots}}
			\infer0[\footnotesize$\Ax$]{\tyjp{}{\varthree}{\varthree \hastype \mset{\ltype_n}}{\ltype_n}}
			\infer3[\footnotesize$\ruleManyVar$]{\tyjp{}{\varthree}{\varthree \hastype 
					\mset{\ltype_1,\dots,\ltype_n}}{\mset{\ltype_1,\dots,\ltype_n}}}
			\end{prooftree}
			\end{equation*}
			where $\mtype = \mset{\ltype_1, \dots, \ltype_n}$ and $\typctx = \varthree \hastype \mtype$ (while $\typctx(\var) 
			= \emptymset$).
			Thus, $\sizem{\tderiv} = 0$ and $\size{\tderiv} = n$.
			Let $\typctxthree$ be the empty context (\ie $\dom{\typctxthree} = 0$) and  $\mtypetwo = \emptymset$.
			By \reflemmap{typing-value-complete}{empty}, there is a derivation 
			$\namedtyjp{\tderivthree}{}{\val}{}{\emptytype}$ (and hence 
			$\namedtyjp{\tderivthree}{}{\val}{\typctxthree}{\mtypetwo}$) such that $\sizem{\tderivthree} = 0 = 
			\size{\tderivthree}$.  
			Let $\typctxtwo = \typctx$  and $\tderivtwo = \tderiv$:
			therefore, $\typctxthree \mplus \typctxtwo = \typctx$
			and $\namedtyjp{\tderivtwo}{}{\tm }{\typctxtwo, \var \hastype \mtypetwo}{\mtype}$  with $\sizem{\tderiv} = 
			\sizem{\tderivtwo} = \sizem{\tderivtwo} + \sizem{\tderivthree}$ and $\size{\tderiv} = \size{\tderivtwo} \leq 
			\size{\tderivtwo} + \size{\tderivthree}$.
		\end{enumerate}
		
		\item \emph{Application}, \ie $\tm = \tm_1\tm_2$. 
		Then $\tm \isub{\var}{\val} = \tm_1 \isub{\var}{\val} \tm_2 \isub{\var}{\val}$ and necessarily
		\begin{equation*}
		\tderiv = 
		\begin{prooftree}
		\hypo{}
		\ellipsis{$\tderiv_{1}$}{\typctx_1 \vdash \tm_1\isub{\var}{\val} \hastype \mset{\larrow{\mtypethree}{\mtype}}}
		\hypo{}
		\ellipsis{$\tderiv_{2}$}{\typctx_2 \vdash \tm_2\isub{\var}{\val} \hastype \mtypethree}
		\infer2[\footnotesize$\ruleAp$]{\typctx \vdash \tm_1\isub{\var}{\val} \tm_2\isub{\var}{\val} \hastype \mtype}
		\end{prooftree}
		\end{equation*}
		with $\sizem{\tderiv} = \sizem{\tderiv_{1}} + \sizem{\tderiv_{2}} + 1$, $\size{\tderiv} = \size{\tderiv_{1}} + 
		\size{\tderiv_{2}} + 1$ and $\typctx = \typctx_1 \mplus \typctx_2$. 		
		By \ih, for all $i \in \{1,2\}$, there are derivations $\namedtyjp{\tderivtwo_i}{}{\tm_i}{\typctxtwo_i, \var 
			\hastype \mtypetwo_i}{\mult{\larrow{\mtypethree}{\mtype}}}$ and $\namedtyjp{\tderivthree_i}{}{\val 
		}{\typctxthree_i}{\mtypetwo_i}$ such that $\typctx_i = \typctxtwo_i \mplus \typctxthree_i$ with $\sizem{\tderiv_{i}} = 
		\sizem{\tderivtwo_{i}} + \sizem{\tderivthree_{i}}$ and $\size{\tderiv_{i}} \leq \size{\tderivtwo_{i}} + 
		\size{\tderivthree_{i}}$.
		According to \reflemmap{typing-value-complete}{merge} applied to $\tderivthree_1$ and $\tderivthree_2$, there is a 
		derivation $\namedtyjp{\tderivthree}{}{\val}{\typctxthree}{\mtypetwo}$ where $\mtypetwo = \mtypetwo_1 \mplus 
		\mtypetwo_2$ and $\typctxthree = \typctxthree_1 \mplus \typctxthree_2$, such that  $\sizem{\tderivtwo} = 
		\sizem{\tderivtwo_1} + \sizem{\tderivtwo_2}$ and $\size{\tderivtwo} = \size{\tderivtwo_{1}} + \size{\tderivtwo_{2}}$.
		We can build the derivation (where $\typctxtwo = \typctxtwo_1 \mplus \typctxtwo_2$)
		\begin{equation*}
		\tderivtwo = 
		\begin{prooftree}
		\hypo{}
		\ellipsis{$\tderivtwo_1$}{\typctxtwo_1, \var \hastype \mtypetwo_1 \vdash \tm_1 \hastype 
			\mset{\larrow{\mtypethree}{\mtype}}}
		\hypo{}
		\ellipsis{$\tderivtwo_2$}{\typctxtwo_2, \var \hastype \mtypetwo_2 \vdash \tm_2 \hastype \mtypethree}
		\infer2[\footnotesize$\ruleAp$]{\typctxtwo, \var \hastype \mtypetwo \vdash \tm \hastype \mtype}
		\end{prooftree}
		\end{equation*}
		with $\sizem{\tderiv} = \sizem{\tderiv_{1}} + \size{\tderiv_{2}} + 1 = \sizem{\tderivtwo_{1}} + 
		\sizem{\tderivthree_{1}} + \sizem{\tderivtwo_{2}} + \sizem{\tderivthree_{2}} + 1 = \sizem{\tderivtwo} + 
		\sizem{\tderivthree}$ 
		and
		$\size{\tderiv} = \size{\tderiv_{1}} + \size{\tderiv_{2}} + 1 \leq \sizem{\tderivtwo_{1}} + 
		\sizem{\tderivthree_{1}} + \sizem{\tderivtwo_{2}} + \size{\tderivthree_{2}} + 1 = \size{\tderivtwo} + 
		\size{\tderivthree}$.
		
		\item \emph{Abstraction}, \ie $\tm = \la{\vartwo}{\tmtwo}$.
		We can suppose without loss of generality that $\vartwo \notin \fv{\val} \cup \{\var \}$, hence $\tm 
		\isub{\var}{\val} = \la{\vartwo}{\tmtwo\isub{\var}\val}$ and necessarily, for some $n \in \nat$,
		\begin{equation*}
		\tderiv =
		\begin{prooftree}[separation=1em]
		\hypo{}
		\ellipsis{$\tderiv_i$}{\typctx_i, \vartwo \hastype \mtypethree_i \vdash \tmtwo\isub{\var}{\val} \hastype \mtype_i}
		
		\infer1[\footnotesize$\ruleFun$]{\tyjp{}{\la{\vartwo}{\tmtwo\isub{\var}{\val}}}{\typctx_{i}}{\ty{\mtypethree_{i}\!}{
					\!\mtype_{1}}}}
		\delims{\left(}{\right)_{1 \leq i \leq n}}
		\hypo{}
		\infer2[\footnotesize$\ruleManyVal$]{\tyjp{}{\la{\vartwo}{\tmtwo\isub{\var}{\val}}}{\typctx}{\mtype}}
		\end{prooftree}
		\end{equation*}
		with $\typctx = \bigmplus_{i=1}^{n} \typctx_i$, $\mtype = \bigmplus_{i=1}^{n} 
		\mset{\larrow{\mtypethree_i\!}{\!\mtype_i}}$, $\sizem{\tderiv} = n + \sum_{i=1}^{n} \sizem{\tderiv_{i}}$ and 
		$\size{\tderiv} = n + \sum_{i=1}^n \size{\tderiv_i} $.
		%
		$\namedtyjp{\tderivtwo}{}{\val}{\typctxtwo, \vartwo \hastype \emptymset}{\mtypetwo}$. 
		There are two subcases:
		\begin{itemize}
			\item \emph{Empty multi type}: If $n = 0$,  then $\mtype = \emptymset$ and $\dom{\typctx} = \emptyset$, with 
			$\sizem{\tderiv} = 0 = \size{\tderiv}$. 
			We can  build the derivation 
			\begin{equation*}
			\tderivtwo = 
			\begin{prooftree}
			\infer0[\footnotesize$\ruleManyVal$]{\tyjp{}{\la{\vartwo}\tmtwo}{}{\emptymset}}
			\end{prooftree}
			\end{equation*}
			where $\sizem{\tderivtwo} = 0 = \size{\tderivtwo}$.
			Let $\mtypetwo = \emptymset$ and $\typctxtwo$ be the empty context (\ie $\dom{\typctxtwo} = \emptyset$): then 
			$\concl{\tderivtwo}{\typctxtwo, \var \hastype \mtypetwo}{\tm}{\mtype}$.			
			According to \reflemmap{typing-value-complete}{empty}, there is a derivation 
			$\concl{\tderivthree}{}{\val}{\emptymset}$ with $\sizem{\tderivthree} = 0 = \size{\tderivthree}$. 
			Let $\typctxthree$ be the empty context (\ie $\dom{\typctxthree} = \emptyset$): so, 
			$\concl{\tderivthree}{\typctxthree}{\val}{\mtypetwo}$ with $\typctx = \typctxtwo \mplus \typctxthree$ and 
			$\sizem{\tderiv} = 0 = \sizem{\tderivtwo} + \sizem{\tderivthree}$ and $\size{\tderiv} = 0 \leq \size{\tderivtwo} + 
			\size{\tderivthree}$.
			
			\item\emph{Non-empty multi type}: If $n > 0$ then by \ih, for all $1 \leq i \leq n$, there are derivations 
			$\concl{\tderivtwo_i}{\typctxtwo_i, \vartwo \hastype \mtypethree_i, \var \hastype \mtypetwo_i}{\tmtwo}{\mtype_i}$ and 
			$\concl{\tderivthree_i}{\typctxthree_i}{\val}{\mtypetwo_i}$ such that $\typctx_i = \typctxtwo_i \mplus \typctxthree_i$ 
			with $\sizem{\tderiv_i} = \sizem{\tderivtwo_i} + \sizem{\tderivthree_i}$ and $\size{\tderiv_i} \leq \size{\tderivtwo_i} 
			+ \size{\tderivthree_i}$.
			We can build the derivation
			\begin{equation*}
			\tderivtwo = 
			\begin{prooftree}[separation=1em]
			\hypo{}
			\ellipsis{$\tderivtwo_i$}{\typctxtwo_i ; \vartwo \hastype \mtypethree_i ; \var \hastype \mtypetwo_i \vdash \tmtwo 
				\hastype \mtype_i}
			\infer1[\footnotesize$\ruleFun$]{\tyjp{}{\la{\vartwo}{\tmtwo}}{\typctxtwo_{i} ; \var \hastype 
					\mtypetwo_{i}}{\ty{\mtypethree_{i}\!}{\!\mtype_{i}}}}
			\delims{ \left( }{ \right)_{1 \leq i \leq n} }
			\infer1[\footnotesize$\ruleManyVal$]{\tyjp{}{\la{\vartwo}{\tmtwo}}{\bigmplus_{i=1}^n \typctxtwo_i ; \var \hastype 
					\mplus_{i=1}^n \mtypetwo_i}{\bigmplus_{i=1}^{n} \mult{\ty{\mtypethree_{i}\!}{\!\mtype_{i}}}}}
			\end{prooftree}
			\end{equation*}
			Thus, $\sizem{\tderivtwo} = n + \sum_{i=1}^{n} \sizem{\tderivtwo_{i}}$ and $\size{\tderivtwo} = n + \sum_{i=1}^n 
			\size{\tderivtwo_i} $.
			By repeatedly applying \reflemmap{typing-value-complete}{merge}, there is a derivation 
			$\concl{\tderivthree}{\typctxthree}{\val}{\mtypetwo}$ with $\typctxthree = \bigmplus_{i=1}^n \typctxthree_i$ such that 
			$\sizem{\tderivthree} = \sum_{i=1}^n\sizem{\tderivthree_i}$ and $\size{\tderivthree} = 
			\sum_{i=1}^n\size{\tderivthree_i}$.
			So, $\typctx = \bigmplus_{i=1}^n \typctx_i = \bigmplus_{i=1}^n (\typctxtwo_i \mplus \typctxthree_i) = \typctxtwo 
			\mplus \typctxthree$ with $\sizem{\tderiv} = n + \sum_{i=1}^n \sizem{\tderiv_i} = n + \sum_{i=1}^n 
			(\sizem{\tderivtwo_i} + \sizem{\tderivthree_i}) = \sizem{\tderivtwo} + \sizem{\tderivthree}$ 
			and $\size{\tderiv} = n + \sum_{i=1}^n \size{\tderiv_i} \leq n + \sum_{i=1}^n (\size{\tderivtwo_i} + 
			\size{\tderivthree_i}) = \size{\tderivtwo} + \size{\tderivthree}$.
		\end{itemize}
		
		\item \emph{Explicit substitution}, \ie $\tm = \tmtwo \esub{\vartwo}{\tmthree}$. 
		We can suppose without loss of generality that $\vartwo \notin \fv{\val} \cup \{\var \}$, hence $\tm 
		\isub{\var}{\val} = \tmtwo\isub{\var}{\val} \esub{\vartwo}{\tmthree\isub{\var}\val}$ and necessarily
		\begin{equation*}
		\tderiv = 
		\begin{prooftree}
		\hypo{}
		\ellipsis{$\tderiv_1$}{\tyjp{}{\tmtwo\isub{\var}{\val}}{\typctx_{1}, \vartwo \hastype \mtypethree}{\mtype}}
		\hypo{}
		\ellipsis{$\tderiv_2$}{\typctx_2 \vdash \tmthree\isub{\var}{\val} \hastype \mtypethree}
		\infer2[\footnotesize$\Es$]{\typctx \vdash \tmtwo \isub{\var}{\val} \esub{\vartwo} {\tmthree \isub{\var}{\val}} 
			\hastype \mtype}
		\end{prooftree}
		\end{equation*}
		with $\sizem{\tderiv} = \sizem{\tderiv_{1}} + \sizem{\tderiv_{2}}$, $\size{\tderiv} = \size{\tderiv_{1}} + 
		\size{\tderiv_{2}} + 1$ and $\typctx = \typctx_1 \mplus \typctx_2$. 
		By \ih applied to $\tderiv_1$ and \refrmk{free-variables}, there are derivations 
		$\concl{\tderivtwo_1}{\typctxtwo_1, \vartwo \hastype \mtypethree, \var \hastype \mtypetwo_1}{\tmtwo}{\mtype}$ and
		$\concl{\tderivthree_1}{\typctxthree_1}{\val}{\mtypetwo_1}$ with $\typctx_1 = \typctxtwo_1 \mplus \typctxthree_1$ 
		such that $\sizem{\tderiv_1} = \sizem{\tderivtwo_1} + \sizem{\tderivthree_1}$ and $\size{\tderiv_1} \leq 
		\size{\tderivtwo_1} + \size{\tderivthree_1}$.
		By \ih applied to $\tderiv_2$ , there are derivations $\concl{\tderivtwo_2}{\typctxtwo_2, \var \hastype 
			\mtypetwo_2}{\tmthree}{\mtype}$ and
		$\concl{\tderivthree_2}{\typctxthree_2}{\val}{\mtypetwo_2}$ with $\typctx_2 = \typctxtwo_2 \mplus \typctxthree_2$ 
		such that $\sizem{\tderiv_2} = \sizem{\tderivtwo_2} + \sizem{\tderivthree_2}$ and $\size{\tderiv_2} \leq 
		\size{\tderivtwo_2} + \size{\tderivthree_2}$.
		According to \reflemmap{typing-value-complete}{merge}, there is a derivation 
		$\namedtyjp{\tderivthree}{}{\val}{\typctxthree}{\mtypetwo}$ with $\typctxthree = \typctxthree_{1} \mplus 
		\typctxthree_{2}$ and $\mtypetwo = \mtypetwo_1 \mplus \mtypetwo_2$ such that $\sizem{\tderivthree} = 
		\sizem{\tderivthree_{1}} + \sizem{\tderivthree_{2}}$ and $\size{\tderivthree} = \size{\tderivthree_{1}} + 
		\size{\tderivthree_{2}}$.
		We can build the derivation (where $\typctxtwo = \typctxtwo_1 \mplus \typctxtwo_2$)
		\begin{equation*}
		\tderivtwo = 
		\begin{prooftree}
		\hypo{}
		\ellipsis{$\tderivtwo_{1}$}{\typctxtwo_1 , \var \hastype \mtypetwo_1 , \vartwo \hastype \mtypethree \vdash \tmtwo 
			\hastype \mtype}
		\hypo{}
		\ellipsis{$\tderivtwo_{2}$}{\typctxtwo_2, \var \hastype \mtypetwo_2 \vdash \tmthree \hastype \mtypethree}
		\infer2[\footnotesize$\Es$]{\typctxtwo, \var \hastype \mtypetwo \vdash \tmtwo \esub{\vartwo}{\tmthree} \hastype 
			\mtype}
		\end{prooftree}
		\end{equation*}
		verifying that $\typctx = \typctx_1 \mplus \typctx_2 = \typctxtwo_1 \mplus \typctxthree_1 \mplus \typctxtwo_2 
		\mplus \typctxthree_2 = \typctxtwo \mplus \typctxthree$ and $\sizem{\tderiv} = \sizem{\tderiv_{1}} + \sizem{\tderiv_{2}} 
		= \sizem{\tderivtwo_{1}} + \sizem{\tderivthree_{1}} + \sizem{\tderivtwo_{2}} + \sizem{\tderivthree_{2}} = 
		\sizem{\tderivtwo} + \sizem{\tderivthree}$ and $\size{\tderiv} = 1 + \size{\tderiv_{1}} + \size{\tderiv_{2}} \leq 1 + 
		(\size{\tderivtwo_{1}} + \size{\tderivthree_{1}}) + (\size{\tderivthree_{2}} + \size{\tderivthree_{2}}) = 
		\size{\tderivtwo} + \size{\tderivthree}$.
		\qedhere
	\end{itemize}	
\end{proof}

\section{Proofs of \Cref{sect:open} (Multi Types for Open \cbv)}

\begin{remark}[Merging and splitting inertness]
	\label{rmk:merge-split-inert}
	Let $\mtype, \mtypetwo, \mtypethree$ be multi types with $\mtype = \mtypetwo \mplus \mtypethree$; then, $\mtype$ is inert iff $\mtypetwo$ and $\mtypethree$ are inert.
	Similarly for type contexts.
\end{remark}

\begin{lemma}
	[Spreading of inertness on judgments]
	\label{lappendix:spread-inert}
	\NoteState{l:spread-inert}
	Let $\concl{\tderiv}{\typctx}{\itm}{\mtype}$ be a derivation and $\itm$ be an inert term. 
	If \,$\typctx$ is a inert type context, then $\mtype$ is a inert multi type.
\end{lemma}

\begin{proof}
	By induction on the definition of inert terms $\itm$ (see \Cref{rmk:inert-alternative}).
	Cases:
	\begin{itemize}
		\item \emph{Variable}, \ie $\itm = \var$. Then necessarily, for some $n \in \nat$,
		\begin{equation*}
		\tderiv = 
		\begin{prooftree}
		\infer0[\footnotesize$\ruleAx$]{\tyjp{}{\var}{\var \hastype \mset{\ltype_1}}{\ltype_1}}
		\hypo{\overset{n \in \nat}{\ldots}}
		\infer0[\footnotesize$\ruleAx$]{\tyjp{}{\var}{\var \hastype \mset{\ltype_n}}{\ltype_n}}
		\infer3[\footnotesize$\ruleManyVar$]{\tyjp{}{\var}{\typctx}{\mtype}}
		\end{prooftree}
		\end{equation*}
		where $\mtype = \mset{\ltype_1, \dots, \ltype_n}$ and $\typctx = \var \hastype \mtype$. 
		As $\typctx$ is inert (by hypothesis), so is $\mtype$.
		
		\item \emph{Application}, \ie $\itm = \itmtwo \fire$. Then necessarily
		\begin{equation*}
		\tderiv = 
		\begin{prooftree}
		\hypo{}
		\ellipsis{$\tderivtwo$}{\typctxtwo \vdash \itmtwo \hastype \mult{\ty{\mtypetwo}{\mtype}}}
		\hypo{}
		\ellipsis{$\tderivthree$}{\typctxthree \vdash \fire \hastype\mtypetwo}
		\infer2[\footnotesize$\ruleApp$]{\tyjp{}{\itmtwo \fire}{\typctxtwo \mplus \typctxthree}{\mtype}}
		\end{prooftree}
		\end{equation*}
		where $\typctx = \typctxtwo \mplus \typctxthree$.
		As $\typctx$ is inert, so is $\typctxtwo$, according to \Cref{rmk:merge-split-inert}.
		By \ih applied to $\tderivtwo$, $\mult{\ty{\mtypetwo}{\mtype}}$ is inert and hence $\mtype$ is inert.
		
		\item \emph{Explicit substitution on inert}, \ie $\itm = \itmtwo \esub{\var}{\pitm}$. Then necessarily
		\begin{equation*}
		\tderiv = 
		\begin{prooftree}
		\hypo{}
		\ellipsis{$\tderivtwo$}{\typctxtwo, \var \hastype \mtypetwo \vdash \itmtwo \hastype \mtype}
		\hypo{}
		\ellipsis{$\tderivthree$}{\typctxthree \vdash \pitm \hastype \mtypetwo}
		\infer2[\footnotesize$\ruleES$]{\tyjp{}{\itmtwo \esub{\var}{\pitm}}{\typctxtwo \mplus \typctxthree}{\mtype}}
		\end{prooftree}
		\end{equation*}
		where $\typctx = \typctxtwo \mplus \typctxthree$.
		As $\typctx$ is inert, so are $\typctxtwo$ and $\typctxthree$, according to \Cref{rmk:merge-split-inert}.
		By \ih applied to $\tderivthree$, $\mtypetwo$ is inert. 
		Therefore, $\typctxtwo, \var \hastype \mtypetwo$ is a inert type context.
		By \ih applied to $\tderivtwo$, the multi type $\mtype$ is inert.
		\qedhere
	\end{itemize}
\end{proof}

\paragraph*{Correctness}

\begin{lemma}[Size of fireballs]
	\label{lappendix:size-fireballs}
	\NoteState{l:size-fireballs}
	Let $\fire$ be a fireball. 
	If $\namedtyjp{\tderiv}{}{\fire}{\typctx}{\mtype}$
	then $\sizem{\tderiv} \geq \sizeo{\fire}$.
	If, moreover, $\typctx$ is inert and ($\mtype$ is ground inert or $\fire$ is inert), then 
$\sizem{\tderiv} = \sizeo{\fire}$.
\end{lemma}

\begin{proof}
	By mutual induction on the definition of fireballs $\fire$ and inert terms $\itm$ (see \Cref{rmk:inert-alternative}).
	Cases for inert terms:
	\begin{itemize}
		\item \emph{Variable}, \ie $\fire = \var$. Then necessarily, for some $n \in \nat$,
		\begin{equation*}
		\tderiv = 
		\begin{prooftree}
		\infer0[\footnotesize$\ruleAx$]{\tyjp{}{\var}{\var \hastype \mset{\ltype_1}}{\ltype_1}}
		\hypo{\overset{n \in \nat}{\ldots}}
		\infer0[\footnotesize$\ruleAx$]{\tyjp{}{\var}{\var \hastype \mset{\ltype_n}}{\ltype_n}}
		\infer3[\footnotesize$\ruleManyVar$]{\tyjp{}{\var}{\typctx}{\mtype}}
		\end{prooftree}
		\end{equation*}
		where $\mtype = \mset{\ltype_1, \dots, \ltype_n}$ and $\typctx = \var \hastype \mtype$. 
		Therefore, $\sizeo{\fire} = 0 = \sizem{\tderiv}$.
		
		\item \emph{Application}, \ie $\fire = \itm \firetwo$. Then necessarily
		\begin{equation*}
		\tderiv = 
		\begin{prooftree}
		\hypo{}
		\ellipsis{$\tderivtwo$}{\typctxtwo \vdash \itm \hastype \mult{\ty{\mtypetwo}{\mtype}}}
		\hypo{}
		\ellipsis{$\tderivthree$}{\typctxthree \vdash \firetwo \hastype\mtypetwo}
		\infer2[\footnotesize$\ruleApp$]{\tyjp{}{\itm \firetwo}{\typctxtwo \mplus \typctxthree}{\mtype}}
		\end{prooftree}
		\end{equation*}
		where $\typctx = \typctxtwo \mplus \typctxthree$.
		By \ih applied to both premises, $\sizem{\tderivtwo} \geq \size{\itm}$ and $\sizem{\tderivthree} \geq 
\size{\firetwo}$.
		Therefore, $\sizeo{\fire} = \sizeo{\itm} + \sizeo{\firetwo} + 1 \leq \sizem{\tderivtwo} + \sizem{\tderivthree} + 1 = 
\sizem{\tderiv}$.
		
		If, moreover, $\typctx$ is inert then so are $\typctxtwo$ and $\typctxthree$, according to 
\Cref{rmk:merge-split-inert}.
		By \ih for inert terms applied to $\tderivtwo$, $\sizem{\tderivtwo} = \sizeo{\itm}$.
		By spreading of inertness (\Cref{l:spread-inert}), $\mset{\larrow{\mtypetwo}{\mtype}}$ is inert, hence $\mtypetwo = \emptytype$.
		By \ih for fireballs applied to $\tderivthree$, $\sizem{\tderivthree} = \sizeo{\firetwo}$.
		Therefore, $\size{\fire} = \size{\itm} + \size{\firetwo} + 1 = \sizem{\tderivtwo} + \sizem{\tderivthree} + 1 = 
\sizem{\tderiv}$.
		
		\item \emph{Explicit substitution on inert}, \ie $\fire = \itm \esub{\var}{\pitm}$. Then necessarily
		\begin{equation*}
		\tderiv = 
		\begin{prooftree}
		\hypo{}
		\ellipsis{$\tderivtwo$}{\typctxtwo, \var \hastype \mtypetwo \vdash \itm \hastype \mtype}
		\hypo{}
		\ellipsis{$\tderivthree$}{\typctxthree \vdash \pitm \hastype \mtypetwo}
		\infer2[\footnotesize$\ruleES$]{\tyjp{}{\itm \esub{\var}{\pitm}}{\typctxtwo \mplus \typctxthree}{\mtype}}
		\end{prooftree}
		\end{equation*}
		where $\typctx = \typctxtwo \mplus \typctxthree$.
		We can then apply \ih to both premises: $\sizem{\tderivtwo} \geq \sizeo{\itm}$ and $\sizem{\tderivthree} \geq \sizeo{\pitm}$. 
		Therefore, $\sizeo{\sfire} = \sizeo{\sitm} + \sizeo{\pitm} \leq \sizem{\tderivtwo} + \sizem{\tderivthree} = \sizem{\tderiv}$.	
		
		If, moreover, $\typctx$ is inert then so are $\typctxtwo$ and $\typctxthree$, according to 
\Cref{rmk:merge-split-inert}.
		By spreading of inertness (\Cref{l:spread-inert}), $\mtypetwo$ is inert, hence $\typctxtwo, \var \hastype \mtypetwo$ is a inert type context.
		By \ih for inert terms applied to $\tderivtwo$ and $\tderivthree$, we have $\sizem{\tderivtwo} = \sizeo{\itm}$.
		and $\sizem{\tderivthree} = \sizeo{\pitm}$.
	 	So, $\sizeo{\fire} = \sizeo{\itm} + \sizeo{\pitm} = \sizem{\tderivtwo} + \sizem{\tderivthree} = \sizem{\tderiv}$.
		
	\end{itemize}

	Cases for fireballs that are not inert terms:
	\begin{itemize}
		\item \emph{Explicit substitution on fireball}, \ie $\fire = \firetwo \esub{\var}{\pitm}$. Then necessarily
		\begin{equation*}
		\tderiv = 
		\begin{prooftree}
		\hypo{}
		\ellipsis{$\tderivtwo$}{\typctxtwo, \var \hastype \mtypetwo \vdash \firetwo \hastype \mtype}
		\hypo{}
		\ellipsis{$\tderivthree$}{\typctxthree \vdash \pitm \hastype \mtypetwo}
		\infer2[\footnotesize$\ruleES$]{\tyjp{}{\firetwo \esub{\var}{\pitm}}{\typctxtwo \mplus \typctxthree}{\mtype}}
		\end{prooftree}
		\end{equation*}
		where $\typctx = \typctxtwo \mplus \typctxthree$.
		We can then apply \ih to both premises: $\sizem{\tderivtwo} \geq \sizeo{\firetwo}$ and $\sizem{\tderivthree} \geq \sizeo{\pitm}$. 
		Therefore, $\sizeo{\fire} = \sizeo{\firetwo} + \sizeo{\pitm} \leq \sizem{\tderivtwo} + \sizem{\tderivthree} = 
\sizem{\tderiv}$ 
		
		If, moreover, $\typctx$ is inert and $\mtype$ is ground inert, then $\typctxtwo$ and $\typctxthree$ are inert, according to \Cref{rmk:merge-split-inert}.
		By \ih for inert terms applied to $\tderivthree$, $\sizem{\tderivthree} = \sizeo{\pitm}$.
		By spreading of inertness (\Cref{l:spread-inert}), $\mtypetwo$ is inert, hence $\typctxtwo, \var \hastype \mtypetwo$ is a inert type context.
		By \ih for fireballs applied to $\tderivtwo$, $\sizem{\tderivtwo} = \sizeo{\firetwo}$.
		Therefore, $\sizeo{\fire} = \sizeo{\firetwo} + \sizeo{\pitm} = \sizem{\tderivtwo} + \sizem{\tderivthree} = 
\sizem{\tderiv}$.
		
		\item \emph{Abstraction}, \ie $\fire = \la{\var}{\tm}$. 
		Then necessarily, for some $n \in \nat$,
		\begin{equation*}
		\tderiv = 
		\begin{prooftree}[separation = 1em]
		\hypo{}
		\ellipsis{$\tderivtwo_1$}{\typctx_1, \var \hastype \mtypethree_1 \vdash \tm \hastype \mtypetwo_1}
		\infer1[\footnotesize$\ruleFun$]{\tyjp{}{\la{\var}{\tm}}{\typctx_1}{\ty{\mtypethree_1}{\mtypetwo_1}}}
		\hypo{\overset{n \in \nat}{\ldots}}
		\hypo{}
		\ellipsis{$\tderivtwo_n$}{\typctx_n, \var \hastype \mtypethree_n \vdash \tm \hastype \mtypetwo_n}
		\infer1[\footnotesize$\ruleFun$]{\tyjp{}{\la{\var}{\tm}}{\typctx_n}{\ty{\mtypethree_n}{\mtypetwo_n}}}
		\infer3[\footnotesize$\ruleManyVal$]{\tyjp{}{\la{\var}{\tm}}{\typctx}{\mtype}}
		\end{prooftree}
		\end{equation*}
		where $\mtype = \bigmplus_{i=1}^n\mset{\larrow{\mtypethree_i}{\mtypetwo_i}}$ and $\typctx = 
\bigmplus_{i=1}^n\typctx_i$. 
		Thus, $\sizeo{\fire} = 0 \leq \sum_{i=1}^n(\sizem{\tderivtwo_i} + 1) = \sizem{\tderiv}$.
		
		If, moreover, $\mtype$ is ground inert, then necessarily $\mtype = \emptytype$ and $n = 0$, hence $\tderiv$ consist of the rule $\ruleManyVal$ with $0$ premises.
		Therefore, $\sizeo{\fire}= 0 = \sizem{\tderiv}$. 
		\qedhere
	\end{itemize}
\end{proof}

\begin{proposition}[Open quantitative subject reduction]
	\label{propappendix:weak-subject-reduction}
	\NoteState{prop:weak-subject-reduction}
	Let $\namedtyjp{\tderiv}{}{\tm}{\typctx}{\mtype}$ be a derivation.
	\begin{enumerate}
		\item If $\tm \towm \tm'$ then there exists a derivation $\namedtyjp{\tderiv'}{}{\tm'}{\typctx}{\mtype}$ such that
		$\sizem{\tderiv'} = \sizem{\tderiv} - 2$ and $\size{\tderiv'} = \size{\tderiv} - 1$; 
		\item If $\tm \towe \tm'$ then there exists a derivation $\namedtyjp{\tderiv'}{}{\tm'}{\typctx}{\mtype}$ such that
		$\sizem{\tderiv'} = \sizem{\tderiv}$ and $\size{\tderiv'} < \size{\tderiv}$.
	\end{enumerate}
\end{proposition}

\begin{proof}
	By induction on the open evaluation context $\weakctx$ such that $\tm = \weakctxp{\tmtwo} \tow \weakctxp{\tmtwo'} = \tm'$ with $\tmtwo \rtom \tmtwo'$ or $\tmtwo' \rtoe \tmtwo'$. 
	Cases for $\weakctx$:
	\begin{itemize}
		\item \emph{Hole}, \ie $\weakctx = \ctxhole$.
		Then there are two sub-cases:
		\begin{enumerate}
			\item \emph{Multiplicative}, \ie $\tm = \subctxp{\la\var\tmtwo}\tmthree \rtom  \subctxp{\tmtwo \esub{\var}{\tmthree}} = \tm'$.
			Then $\tderiv$ has necessarily the form:
			\begin{equation*}
			\tderiv = 
			\begin{prooftree}[separation = 1em]
			\hypo{}
			\ellipsis{$\tderivtwo$}{\typctx', \var \hastype \mtypetwo \vdash \tmtwo \hastype \mtype}
			\infer1[\footnotesize$\ruleFun$]{\typctx' \vdash \la\var\tmtwo \hastype \larrow{\mtypetwo}{\mtype}}
			\infer1[\footnotesize$\ruleManyVal$]{\typctx' \vdash \la\var\tmtwo \hastype \mset{\larrow{\mtypetwo}{\mtype}}}
			\hypo{}
			\ellipsis{$\tderiv_1$}{\quad}
			\infer2[\footnotesize$\ruleES$]{}
			\ellipsis{}{\quad}
			\hypo{}
			\ellipsis{$\tderiv_n$}{\quad}
			\infer2[\footnotesize$\ruleES$]{\typctx \vdash \subctxp{\la\var\tmtwo} \hastype \mset{\larrow{\mtypetwo}{\mtype}}}
			\hypo{}
			\ellipsis{$\tderivthree$}{\typctxtwo\vdash\tmthree \hastype \mtypetwo}
			\infer2[\footnotesize$\ruleApp$]{\typctx \uplus \typctxtwo \vdash \subctxp{\la\var\tmtwo}\tmthree \hastype \mtype}
			\end{prooftree}
			\end{equation*}
			
			with $\sizem{\tderiv} = 2 + \sizem{\tderivtwo} + \sizem{\tderivthree} + \sum_{i=1}^{n} \sizem{\tderiv_{i}}$ and $\size{\tderiv} = 2 + n + \size{\tderivtwo} + \size{\tderivthree} + \sum_{i=1}^{n} \size{\tderiv_{i}}$.
			We can then build $\tderiv'$ as follows:
			\begin{equation*}
			\tderiv' = 
			\begin{prooftree}
			\hypo{}
			\ellipsis{$\tderivtwo$}{\tyjp{}{\tmtwo}{\typctx' ; \var \hastype \mtypetwo}{\mtype}}
			\hypo{}
			\ellipsis{$\tderivthree$}{\tyjp{}{\tmthree}{\typctxtwo}{\mtypetwo}}
			\infer2[\footnotesize$\ruleES$]{\typctx'\uplus\typctxtwo \vdash\tmtwo\esub\var\tmthree \hastype \mtype}
			\hypo{}
			\ellipsis{$\tderiv_1$}{\quad}
			\infer2[\footnotesize$\ruleES$]{}
			\ellipsis{}{\quad}
			\hypo{}
			\ellipsis{$\tderiv_n$}{\quad}
			\infer2[\footnotesize$\Es$]{\typctx\uplus\typctxtwo \vdash \subctxp{\tmtwo\esub\var\tmthree} \hastype \mtype}
			\end{prooftree}
			\end{equation*}
			where $\sizem{\tderiv'} = \sizem{\tderivtwo} + \sizem{\tderivthree} + \sum_{i=1}^{n} \sizem{\tderiv_{i}} = \sizem{\tderiv} - 2$ and $\size{\tderiv'} = 1+ n + \size{\tderivtwo} + \size{\tderivthree} + \sum_{i=1}^{n} \size{\tderiv_{i}} = \size{\tderiv} - 1$.
			
			\item \emph{Exponential}, \ie $\tm = \tmtwo\esub\var{\subctxp{\val}} \rtoe \subctxp{\tmtwo \isub{\var}{\val}} = \tmp$.
			Then the derivation $\tderiv$ has necessarily the form:
			\begin{equation*}
			\tderiv = 
			\begin{prooftree}
			\hypo{}
			\ellipsis{$\tderivtwo$}{\tyjp{}{\tmtwo}{\typctxtwo, \var \hastype \mtypetwo}{\mtype}}
			\hypo{}
			\ellipsis{$\tderivthree$}{\tyjp{}{\val}{\typctxthree'}{\mtypetwo}}
			\hypo{}
			\ellipsis{$\tderiv_1$}{\quad}
			\infer2[\footnotesize$\Es$]{}
			\ellipsis{}{\quad}
			\hypo{}
			\ellipsis{$\tderiv_n$}{\quad}
			\infer2[\footnotesize$\Es$]{\typctxthree \vdash \subctxp{\val} \hastype \mtypetwo}
			\infer2[\footnotesize$\Es$]{\typctxtwo\uplus\typctxthree \vdash\tmtwo\esub\var{\subctxp{\val}}\hastype \mtype}
			\end{prooftree}
			\end{equation*}
			where $\typctx = \typctxtwo \uplus \typctxthree$, $\sizem{\tderiv} = \sizem{\tderivtwo} + \sizem{\tderivthree} + \sum_{i=1}^{n} \sizem{\tderiv_{i}}$ and $\size{\tderiv} = 1 + n + \size{\tderivtwo} + \size{\tderivthree} + \sum_{i=1}^{n} \size{\tderiv_{i}}$.
			By the substitution lemma (\reflemma{substitution}), there is a derivation $\namedtyjp{\tderiv''}{}{\tmtwo\isub{\var}{\val}}{\typctxtwo \mplus \typctxthree'}{\mtype}$
			such that $\sizem{\tderiv''} = \sizem{\tderivtwo} + \sizem{\tderivthree}$ and $\size{\tderiv''} \leq \size{\tderivtwo} + \size{\tderivthree}$.
			We can then build the following derivation $\tderiv'$:
			\begin{equation*}
			\tderiv' = 
			\begin{prooftree}
			\hypo{}
			\ellipsis{$\tderiv''$}{\typctxtwo \mplus\typctxthree' \vdash \tmtwo\isub\var \val \hastype \mtype}
			\hypo{}
			\ellipsis{$\tderiv_1$}{\quad}
			\infer2[\footnotesize{$\Es$}]{}
			\ellipsis{}{}
			\hypo{}
			\ellipsis{$\tderiv_n$}{\quad}
			\infer2[\footnotesize$\Es$]{\typctxtwo\mplus\typctxthree \vdash \subctxp{\tmtwo\isub\var \val} \hastype \mtype}
			\end{prooftree}
			\end{equation*}
			where $\sizem{\tderiv'} = \sizem{\tderiv''} + \sum_{i=1}^{n} \sizem{\tderiv_{i}} = \sizem{\tderivtwo} + \sizem{\tderivthree} + \sum_{i=1}^{n} \sizem{\tderiv_{i}} = \sizem{\tderiv}$ and $\size{\tderiv'} = n + \size{\tderiv''} + \sum_{i=1}^{n} \size{\tderiv_{i}} \leq n + \size{\tderivtwo} + \size{\tderivthree} + \sum_{i=1}^{n} \size{\tderiv_{i}} < 1 + n + \size{\tderivtwo} + \size{\tderivthree} + \sum_{i=1}^{n} \size{\tderiv_{i}} = \size{\tderiv}$ ($\tderiv'$ contains at least one rule $\Es$ less than $\tderiv$).
		\end{enumerate}
		
		\item \emph{Application left}, \ie $\weakctx = \weakctxtwo\tmthree$.
		Then, $\tm = \weakctxp{\tmtwo} = \weakctxtwop{\tmtwo} \tmthree \rootRew{a} \weakctxtwop{\tmtwo'} \tmthree = \weakctxp{\tmtwo'} = \tm'$ with $\tmtwo \rootRew{a} \tmtwo'$ and $a \in \{\msym, \esym\}$.
		The derivation $\tderiv$ is necessarily
		\begin{equation*}
		\tderiv = 
		\begin{prooftree}
		\hypo{}
		\ellipsis{$\tderivtwo$}{\tyjp{}{\weakctxtwop{\tmtwo}}{\typctxtwo}{\mult{\ty{\mtypetwo}{\mtype}}}}
		\hypo{}
		\ellipsis{$\tderivthree$}{\tyjp{}{\tmthree}{\typctxthree}{\mtypetwo}}
		\infer2[\footnotesize$\ruleAp$]{\tyjp{}{\weakctxtwop{\tmtwo} \tmthree}{\typctxtwo \mplus \typctxthree}{\mtype}}
		\end{prooftree}
		\end{equation*}
		where $\typctx = \typctxtwo \uplus \typctxthree$, $\sizem{\tderiv} = 1 + \sizem{\tderivtwo} + \sizem{\tderivthree}$ and $\size{\tderiv} = 1 + \size{\tderivtwo} + \size{\tderivthree}$.
		By \ih, there is a derivation $\namedtyjp{\tderivtwo'}{}{\weakctxtwop{\tmtwo'}}{\typctxtwo}{\mult{\ty{\mtypetwo}{\mtype}}}$ with: 
		\begin{enumerate}
			\item $\sizem{\tderivtwo'} = \sizem{\tderivtwo} - 2$ and $\size{\tderivtwo'} = \size{\tderivtwo} - 1$ if $\tmtwo \rtom \tmtwo'$; 
			\item $\sizem{\tderivtwo'} = \sizem{\tderivtwo}$ and $\size{\tderivtwo'} < \size{\tderivtwo}$ if $\tmtwo \rtoe \tmtwo'$.
		\end{enumerate}
		We can then build the derivation 
		\begin{equation*}
		\tderiv' = 
		\begin{prooftree}
		\hypo{}
		\ellipsis{$\tderivtwo'$}{\tyjp{}{\weakctxtwop{\tmtwo'}}{\typctxtwo}{\mult{\ty{\mtypetwo}{\mtype}}}}
		\hypo{}
		\ellipsis{$\tderivthree$}{\tyjp{}{\tmthree}{\typctxthree}{\mtypetwo}}
		\infer2[\footnotesize$\ruleAp$]{\tyjp{}{\weakctxtwop{\tmtwo'} \tmthree}{\typctxtwo \mplus \typctxthree}{\mtype}}
		\end{prooftree}
		\end{equation*}
		noting that
		\begin{enumerate}
			\item If $\tmtwo \rtom \tmtwo'$ then $\sizem{\tderiv'} = 1 + \sizem{\tderivtwo'} + \sizem{\tderivthree} = 1 + (\sizem{\tderivtwo} - 2) + \sizem{\tderivthree} = \sizem{\tderiv} - 2$ and $\size{\tderiv'} = 1 + \size{\tderivtwo'} + \size{\tderivthree} = 1 + (\size{\tderivtwo} - 1) + \size{\tderivthree} = \size{\tderiv} - 1$; 
			\item If $\tmtwo \rtoe \tmtwo'$ then $\sizem{\tderiv'} = 1 + \sizem{\tderivtwo'} + \sizem{\tderivthree} = 1 + \sizem{\tderivtwo} + \sizem{\tderivthree} = \sizem{\tderiv}$ and $\size{\tderiv'} = 1 + \size{\tderivtwo'} + \size{\tderivthree} < 1 + \size{\tderivtwo} + \size{\tderivthree} = \size{\tderiv}$.
		\end{enumerate}
		
		\item \emph{Application right}, \ie $\weakctx = \tmthree \weakctxtwo$.
		Analogous to the previous case.
		
		\item \emph{Explicit substitution left}, \ie $\weakctx = \weakctxtwo\esub{\var}{\tmthree}$. 
		Then, $\tm = \weakctxp{\tmtwo} = \weakctxtwop{\tmtwo} \esub{\var}{\tmthree} \rootRew{a} \weakctxtwop{\tmtwo'}\esub{\var}{\tmthree} = \weakctxp{\tmtwo'} = \tm'$ with $\tmtwo \rootRew{a} \tmtwo'$ and $a \in \{\msym, \esym\}$.
		The derivation $\tderiv$ is necessarily
		\begin{equation*}
		\tderiv = 
		\begin{prooftree}
		\hypo{}
		\ellipsis{$\tderivtwo$}{\tyjp{}{\weakctxtwop{\tmtwo}}{\typctxtwo ; \var \hastype \mtypetwo}{\mtype}}
		\hypo{}
		\ellipsis{$\tderivthree$}{\tyjp{}{\tmthree}{\typctxthree}{\mtypetwo}}
		\infer2[\footnotesize$\Es$]{\tyjp{}{\weakctxtwop{\tmtwo} \esub{\var}{\tmthree}}{\typctxtwo \mplus \typctxthree}{\mtype}}
		\end{prooftree}
		\end{equation*}
		where $\typctx = \typctxtwo \uplus \typctxthree$, $\sizem{\tderiv} = \sizem{\tderivtwo} + \sizem{\tderivthree}$ and $\size{\tderiv} = 1 + \size{\tderivtwo} + \size{\tderivthree}$.
		By \ih, there is a derivation $\namedtyjp{\tderivtwo'}{}{\weakctxtwop{\tmtwo'}}{\typctxtwo, \var \hastype \mtypetwo}{\mtype}$ with: 
		\begin{enumerate}
			\item $\sizem{\tderivtwo'} = \sizem{\tderivtwo} - 2$ and $\size{\tderivtwo'} = \size{\tderivtwo} - 1$ if $\tmtwo \rtom \tmtwo'$; 
			\item $\sizem{\tderivtwo'} = \sizem{\tderivtwo}$ and $\size{\tderivtwo'} < \size{\tderivtwo}$ if $\tmtwo \rtoe \tmtwo'$.
		\end{enumerate}
		We can then build the derivation 
		\begin{equation*}
		\tderiv' = 
		\begin{prooftree}
		\hypo{}
		\ellipsis{$\tderivtwo'$}{\tyjp{}{\weakctxtwop{\tmtwo'}}{\typctxtwo ; \var \hastype \mtypetwo}{\mtype}}
		\hypo{}
		\ellipsis{$\tderivthree$}{\tyjp{}{\tmthree}{\typctxthree}{\mtypetwo}}
		\infer2[\footnotesize$\Es$]{\typctxtwo \uplus \typctxthree \vdash \weakctxtwop{\tmtwo'} \esub{\var}{\tmthree} \hastype \mtype}
		\end{prooftree}
		\end{equation*}
		noting that 
		\begin{enumerate}
			\item If $\tmtwo \rtom \tmtwo'$ then $\sizem{\tderiv'} = \sizem{\tderivtwo'} + \sizem{\tderivthree} = (\sizem{\tderivtwo} - 2) + \sizem{\tderivthree} = \sizem{\tderiv} - 2$ and $\size{\tderiv'} = \size{\tderivtwo'} + \size{\tderivthree} = (\size{\tderivtwo} - 1) + \size{\tderivthree} = \size{\tderiv} - 1$; 
			\item If $\tmtwo \rtoe \tmtwo'$ then $\sizem{\tderiv'} = \sizem{\tderivtwo'} + \sizem{\tderivthree} = \sizem{\tderivtwo} + \sizem{\tderivthree} = \sizem{\tderiv}$ and $\size{\tderiv'} = \size{\tderivtwo'} + \size{\tderivthree} < \size{\tderivtwo} + \size{\tderivthree} = \size{\tderiv}$.
		\end{enumerate}
		
		\item \emph{Explicit substitution right}, \ie $\weakctx = \tmthree \esub{\var}{\weakctxtwo}$. 
		Analogous to the previous case.
		\qedhere
	\end{itemize}
\end{proof}

\begin{theorem}[Open correctness]
	\label{thmappendix:open-correctness}
	\NoteState{thm:open-correctness}
	Let $\derive{\tderiv}{\tm}$ be a derivation.
	Then there is an $\osym$-normalizing evaluation $\deriv \colon \tm \tovsubo^* \tmtwo$ with $2\sizem{\deriv} + \sizeo{\tmtwo} \leq \sizem{\tderiv}$.
And if $\tderiv$ is tight, then $2\sizem{\deriv} + \sizeo{\fire} = \sizem{\tderiv}$.
\end{theorem}

\begin{proof}
	Given the derivation (\resp tight derivation) $\concl{\tderiv}{\typctx}{\tm}{\mtype}$, we proceed by induction on the general size $\size{\tderiv}$ of $\tderiv$.
	
	If $\tm$ is normal for $\tovsubo$, then $\tm = \fire$ is a fireball.
	Let $\deriv$ be the empty evaluation (so $\sizem{\deriv} = 0$), thus $\sizem{\tderiv} \geq \sizeo{\fire} = \sizeo{\fire} + 2\sizem{\deriv}$ 
	(\resp $\sizem{\tderiv} = \sizeo{\fire} = \sizeo{\fire} + 2\sizem{\deriv}$) by \reflemma{size-fireballs}.
	
	Otherwise, $\tm$ is not normal for $\tomo$ and so $\tm \tovsubo \tmtwo$.
	According to open subject reduction (\Cref{prop:weak-subject-reduction}), there is a derivation $\concl{\tderivtwo}{\typctx}{\tmtwo}{\mtype}$ such that $\size{\tderivtwo} < \size{\tderiv}$  and 
	\begin{itemize}
		\item $\sizem{\tderivtwo} \leq \sizem{\tderiv} - 2$ (\resp $\sizem{\tderivtwo} = \sizem{\tderiv} - 2$) if $\tm \tomo \tmtwo$,
		\item $\sizem{\tderivtwo} = \sizem{\tderiv}$ if $\tm \toeo \tmtwo$.
	\end{itemize}
	By \ih, there exists a fireball $\fire$ and a reduction sequence $\deriv' \colon \tmtwo \tovsubo^* \fire$ with 
	$2\sizem{\deriv'} + \sizeo{\fire} \leq \sizem{\tderivtwo}$ (\resp $2\sizem{\deriv'} + \sizeo{\fire} = \sizem{\tderivtwo}$).
	Let $\deriv$ be the $\osym$-evaluation obtained by concatenating the first step $\tm \tovsubo \tmtwo$ and $\deriv'$.
	There are two cases:
	\begin{itemize}
		\item \emph{Multiplicative:} if $\tm \tomo \tmtwo$ then $\sizem{\tderiv} \geq \sizem{\tderivtwo} + 2 \geq \sizeo{\fire} 
		+ 2\sizem{\deriv'} + 2 = \sizeo{\fire} + 2\sizem{\deriv}$ (\resp $\sizem{\tderiv} = \sizem{\tderivtwo} + 2 = \sizeo{\fire} + 
		2\sizem{\deriv'} + 2 = \sizeo{\fire} + 2\size{\deriv}$), since $\sizem{\deriv} = \sizem{\deriv'} + 1$.
		\item \emph{Exponential:} if $\tm \toeo \tmtwo$ then $\sizem{\tderiv} = \sizem{\tderivtwo} \geq \sizeo{\fire} + 2\sizem{\deriv'} = \sizeo{\fire} + 2\sizem{\deriv}$ (\resp $\sizem{\tderiv} = \sizem{\tderivtwo} = \sizeo{\fire} + 2\sizem{\deriv'} = \sizeo{\fire} + 2\sizem{\deriv}$),  since $\sizem{\deriv} = \sizem{\deriv'}$.
		\qedhere
	\end{itemize} 
\end{proof}

\paragraph*{Completeness}

\begin{proposition}[Tight typability of open normal forms]
	\label{propappendix:precise-open-typability-nf}
	\NoteState{prop:precise-open-typability-nf}
	\begin{enumerate}
		\item \emph{Inert:}\label{pappendix:precise-open-typability-nf-inert} if $\tm$ is an inert term then, for any inert multi type $\mtype$, there is an inert
		type derivation $\concl{\tderiv}{\typctx}{\tm}{\mtype}$.
		\item \emph{Fireball:}\label{pappendix:precise-open-typability-nf-fireball} if $\tm$ is a fireball then there is a tight derivation $\concl{\tderiv}{\typctx}{\tm}{\emptytype}$.		
	\end{enumerate}
\end{proposition}

\begin{proof}
	We prove simultaneously \Cref{p:precise-open-typability-nf-fireball,p:precise-open-typability-nf-inert} by 
		mutual induction on the definition of fireball and inert term (see \Cref{rmk:inert-alternative}).
		Note that \Cref{p:precise-open-typability-nf-inert} is stronger than \Cref{p:precise-open-typability-nf-fireball}.
		Cases for inert terms:
		\begin{itemize}
			\item \emph{Variable}, \ie $\tm = \var$, which is an inert term. 
			Let $\imtype$ be an inert multi type: hence, $\imtype = \mset{\ltype_1, \dots, \ltype_n}$ for some $n \in \nat$ and some $\ltype_1, \dots, \ltype_n$ inert linear types.
			We can then build  the derivation 
			\begin{equation*}
			\tderiv = 
			\begin{prooftree}[separation = 1em]
			\infer0[\footnotesize$\ruleAx$]{\var \hastype \mset{\ltype_1} \vdash \var \hastype \ltype_1}
			\hypo{\dots}
			\infer0[\footnotesize$\ruleAx$]{\var \hastype \mset{\ltype_n} \vdash \var \hastype \ltype_n}
			\infer3[\footnotesize$\ruleManyVar$]{\var \hastype \mset{\ltype_1, \dots, \ltype_n} \vdash \var \hastype \mset{\ltype_1, \dots, \ltype_n}}
			\end{prooftree}
			\end{equation*}
			where $\Gamma = \var \hastype \mset{\ltype_1, \dots, \ltype_n}$ is an inert type context.

			\item \emph{Inert application}, \ie $\tm = \itm \fire$ for some inert term $\itm$ and fireball $\fire$.
			Let $\imtype$ be a multi type.
			By \ih for fireballs, there is a derivation $\concl{\tderivthree}{\typctxthree}{\fire}{\emptytype}$  for some inert type context $\typctxthree$.
			By \ih for inert terms, since $\mset{\larrow{\emptytype}{\imtype}}$ is an inert multi type, there is a derivation $\concl{\tderivtwo}{\typctxtwo}{\itm}{\mset{\larrow{\emptytype}{\imtype}}}$ for some inert type context $\typctxtwo$.
			We can then build the derivation 
			\begin{equation*}
			\tderiv  =
			\begin{prooftree}
			\hypo{}
			\ellipsis{$\tderivtwo$}{\typctxtwo \vdash \itm \hastype \mset{\larrow{\emptytype}{\imtype}}}
			\hypo{}
			\ellipsis{$\tderivthree$}{\typctxthree \vdash \fire \hastype \emptytype}
			\infer2[\footnotesize$\ruleApp$]{\typctxtwo \mplus \typctxthree \vdash \itm \fire \hastype \imtype}				
			\end{prooftree}
			\end{equation*}
			where $\typctx = \typctxtwo \mplus \typctxthree$ is an inert type context, by \Cref{rmk:merge-split-inert}.
			
			\item \emph{Explicit substitution on inert}, \ie $\tm = \itm \esub{\var}{\pitm}$ for some inert term $\itm$ and proper inert term $\pitm$.
			Let $\imtype$ be an inert multi type.
			By \ih for inert terms applied to $\itm$, there is a derivation $\concl{\tderivtwo}{\typctxtwo, \var \hastype \mtypetwo}{\itm}{\imtype}$ for some inert multi type $\mtypetwo$ and inert type context $\typctxtwo$.
			By \ih for inert terms applied to $\pitm$, there is a derivation $\concl{\tderivthree}{\typctxthree}{\pitm}{\mtypetwo}$  for some inert type context $\typctxthree$.
			We can then build the derivation 
			\begin{equation*}
				\tderiv =
				\begin{prooftree}
				\hypo{}
				\ellipsis{$\tderivtwo$}{\typctxtwo, \var \hastype \mtypetwo \vdash \itm \hastype \imtype}
				\hypo{}
				\ellipsis{$\tderivthree$}{\typctxthree \vdash \pitm \hastype \mtypetwo}
				\infer2[\footnotesize$\ruleES$]{\typctxtwo \mplus \typctxthree \vdash \itm \esub{\var} {\pitm} \hastype \imtype}				
				\end{prooftree}
			\end{equation*}
			where $\typctx = \typctxtwo \mplus \typctxthree$ is an inert type context, by \Cref{rmk:merge-split-inert}.
		\end{itemize}
	
		Cases for fireballs that may not be inert terms:
		\begin{itemize}
			\item \emph{Abstraction}, \ie $\tm  = \la{\var}{\tmtwo}$.
			We can then build the derivation
			\begin{equation*}
				\tderiv = 
				\begin{prooftree}
				\infer0[\footnotesize$\ruleManyVal$]{\vdash \la{\var}{\tmtwo} \hastype \emptytype}
			\end{prooftree}
			\end{equation*}
			where the type context $\typctx$ is empty and hence inert, thus $\tderiv$ is tight.
			
			\item \emph{Explicit substitution on fireball}, \ie $\tm = \fire \esub{\var}{\pitm}$ for some fireball $\fire$ and proper inert term $\pitm$.
			By \ih for fireballs applied to $\fire$, there is a derivation $\concl{\tderivtwo}{\typctxtwo, \var \hastype \imtypetwo}{\fire}{\emptytype}$ for some inert multi type $\imtypetwo$ and inert type context $\typctxtwo$.
			By \ih for inert terms applied to $\pitm$, there is a derivation $\concl{\tderivthree}{\typctxthree}{\pitm}{\imtypetwo}$  for some inert type context $\typctxthree$.
			We can build the derivation 
			\begin{equation*}
			\tderiv = 
			\begin{prooftree}
			\hypo{}
			\ellipsis{$\tderivtwo$}{\typctxtwo, \var \hastype \imtypetwo \vdash \fire \hastype \emptytype}
			\hypo{}
			\ellipsis{$\tderivthree$}{\typctxthree \vdash \pitm \hastype \imtypetwo}
			\infer2[\footnotesize$\ruleES$]{\typctxtwo \mplus \typctxthree \vdash \fire \esub{\var} {\pitm} \hastype \emptytype}				
			\end{prooftree}
			\end{equation*}
			where $\typctx = \typctxtwo \mplus \typctxthree$ is an inert type context, by \Cref{rmk:merge-split-inert}.
			Therefore, $\tderiv$ is tight.
			\qedhere
		\end{itemize}
\end{proof}

\begin{proposition}[Open quantitative subject expansion]
	\label{propappendix:weak-subject-expansion}
	\NoteState{prop:weak-subject-expansion}
	Let $\namedtyjp{\tderiv'}{}{\tm'}{\typctx}{\mtype}$ be a derivation.
	\begin{enumerate}
		\item \emph{Multiplicative step:} if $\tm \towm \tm'$ then there is a derivation $\namedtyjp{\tderiv}{}{\tm}{\typctx}{\mtype}$ with
		$\sizem{\tderiv'} = \sizem{\tderiv} - 2$ and $\size{\tderiv'} = \size{\tderiv} - 1$; 
		\item \emph{Exponential step:} if $\tm \towe \tm'$ then there is a derivation $\namedtyjp{\tderiv}{}{\tm}{\typctx}{\mtype}$ such that
		$\sizem{\tderiv'} = \sizem{\tderiv}$ and $\size{\tderiv'} < \size{\tderiv}$.
	\end{enumerate}
\end{proposition}

\begin{proof}
	By induction on 
	the evaluation context $\weakctx$ in the step $\tm = \weakctxp{\tmtwo} \tow \weakctxp{\tmtwo'} = \tm'$ with $\tmtwo \rtom \tmtwo'$ or $\tmtwo \rtoe \tmtwo'$. 
	Cases for $\weakctx$:
	\begin{itemize}
		\item \emph{Hole}, \ie $\weakctx = \ctxhole$.
		Then there are two sub-cases:
		\begin{enumerate}
			\item \emph{Multiplicative}, \ie $\tm = \subctxp{\la\var\tmtwo}\tmthree \rtom  \subctxp{\tmtwo \esub{\var}{\tmthree}} = \tm'$.
			Then $\tderiv'$ has necessarily the form:
			\begin{equation*}
			\tderiv' = 
			\begin{prooftree}
			\hypo{}
			\ellipsis{$\tderivtwo$}{\tyjp{}{\tmtwo}{\typctx' ; \var \hastype \mtypetwo}{\mtype}}
			\hypo{}
			\ellipsis{$\tderivthree$}{\tyjp{}{\tmthree}{\typctxtwo}{\mtypetwo}}
			\infer2[\footnotesize$\ruleES$]{\typctx'\uplus\typctxtwo \vdash\tmtwo\esub\var\tmthree \hastype \mtype}
			\hypo{}
			\ellipsis{$\tderiv_1$}{\quad}
			\infer2[\footnotesize$\ruleES$]{}
			\ellipsis{}{\quad}
			\hypo{}
			\ellipsis{$\tderiv_n$}{\quad}
			\infer2[\footnotesize$\Es$]{\typctx\uplus\typctxtwo \vdash \subctxp{\tmtwo\esub\var\tmthree} \hastype \mtype}
			\end{prooftree}
			\end{equation*}
			with $\sizem{\tderiv'} = \sizem{\tderivtwo} + \sizem{\tderivthree} + \sum_{i=1}^{n} \sizem{\tderiv_{i}}$ and $\size{\tderiv' } = 1 + n + \size{\tderivtwo} + \size{\tderivthree} + \sum_{i=1}^{n} \size{\tderiv_{i}}$.
			
			We can build $\tderiv$ as follows:
			\begin{equation*}
			\tderiv = 
			\begin{prooftree}[separation = 1em]
			\hypo{}
			\ellipsis{$\tderivtwo$}{\typctx', \var \hastype \mtypetwo \vdash \tmtwo \hastype \mtype}
			\infer1[\footnotesize$\ruleFun$]{\typctx' \vdash \la\var\tmtwo \hastype \larrow{\mtypetwo}{\mtype}}
			\infer1[\footnotesize$\ruleManyVal$]{\typctx' \vdash \la\var\tmtwo \hastype \mset{\larrow{\mtypetwo}{\mtype}}}
			\hypo{}
			\ellipsis{$\tderiv_1$}{\quad}
			\infer2[\footnotesize$\ruleES$]{}
			\ellipsis{}{\quad}
			\hypo{}
			\ellipsis{$\tderiv_n$}{\quad}
			\infer2[\footnotesize$\ruleES$]{\typctx \vdash \subctxp{\la\var\tmtwo} \hastype \mset{\larrow{\mtypetwo}{\mtype}}}
			\hypo{}
			\ellipsis{$\tderivthree$}{\typctxtwo\vdash\tmthree \hastype \mtypetwo}
			\infer2[\footnotesize$\ruleApp$]{\typctx \uplus \typctxtwo \vdash \subctxp{\la\var\tmtwo}\tmthree \hastype \mtype}
			\end{prooftree}
			\end{equation*}
			where $\sizem{\tderiv} = 2 + \sizem{\tderivtwo} + \sizem{\tderivthree} + \sum_{i=1}^{n} \sizem{\tderiv_{i}} = \sizem{\tderiv'} + 2$ and $\size{\tderiv} = 2 + n + \size{\tderivtwo} + \size{\tderivthree} + \sum_{i=1}^{n} \size{\tderiv_{i}} = \size{\tderiv'} + 1$.
			
			\item \emph{Exponential}, \ie $\tm = \tmtwo\esub\var{\subctxp{\val}} \rtoe \subctxp{\tmtwo \isub{\var}{\val}} = \tmp$.
			Then the derivation $\tderiv$ has necessarily the form:
			\begin{equation*}
			\tderiv' = 
			\begin{prooftree}
			\hypo{}
			\ellipsis{$\tderiv''$}{\typctxtwo \mplus\typctxthree' \vdash \tmtwo\isub\var \val \hastype \mtype}
			\hypo{}
			\ellipsis{$\tderiv_1$}{\quad}
			\infer2[\footnotesize{$\Es$}]{}
			\ellipsis{}{}
			\hypo{}
			\ellipsis{$\tderiv_n$}{\quad}
			\infer2[\footnotesize$\Es$]{\typctxtwo\uplus\typctxthree \vdash \subctxp{\tmtwo\isub\var \val} \hastype \mtype}
			\end{prooftree}
			\end{equation*}
			where $\typctx = \typctxtwo \uplus \typctxthree$, $\sizem{\tderiv'} = \sizem{\tderiv''} + \sum_{i=1}^{n} \sizem{\tderiv_{i}}$ and $\size{\tderiv'} =  n + \size{\tderiv''} + \sum_{i=1}^{n} \size{\tderiv_{i}}$.
			By the removal lemma (\reflemma{anti-substitution}), there are  derivations $\namedtyjp{\tderivtwo}{}{\tmtwo}{\typctxtwo, \var\hastype\mtypetwo}{\mtype}$ and $\namedtyjp{\tderivthree}{}{\val}{\typctxthree'}{\mtypetwo}$
			such that $\sizem{\tderiv''} = \sizem{\tderivtwo} + \sizem{\tderivthree}$ and $\size{\tderiv''} \leq \size{\tderivtwo} + \size{\tderivthree}$.
			We can then build the following derivation $\tderiv$:
			\begin{equation*}
			\tderiv = 
			\begin{prooftree}
			\hypo{}
			\ellipsis{$\tderivtwo$}{\tyjp{}{\tmtwo}{\typctxtwo, \var \hastype \mtypetwo}{\mtype}}
			\hypo{}
			\ellipsis{$\tderivthree$}{\tyjp{}{\val}{\typctxthree'}{\mtypetwo}}
			\hypo{}
			\ellipsis{$\tderiv_1$}{\quad}
			\infer2[\footnotesize$\Es$]{}
			\ellipsis{}{\quad}
			\hypo{}
			\ellipsis{$\tderiv_n$}{\quad}
			\infer2[\footnotesize$\Es$]{\typctxthree \vdash \subctxp{\val} \hastype \mtypetwo}
			\infer2[\footnotesize$\Es$]{\typctxtwo \mplus \typctxthree \vdash\tmtwo\esub\var{\subctxp{\val}}\hastype \mtype}
			\end{prooftree}
			\end{equation*}
			where $\sizem{\tderiv} = \sizem{\tderivtwo} + \sizem{\tderivthree} + \sum_{i=1}^{n} \sizem{\tderiv_{i}} = \sizem{\tderiv''} + \sum_{i=1}^{n} \sizem{\tderiv_{i}} = \sizem{\tderiv'}$ and $\size{\tderiv} = 1 + n + \size{\tderivtwo} + \size{\tderivthree} + \sum_{i=1}^{n} \size{\tderiv_{i}} > n + \size{\tderiv''} +  \sum_{i=1}^{n} \size{\tderiv_{i}} = \size{\tderiv'}$.
		\end{enumerate}
		
		\item \emph{Application left}, \ie $\weakctx = \weakctxtwo \tmthree$.
		Then, $\tm = \weakctxp{\tmtwo} = \weakctxtwop{\tmtwo} \tmthree \rootRew{a} \weakctxtwop{\tmtwo'} \tmthree = \weakctxp{\tmtwo'} = \tm'$ with $\tmtwo \rootRew{a} \tmtwo'$ and $a \in \{\msym, \esym\}$.
		The derivation $\tderiv'$ is necessarily
		\begin{equation*}
		\tderiv' = 
		\begin{prooftree}
		\hypo{}
		\ellipsis{$\tderivtwo'$}{\tyjp{}{\weakctxtwop{\tmtwo'}}{\typctxtwo}{\mult{\ty{\mtypetwo}{\mtype}}}}
		\hypo{}
		\ellipsis{$\tderivthree$}{\tyjp{}{\tmthree}{\typctxthree}{\mtypetwo}}
		\infer2[\footnotesize$\ruleAp$]{\tyjp{}{\weakctxtwop{\tmtwo'} \tmthree}{\typctxtwo \mplus \typctxthree}{\mtype}}
		\end{prooftree}
		\end{equation*}
		where $\typctx = \typctxtwo \uplus \typctxthree$, $\sizem{\tderiv'} = 1 +  \sizem{\tderivtwo'} + \sizem{\tderivthree}$ and $\size{\tderiv'} = 1 +  \size{\tderivtwo'} + \size{\tderivthree}$.
		By \ih, there is a derivation $\namedtyjp{\tderivtwo}{}{\weakctxtwop{\tmtwo}}{\typctxtwo}{\mult{\ty{\mtypetwo}{\mtype}}}$ with:
		\begin{enumerate}
			\item $\sizem{\tderivtwo} = \sizem{\tderivtwo'} + 2$ and $\size{\tderivtwo} = \size{\tderivtwo'} + 1$ if $\tmtwo \rtom \tmtwo'$; 
			\item $\sizem{\tderivtwo} = \sizem{\tderivtwo'}$ and $\size{\tderivtwo} < \size{\tderivtwo'}$ if $\tmtwo \rtoe \tmtwo'$.
		\end{enumerate}
		We can then build the derivation 
		\begin{equation*}
		\tderiv = 
		\begin{prooftree}
		\hypo{}
		\ellipsis{$\tderivtwo$}{\tyjp{}{\weakctxtwop{\tmtwo}}{\typctxtwo}{\mult{\ty{\mtypetwo}{\mtype}}}}
		\hypo{}
		\ellipsis{$\tderivthree$}{\tyjp{}{\tmthree}{\typctxthree}{\mtypetwo}}
		\infer2[\footnotesize$\ruleAp$]{\tyjp{}{\weakctxtwop{\tmtwo} \tmthree}{\typctxtwo \mplus \typctxthree}{\mtype}}
		\end{prooftree}
		\end{equation*}
		noting that
		\begin{enumerate}
			\item If $\tmtwo \rtom \tmtwo'$ then $\sizem{\tderiv} = 1 + \sizem{\tderivtwo} + \sizem{\tderivthree} = 1 + (\sizem{\tderivtwo'} + 2) + \sizem{\tderivthree} = \sizem{\tderiv'} + 2$ and $\size{\tderiv} = 1 + \size{\tderivtwo} + \size{\tderivthree} = 1 + 1 +\size{\tderivtwo'} + \size{\tderivthree} = \size{\tderiv'} + 1$; 
			\item If $\tmtwo \rtoe \tmtwo'$ then $\sizem{\tderiv} = 1 + \sizem{\tderivtwo} + \sizem{\tderivthree} = 1 + \sizem{\tderivtwo'} + \sizem{\tderivthree} = \sizem{\tderiv'}$ and $\size{\tderiv} = 1 + \size{\tderivtwo} + \size{\tderivthree} > 1 + \size{\tderivtwo'} + \size{\tderivthree} = \size{\tderiv'}$.
		\end{enumerate}
		
		\item \emph{Application right}, \ie $\weakctx = \tmthree \weakctxtwo$.
		Analogous to the previous case.
		
		\item \emph{Explicit substitution left}, \ie $\weakctx = \weakctxtwo\esub{\var}{\tmthree}$. 
		Then, $\tm = \weakctxp{\tmtwo} = \weakctxtwop{\tmtwo} \esub{\var}{\tmthree} \rootRew{a} \weakctxtwop{\tmtwo'}\esub{\var}{\tmthree} = \weakctxp{\tmtwo'} = \tm'$ with $\tmtwo \rootRew{a} \tmtwo'$ and $a \in \{\msym, \esym\}$.
		The derivation $\tderiv'$ is necessarily
		\begin{equation*}
		\tderiv' = 
		\begin{prooftree}
		\hypo{}
		\ellipsis{$\tderivtwo'$}{\tyjp{}{\weakctxtwop{\tmtwo'}}{\typctxtwo ; \var \hastype \mtypetwo}{\mtype}}
		\hypo{}
		\ellipsis{$\tderivthree$}{\tyjp{}{\tmthree}{\typctxthree}{\mtypetwo}}
		\infer2[\footnotesize$\Es$]{\typctxtwo \mplus \typctxthree \vdash \weakctxtwop{\tmtwo'} \esub{\var}{\tmthree} \hastype \mtype}
		\end{prooftree}
		\end{equation*}
		where $\typctx = \typctxtwo \mplus \typctxthree$, $\sizem{\tderiv} = \sizem{\tderivtwo} + \sizem{\tderivthree}$ and $\size{\tderiv} = 1 + \size{\tderivtwo} + \size{\tderivthree}$.
		By \ih, there is a derivation $\namedtyjp{\tderivtwo}{}{\weakctxtwop{\tmtwo}}{\typctxtwo, \var \hastype \mtypetwo}{\mtype}$ with: 
		\begin{enumerate}
			\item $\sizem{\tderivtwo} = \sizem{\tderivtwo'} + 2$ and $\size{\tderivtwo} =  \size{\tderivtwo'} + 1$ if $\tmtwo \rtom \tmtwo'$; 
			\item $\sizem{\tderivtwo} = \sizem{\tderivtwo'}$ and $\size{\tderivtwo} > \size{\tderivtwo'}$ if $\tmtwo \rtoe \tmtwo'$.
		\end{enumerate}
		We can then build the derivation 
		\begin{equation*}
		\tderiv = 
		\begin{prooftree}
		\hypo{}
		\ellipsis{$\tderivtwo$}{\tyjp{}{\weakctxtwop{\tmtwo}}{\typctxtwo ; \var \hastype \mtypetwo}{\mtype}}
		\hypo{}
		\ellipsis{$\tderivthree$}{\tyjp{}{\tmthree}{\typctxthree}{\mtypetwo}}
		\infer2[\footnotesize$\Es$]{\tyjp{}{\weakctxtwop{\tmtwo} \esub{\var}{\tmthree}}{\typctxtwo \mplus \typctxthree}{\mtype}}
		\end{prooftree}
		\end{equation*}
		noting that 
		\begin{enumerate}
			\item If $\tmtwo \rtom \tmtwo'$ then $\sizem{\tderiv} = \sizem{\tderivtwo} + \sizem{\tderivthree} = (\sizem{\tderivtwo'} + 2) + \sizem{\tderivthree} = \sizem{\tderiv'} + 2$ and $\size{\tderiv} = 1 + \size{\tderivtwo} + \size{\tderivthree} = 1 + 1 +  \size{\tderivtwo'} + \size{\tderivthree} = \size{\tderiv'}$; 
			\item If $\tmtwo \rtoe \tmtwo'$ then $\sizem{\tderiv} = \sizem{\tderivtwo} + \sizem{\tderivthree} = \sizem{\tderivtwo'} + \sizem{\tderthree} = \sizem{\tderiv'}$ and $\size{\tderiv} = \size{\tderivtwo} + \size{\tderivthree} > \size{\tderivtwo'} + \size{\tderivthree} = \size{\tderiv'}$.
		\end{enumerate}
		
		\item \emph{Explicit substitution right}, \ie $\weakctx = \tmthree \esub{\var}{\weakctxtwo}$. 
		Analogous to the previous case.
		\qedhere
	\end{itemize}
\end{proof}

\begin{theorem}[Open completeness]
	\label{thmappendix:open-completeness}
	\NoteState{thm:open-completeness}
	Let $\deriv \colon \tm \tovsubo^* \tmtwo$ be an $\osym$-normalizing evaluation.
	Then there is a tight derivation $\concl{\tderiv}{\typctx}{\tm}{\emptytype}$ such that $2\sizem{\deriv} + \sizes{\fire} = \sizem{\tderiv}$.
\end{theorem}

\begin{proof}
	By induction on the length $\size{\deriv}$ of the $\osym$-evaluation $\deriv$.
	
	If $\size{\deriv} = 0$ then $\sizem{\deriv} = 0$ and $\tm = \fire$ is $\osym$-normal and hence a fireball.
	According to tight typability of fireballs (\Cref{prop:precise-open-typability-nf}), there is a tight derivation $\concl{\tderiv}{\typctx}{\tm}{\emptytype}$.
	Therefore, $\sizem{\tderiv} = \sizeo{\fire} = \sizeo{\fire} + 2\sizem{\deriv}$ by \Cref{l:size-fireballs}.
	
	Otherwise, $\size{\deriv} > 0$ and $\deriv$ is the concatenation of a first step $\tm \tovsubo \tmtwo$ and an evaluation $\deriv' \colon \tmtwo \tovsubo^* \fire$, with $\size{\deriv} = 1 + \size{\deriv'}$.
	By \ih, there is a tight derivation $\concl{\tderivtwo}{\typctx}{\tmtwo}{\emptytype}$  such that $\sizem{\tderivtwo} = \sizeo{\fire} + 2\sizem{\deriv'}$. 
	According to open subject expansion (\Cref{prop:weak-subject-expansion}), there is a derivation 
	$\concl{\tderiv}{\typctx}{\tm}{\emptytype}$ with 
	\begin{itemize}
		\item $\sizem{\tderiv} = \sizem{\tderivtwo} + 2 = \sizeo{\fire} + 2\sizem{\deriv'} + 2 = \sizeo{\fire} + 2\sizem{\deriv}$ 
		if $\tm \tomo \tmtwo$, since $\sizem{\deriv} = \sizem{\deriv'} + 1$;
		\item $\sizem{\tderiv} = \sizem{\tderivtwo} = \sizeo{\fire} + 2\sizem{\deriv'} = \sizeo{\fire} + 2\sizem{\deriv}$ if $\tm \toeo \tmtwo$, since $\sizem{\deriv} = \sizem{\deriv'}$.
		\qedhere
	\end{itemize} 
\end{proof}

%
%
\section{Proofs of \Cref{sect:solvable} (Multi Types for \cbv Solvability)}

\paragraph*{Correctness}

\begin{lemma}[Size of solvable fireballs]
	\label{lappendix:size-solvable-nf}
	\NoteState{l:size-solvable-nf}
	Let $\solvnf$ be a solvable fireball.
	If $\namedtyjp{\tderiv}{}{\solvnf}{\typctx}{\mtype}$ with $\mtype$ solvable (\resp $\typctx$  inert and $\mtype$ precisely solvable),
	then $\sizem{\tderiv} \geq \sizes{\solvnf}$ (\resp $\sizem{\tderiv} = \sizes{\solvnf}$).
\end{lemma}

\begin{proof}
	By induction on the definition of solvable fireball $\solvnf$.
	For each case, we shall first prove the part of the statement before ``moreover'', and then we shall prove the part of the statement after ``moreover''.
	Cases:
	
	\begin{itemize}
		\item \emph{Inert}, \ie $\solvnf$ is an inert term.
		According to \Cref{l:size-fireballs}, $\sizem{\tderiv} \geq \sizes{\solvnf}$.
		If, moreover, $\typctx$ is inert then $\sizem{\tderiv} = \sizeo{\solvnf} = \sizes{\solvnf}$ by \Cref{l:size-fireballs} and \Cref{l:sizes-inert}.		
		Note that in this case the hypothesis that $\mtype$ is a (precisely) solvable multi type is not used.
		
		\item \emph{Explicit substitution on solvable normal form}, \ie $\solvnf = \solvnftwo \esub{\var}{\pitm}$ for some solvable fireball $\solvnftwo$ and proper inert term $\pitm$. Then necessarily
		\begin{equation*}
		\tderiv = 
		\begin{prooftree}
		\hypo{}
		\ellipsis{$\tderivtwo$}{\typctxtwo, \var \hastype \mtypetwo \vdash \solvnftwo \hastype \mtype}
		\hypo{}
		\ellipsis{$\tderivthree$}{\typctxthree \vdash \pitm \hastype \mtypetwo}
		\infer2[\footnotesize$\ruleES$]{\tyjp{}{\solvnftwo \esub{\var}{\pitm}}{\typctxtwo \mplus \typctxthree}{\mtype}}
		\end{prooftree}
		\end{equation*}
		where $\typctx = \typctxtwo \mplus \typctxthree$.
		According to \Cref{l:size-fireballs} for inert terms, $\sizem{\tderivthree} \geq \sizeo{\pitm}$.
		By \ih applied to $\tderivtwo$, $\sizem{\tderivtwo} \geq \sizes{\solvnftwo}$. 
		Therefore, $\sizes{\solvnf} = \sizes{\solvnftwo} + \sizeo{\pitm} \leq \sizem{\tderivtwo} + \sizem{\tderivthree} = 
		\sizem{\tderiv}$ 
		
		If, moreover, $\typctx$ is inert and $\mtype$ is precisely solvable, then $\typctxtwo$ and $\typctxthree$ are inert, according to \Cref{rmk:merge-split-inert}.
		By \Cref{l:size-fireballs} for inert terms, $\sizem{\tderivthree} = \sizeo{\pitm}$.
		By spreading of inertness (\Cref{l:spread-inert}), $\mtypetwo$ is inert and hence $\typctxtwo, \var \hastype \mtypetwo$ is a inert type context.
		By \ih applied to $\tderivtwo$, $\sizem{\tderivtwo} = \sizes{\solvnftwo}$.
		Therefore, $\sizes{\solvnf} = \sizes{\solvnftwo} + \sizeo{\pitm} = \sizem{\tderivtwo} + \sizem{\tderivthree} = 
		\sizem{\tderiv}$.
		
		\item \emph{Abstraction}, \ie $\solvnf = \la{\var}{\solvnftwo}$ for some solvable fireball $\solvnftwo$. 
		Then $\mtype = \bigmplus_{i=1}^n\mset{\larrow{\mtypethree_i}{\mtypetwo_i}}$ for some $n > 0$ (as $\mtype$ is solvable), and
		\begin{equation*}
		\tderiv = 
		\begin{prooftree}[separation = 1em]
		\hypo{}
		\ellipsis{$\tderivtwo_1$}{\typctx_1, \var \hastype \mtypethree_1 \vdash \solvnftwo \hastype \mtypetwo_1}
		\infer1[\footnotesize$\ruleFun$]{\tyjp{}{\la{\var}{\solvnftwo}}{\typctx_1}{\ty{\mtypethree_1}{\mtypetwo_1}}}
		\hypo{\overset{n \in \nat}{\ldots}}
		\hypo{}
		\ellipsis{$\tderivtwo_n$}{\typctx_n, \var \hastype \mtypethree_n \vdash \solvnftwo \hastype \mtypetwo_n}
		\infer1[\footnotesize$\ruleFun$]{\tyjp{}{\la{\var}{\solvnftwo}}{\typctx_n}{\larrow{\mtypethree_n}{\mtypetwo_n}}}
		\infer3[\footnotesize$\ruleManyVal$]{\tyjp{}{\la{\var}{\solvnftwo}}{\typctx}{\mtype}}
		\end{prooftree}
		\end{equation*}
		where $\typctx = \bigmplus_{i=1}^n\typctx_i$. 
		By \ih, $\sizes{\solvnftwo} \leq \sizem{\tderivtwo_i}$ for all $1 \leq i \leq n$.
		Therefore, $\sizes{\solvnf} =  1 + \sizes{\solvnftwo} = 1 + \sum_{i=1}^n\sizem{\tderivtwo_i} \leq \sum_{i=1}^n(\sizem{\tderivtwo_i} + 1) = \sizem{\tderiv}$, where the inequality holds because $n > 0$.
		
		If, moreover, $\typctx$ is inert and $\mtype$ is precisely solvable, then $n = 1$ and $\typctx = \typctx_1$ and $\mtype = \mset{\larrow{\mtypethree_1}{\mtypetwo_1}}$ with $\mtypethree_1$ inert. 
		Thus, $\typctx_1, \var \hastype \mtypethree_1$ is an inert type context.
		By \ih, $\sizem{\tderivtwo_1} = \sizes{\solvnftwo}$.
		So, $\sizes{\solvnf}= \sizes{\solvnftwo} + 1 = \sizem{\tderivtwo_1} + 1 = \sizem{\tderiv}$. 
		\qedhere
	\end{itemize}
\end{proof}

\begin{proposition}[Solvable quantitative subject reduction]
	\label{propappendix:solvable-subject-reduction}
	\NoteState{prop:solvable-subject-reduction}
	Assume $\concl{\tderiv}{\typctx}{\tm}{\mtype}$ with $\mtype$ solvable (\resp unitary solvable).
	\begin{enumerate}
		\item \emph{Multiplicative step:} If $\tm \tosolvm \tm'$ then there is a derivation 
$\concl{\tderiv'}{\typctx}{\tm'}{\mtype}$ such that $\sizem{\tderiv'} \leq \sizem{\tderiv}-2$ and $\size{\tderiv'} < 
\size{\tderiv}$
		(\resp $\sizem{\tderiv'} = \sizem{\tderiv}-2$ and $\size{\tderiv'} = \size{\tderiv}-1$);
		
		\item \emph{Exponential step:} if $\tm \tosolve \tm'$ then there is a derivation 
$\concl{\tderiv'}{\typctx}{\tm'}{\mtype}$ such that
		$\sizem{\tderiv'} = \sizem{\tderiv}$ and $\size{\tderiv'} < \size{\tderiv}$.
	\end{enumerate}
\end{proposition}

\begin{proof}
	First, we prove the unitary solvable version of the claim (\ie under the hypothesis that $\mtype$ is unitary solvable).
	The proof is by induction on the solvable context $\solvctx$ such that $\tm = \solvctxp{\tmtwo} \tosolv \solvctxp{\tmtwo'} = \tm'$ with $\tmtwo \tomo \tmtwo'$ or $\tmtwo' \toeo \tmtwo'$. 
	Cases for $\solvctx$:
	\begin{itemize}
		\item \emph{Hole context}, \ie{} $\solvctx = \ctxhole$ and $\tm \Rew{\wsym a}  \tm'$ with $a \in \{\msym, \esym\}$.
		According to subject reduction for $\tow$ (\Cref{prop:weak-subject-reduction}),
		\begin{itemize}
			\item if $\tm \towm \tm'$ then there exists a derivation $\concl{\tderiv'}{\typctx}{\tm'}{\mtype}$ such that
			$\sizem{\tderiv'} = \sizem{\tderiv}-2$ and $\size{\tderiv'} = \size{\tderiv}-1$;
			\item if $\tm \towe \tm'$ then there is a derivation $\concl{\tderiv'}{\typctx}{\tm'}{\mtype}$ such that 
$\sizem{\tderiv'} = \sizem{\tderiv}$ and
			$\size{\tderiv'} < \size{\tderiv}$.
		\end{itemize}
		Note that in this case the hypothesis that $\mtype$ is a (unitary) solvable multi type is not used.
		
		\item \emph{Abstraction}, \ie $\solvctx = \la{\var}{\solvctxtwo}$. 
		So, $\tm = \solvctxp{\tmtwo} = \la{\var}{\solvctxtwop{\tmtwo}} \Rew{\solvredsym a} \la{\var}{\solvctxtwop{\tmtwo'}} = 
\solvctxp{\tmtwo'} = \tm'$ with $\tmtwo \Rew{\wsym a} \tmtwo'$ and $a \in \{\msym, \esym\}$.
		Since $\mtype$ is a unitary solvable multi type by hypothesis and hence it has the form $\mtype = \mset{\larrow{\mtypethree}{\smtypetwo}}$ where $\smtypetwo$ is unitary solvable.
		Thus, the derivation $\tderiv$ is necessarily
		\begin{equation*}
		\tderiv = 
		\begin{prooftree}
		\hypo{}
		\ellipsis{$\tderivtwo$}{\typctx, \var \hastype \mtypethree \vdash \solvctxtwop{\tmtwo} \hastype \smtypetwo}
		\infer1[\footnotesize$\lambda$]{\typctx \vdash \la{\var}\solvctxtwop{\tmtwo} \hastype \larrow{\mtypethree}{\smtypetwo}}
		\infer1[\footnotesize$\ruleManyVal$]{\typctx \vdash \la{\var}\solvctxtwop{\tmtwo} \hastype 
\mset{\larrow{\mtypethree}{\smtypetwo}}}
		\end{prooftree}
		\end{equation*}
		By \ih, there is a derivation $\concl{\tderivtwo'}{\typctx, \var \hastype 
\mtypetwo}{\solvctxtwop{\tmtwo'}}{\smtypetwo}$ with: 
		\begin{enumerate}
			\item $\sizem{\tderivtwo'} = \sizem{\tderivtwo} - 2$ and $\size{\tderivtwo'} = \size{\tderivtwo}-1$ if $\tmtwo 
\tomo \tmtwo'$ ; 
			\item $\sizem{\tderivtwo'} = \sizem{\tderivtwo}$ and $\size{\tderivtwo'} < \size{\tderivtwo}$ if $\tmtwo \toeo 
\tmtwo'$.
		\end{enumerate}
		We can then build the derivation 
		\begin{equation*}
		\tderiv' = 
		\begin{prooftree}
		\hypo{}
		\ellipsis{$\tderivtwo'$}{\typctx, \var \hastype \mtypethree \vdash \solvctxtwop{\tmtwo'} \hastype \smtypetwo}
		\infer1[\footnotesize$\lambda$]{\typctx \vdash \la{\var}\solvctxtwop{\tmtwo'} \hastype \larrow{\mtypethree}{\smtypetwo}}
		\infer1[\footnotesize$\ruleManyVal$]{\typctx \vdash \la{\var}\solvctxtwop{\tmtwo'} \hastype 
\mset{\larrow{\mtypethree}{\smtypetwo}}}
		\end{prooftree}
		\end{equation*}
		where
		\begin{enumerate}
			\item $\sizem{\tderiv'} = \sizem{\tderivtwo'} +1 = \sizem{\tderivtwo} + 1 - 2 = \sizem{\tderiv} - 2$ and 
$\size{\tderiv'} = \size{\tderivtwo'} +1 = \size{\tderivtwo} + 1 - 1 = \size{\tderiv} - 1$ if $\tmtwo \tomo \tmtwo'$ (\ie $\tm \tosolvm \tm'$); 
			\item $\sizem{\tderiv'} = \sizem{\tderivtwo'} +1 = \sizem{\tderivtwo} + 1 = \sizem{\tderiv}$ and $\size{\tderiv'} 
= \size{\tderivtwo'} +1 < \size{\tderivtwo} + 1 = \size{\tderiv}$ if $\tmtwo \toeo \tmtwo'$ (\ie $\tm \tosolve \tm'$).
		\end{enumerate}

		\item \emph{Explicit substitution}, \ie $\solvctx = \solvctxtwo\esub{\var}{\tmthree}$. 
		Then, $\tm = \solvctxp{\tmtwo} = \solvctxtwop{\tmtwo} \esub{\var}{\tmthree} \Rew{\solvredsym a} 
\solvctxtwop{\tmtwo'}\esub{\var}{\tmthree} = \solvctxp{\tmtwo'} = \tm'$ with $\tmtwo \Rew{\wsym a} \tmtwo'$ and $a \in 
\{\msym, \esym\}$.
		The derivation $\tderiv$ is necessarily
		\begin{equation*}
		\tderiv = 
		\begin{prooftree}
		\hypo{}
		\ellipsis{$\tderivtwo$}{\typctxtwo, \var \hastype \mtypetwo \vdash \solvctxtwop{\tmtwo} \hastype \mtype}
		\hypo{}
		\ellipsis{$\tderivthree$}{\typctxthree \vdash \tmthree \hastype \mtypetwo}
		\infer2[\footnotesize$\Es$]{\typctxtwo \uplus \typctxthree \vdash \solvctxtwop{\tmtwo} \esub{\var}{\tmthree} \hastype \mtype}
		\end{prooftree}
		\end{equation*}
		where $\typctx = \typctxtwo \uplus \typctxthree$.
		By \ih applied to $\tderivtwo$, there is a derivation $\concl{\tderivtwo'}{\typctxtwo, \var \hastype \mtypetwo}{\solvctxtwop{\tmtwo'}}{\mtype}$ with: 
		\begin{enumerate}
			\item $\sizem{\tderivtwo'} = \sizem{\tderivtwo} - 2$ and $\size{\tderivtwo'} = \size{\tderivtwo} - 1$ if $\tmtwo 
\tomo \tmtwo'$; 
			\item $\sizem{\tderivtwo'} = \sizem{\tderivtwo}$ and $\size{\tderivtwo'} < \size{\tderivtwo}$ if $\tmtwo \toeo \tmtwo'$.
		\end{enumerate}
		We can then build the derivation 
		\begin{equation*}
		\tderiv' = 
		\begin{prooftree}
		\hypo{}
		\ellipsis{$\tderivtwo'$}{\typctxtwo, \var \hastype \mtypetwo \vdash \solvctxtwop{\tmtwo'} \hastype \mtype}
		\hypo{}
		\ellipsis{$\tderivthree$}{\typctxthree \vdash \tmthree \hastype \mtypetwo}
		\infer2[\footnotesize$\Es$]{\typctxtwo \uplus \typctxthree \vdash \solvctxtwop{\tmtwo'} \esub{\var}{\tmthree} \hastype \mtype}
		\end{prooftree}
		\end{equation*}
		where $\typctx = \typctxtwo \uplus \typctxthree$ and
		\begin{enumerate}
			\item $\sizem{\tderiv'} = \sizem{\tderivtwo'} + \sizem{\tderivthree} = \sizem{\tderivtwo} + \sizem{\tderivthree} - 2 = \sizem{\tderiv} - 2$ and $\size{\tderiv'} = \size{\tderivtwo'} + \size{\tderivthree} +1 = \size{\tderivtwo} -1 + 
\size{\tderivthree} +1 = \size{\tderiv} - 1$ if $\tmtwo \tomo \tmtwo'$ (\ie $\tm \tosolvm \tm'$); 
			\item $\sizem{\tderiv'} = \sizem{\tderivtwo'} + \sizem{\tderivthree} = \sizem{\tderivtwo} + \sizem{\tderivthree} = \sizem{\tderiv}$ and $\size{\tderiv'} = \size{\tderivtwo'} + \size{\tderivthree} +1 < \size{\tderivtwo} + 
\size{\tderivthree} +1 = \size{\tderiv}$ if $\tmtwo \toeo \tmtwo'$ (\ie $\tm \tosolve \tm'$).
		\end{enumerate}

		\item \emph{Application}, \ie $\solvctx = \solvctxtwo\tmthree$. 
		Then, $\tm = \solvctxp{\tmtwo} = \solvctxtwop{\tmtwo} \tmthree \Rew{\solvredsym a} \solvctxtwop{\tmtwo'} \tmthree = \solvctxp{\tmtwo'} = \tm'$ with $\tmtwo \Rew{\wsym a} \tmtwo'$ and $a \in \{\msym, \esym\}$.
		The derivation $\tderiv$ is necessarily
		\begin{equation*}
		\tderiv = 
		\begin{prooftree}
		\hypo{}
		\ellipsis{$\tderivtwo$}{\typctxtwo \vdash \solvctxtwop{\tmtwo} \hastype \mset{\larrow{\mtypetwo}{\mtype}}}
		\hypo{}
		\ellipsis{$\tderivthree$}{\typctxthree \vdash \tmthree \hastype \mtypetwo}
		\infer2[\footnotesize$\ruleAp$]{\typctxtwo \uplus \typctxthree \vdash \solvctxtwop{\tmtwo} \tmthree \hastype \mtype}
		\end{prooftree}
		\end{equation*}
		where $\typctx = \typctxtwo \mplus \typctxthree$.
		By \ih applied to $\tderivtwo$ (since $\mset{\larrow{\mtypetwo}{\mtype}}$ is a unitary solvable multi type), there 
is 
a derivation $\concl{\tderivtwo'}{\typctxtwo}{\solvctxtwop{\tmtwo'}}{\mset{\larrow{\mtypetwo}{\mtype}}}$ with: 
		\begin{enumerate}
			\item $\sizem{\tderivtwo'} = \sizem{\tderivtwo} - 2$ and $\size{\tderivtwo'} = \size{\tderivtwo} - 1$ if $\tmtwo 
\tomo \tmtwo'$; 
			\item $\sizem{\tderivtwo'} = \sizem{\tderivtwo}$ and $\size{\tderivtwo'} < \size{\tderivtwo}$ if $\tmtwo \toeo 
\tmtwo'$.
		\end{enumerate}
		We can then build the derivation 
		\begin{equation*}
		\tderiv' = 
		\begin{prooftree}
		\hypo{}
		\ellipsis{$\tderivtwo'$}{\typctxtwo \vdash \solvctxtwop{\tmtwo'} \hastype \mset{\larrow{\mtypetwo}{\mtype}}}
		\hypo{}
		\ellipsis{$\tderivthree$}{\typctxthree \vdash \tmthree \hastype \mtypetwo}
		\infer2[\footnotesize$\ruleAp$]{\typctxtwo \uplus \typctxthree \vdash \solvctxtwop{\tmtwo'} \tmthree \hastype 
\mtype}
		\end{prooftree}
		\end{equation*}
		where $\typctx = \typctxtwo \uplus \typctxthree$ and
		\begin{enumerate}
			\item $\sizem{\tderiv'} = \sizem{\tderivtwo'} + \sizem{\tderivthree} +1 = \sizem{\tderivtwo} - 2 + 
\sizem{\tderivthree} +1 = \sizem{\tderiv} - 2$ and $\size{\tderiv'} = \size{\tderivtwo'} + \size{\tderivthree} +1 = 
\size{\tderivtwo} - 1 + \size{\tderivthree} +1 = \size{\tderiv} - 1$ if $\tmtwo \tomo \tmtwo'$ (\ie $\tm \tosolvm \tm'$); 
			\item $\sizem{\tderiv'} = \sizem{\tderivtwo'} + \sizem{\tderivthree} +1 = \sizem{\tderivtwo} + \sizem{\tderivthree} +1 = \sizem{\tderiv}$ and $\size{\tderiv'} = \size{\tderivtwo'} + \size{\tderivthree} +1 <  \size{\tderivtwo} + \size{\tderivthree} +1 = \size{\tderiv}$ if $\tmtwo \toeo \tmtwo'$ (\ie $\tm \tosolve \tm'$).
		\end{enumerate}
	\end{itemize}
	
	This completes the proof for the unitary solvable case.
	
	In the solvable case (\ie under the weaker hypothesis that  $\mtype$ is a solvable multi type), the proof is 
analogous to the one for unitary solvable, except for the \emph{Abstraction} case.
	Indeed, in the base case (\emph{Hole context}) unitary solvability does not play any role, and the other cases follow 
from the \ih in a way analogous to unitary solvable.
	Let us see the only substantially different case: 
	\begin{itemize}
		\item \emph{Abstraction}, \ie $\solvctx = \la{\var}{\solvctxtwo}$. 
		So, $\tm = \solvctxp{\tmtwo} = \la{\var}{\solvctxtwop{\tmtwo}} \Rew{\solvredsym a} \la{\var}{\solvctxtwop{\tmtwo'}} = 
\solvctxp{\tmtwo'} = \tm'$ with $\tmtwo \Rew{\wsym a} \tmtwo'$ and $a \in \{\msym, \esym\}$.
		Since $\mtype$ is a solvable multi type by hypothesis, it has the form $\mtype = \mset{\larrow{\mtype_1}{\smtype_1}, \dots, \larrow{\mtype_n}{\smtype_n}}$ for some $n > 0$, where $\smtype_j$ is solvable for all $1 \leq j \leq n$.
		Thus, the derivation $\tderiv$ is necessarily
		\begin{equation*}
		\tderiv = 
		\begin{prooftree}[separation=1em]
		\hypo{}
		\ellipsis{$\tderivtwo_j$}{\typctx_j, \var \hastype \mtype_j \vdash \solvctxtwop{\tmtwo} \hastype \smtype_j}
		\infer1[\footnotesize$\lambda$]{\typctx_j \vdash \la{\var}\solvctxtwop{\tmtwo} \hastype 
\larrow{\mtype_j}{\smtype_j}}
		\delims{ \left( }{ \right)_{1 \leq j \leq n} }
		\infer1[\footnotesize$\ruleManyVal$]{ \bigmplus_{j=1}^n \typctx_{j} \vdash \la{\var}\solvctxtwop{\tmtwo} \hastype  
\bigmplus_{j=1}^n \mset{\larrow{\mtype_j}{\smtype_j}}}
		\end{prooftree}
		\end{equation*}
		For all $1 \leq j \leq n$, by \ih, there is a derivation $\concl{\tderivtwo_j'}{\typctx_j, \var \hastype 
\mtype_j}{\solvctxtwop{\tmtwo'}}{\smtype_j}$ with: 
		\begin{enumerate}
			\item $\sizem{\tderivtwo_j'} \leq \sizem{\tderivtwo_j} - 2$ and $\size{\tderivtwo_j'} < \size{\tderivtwo_j}$ 
if $\tmtwo \tomo \tmtwo'$ ; 
			\item $\sizem{\tderivtwo_j'} = \sizem{\tderivtwo_j}$ and $\size{\tderivtwo_j'} < \size{\tderivtwo_j}$ if $\tmtwo 
\toeo \tmtwo'$.
		\end{enumerate}
		We can then build the derivation 
		\begin{equation*}
		\tderiv' = 
		\begin{prooftree}[separation=1em]
		\hypo{}
		\ellipsis{$\tderivtwop_j$}{\typctx_j, \var \hastype \mtype_j \vdash \solvctxtwop{\tmtwo'} \hastype \smtype_i}
		\infer1[\footnotesize$\lambda$]{\typctx_j \vdash \la{\var}\solvctxtwop{\tmtwo'} \hastype 
\larrow{\mtype_j}{\smtype_j}}
		\delims{ \left( }{ \right)_{1 \leq j \leq n} }
		\infer1[\footnotesize$\ruleManyVal$]{ \bigmplus_{j=1}^n \typctx_{j} \vdash \la{\var}\solvctxtwop{\tmtwop} \hastype  
\bigmplus_{j=1}^n \mset{\larrow{\mtype_j}{\smtype_j}}}
		\end{prooftree}
		\end{equation*}
		where
		\begin{enumerate}
			\item $\sizem{\tderiv'} = \sum_{j=1}^n(\sizem{\tderivtwo_j'} +1) \leq \sum_{j=1}^n(\sizem{\tderivtwo_j} + 1 - 2) = \sizem{\tderiv} - 2n \leq \sizem{\tderiv} - 2$ (where the last inequality holds because $n >0$) and $\size{\tderiv'} =  \sum_{j=1}^n(\size{\tderivtwo_j'} +1) =  \sum_{j=1}^n(\size{\tderivtwo_j} + 1 - 1) \leq \size{\tderiv} - n < \size{\tderiv}$ (the last inequality holds because $n >0$) if $\tmtwo \tomo \tmtwo'$, \ie $\tm \tosolvm \tm'$; 
			\item $\sizem{\tderiv'} = \sum_{j=1}^n(\sizem{\tderivtwo_j'} +1) = \sum_{j=1}^n(\sizem{\tderivtwo_j} + 1) = \sizem{\tderiv}$ and $\size{\tderiv'} = \sum_{j=1}^n(\size{\tderivtwo_j'} +1) < \sum_{j=1}^n(\size{\tderivtwo_j} + 1) = \size{\tderiv}$ if $\tmtwo \toeo \tmtwo'$, \ie $\tm \tosolve \tm'$ (the inequality holds because $n > 0$).
			\qedhere
		\end{enumerate}
	\end{itemize}	
\end{proof}

\begin{theorem}[Solvable correctness]
	\label{thmappendix:solvable-correctness}
	\NoteState{thm:solvable-correctness}
	Let $\concl{\tderiv}{\typctx}{\tm}{\mtype}$ be a derivation with $\mtype$ solvable (\resp $\typctx$ inert and $\mtype$ precisely solvable).
	Then there is an $\solvredsym$-normalizing evaluation $\deriv \colon \tm \tosolv^* \tmtwo$ with $2\sizem{\deriv} + \sizes{\tmtwo} \leq \sizem{\tderiv}$ (\resp $2\sizem{\deriv} + \sizes{\tmtwo} = \sizem{\tderiv}$).
\end{theorem}

\begin{proof}
	By induction on the general size $\size{\tderiv}$ of $\tderiv$.
	
	If $\tm$ is normal for $\tosolv$, then $\tm = \solvnf$ is a solvable fireball. 
	Let $\deriv$ be the empty evaluation  (so $\sizem{\deriv} = 0$), thus $\sizem{\tderiv} \geq \sizes{\solvnf} = \sizes{\solvnf} + 2\sizem{\deriv}$ 
	(\resp $\sizem{\tderiv} = \sizes{\solvnf} = \sizes{\solvnf} + 2\sizem{\deriv}$) by \reflemma{size-solvable-nf}.
	
	Otherwise, $\tm$ is not normal for $\tosolv$ and so $\tm \tosolv \tmtwo$.
	According to solvable subject reduction (\Cref{prop:solvable-subject-reduction}), there is a derivation $\concl{\tderivtwo}{\typctx}{\tmtwo}{\mtype}$ such that $\size{\tderivtwo} < \size{\tderiv}$  and 
	\begin{itemize}
		\item $\sizem{\tderivtwo} \leq \sizem{\tderiv} - 2$ (\resp $\sizem{\tderivtwo} = \sizem{\tderiv} - 2$) if $\tm \tosolvm \tmtwo$,
		\item $\sizem{\tderivtwo} = \sizem{\tderiv}$ if $\tm \tosolve \tmtwo$.
	\end{itemize}
	By \ih, there is a solvable fireball $\solvnf$ and a reduction sequence $\deriv' \colon \tmtwo \tosolv^* \solvnf$ with 
	$2\sizem{\deriv'} + \sizes{\solvnf} \leq \sizem{\tderivtwo}$ (\resp $2\sizem{\deriv'} + \sizes{\solvnf} = \sizem{\tderivtwo}$).
	Let $\deriv$ be the $\solvredsym$-evaluation obtained by concatenating the first step $\tm \tosolv \tmtwo$ and $\deriv'$.
	There are two cases:
	\begin{itemize}
		\item \emph{Multiplicative:} if $\tm \tosolvm \tmtwo$ then $\sizem{\tderiv} \geq \sizem{\tderivtwo} + 2 \geq \sizes{\solvnf} 
		+ 2\sizem{\deriv'} + 2 = \sizes{\solvnf} + 2\sizem{\deriv}$ (\resp $\sizem{\tderiv} = \sizem{\tderivtwo} + 2 = \sizes{\solvnf} + 
		2\sizem{\deriv'} + 2 = \sizes{\solvnf} + 2\sizem{\deriv}$), since $\sizem{\deriv} = \sizem{\deriv'} + 1$.
		\item \emph{Exponential:} if $\tm \tosolve \tmtwo$ then $\sizem{\tderiv} = \sizem{\tderivtwo} \geq \sizes{\solvnf} + 2\sizem{\deriv'} = \sizes{\solvnf} + 2\sizem{\deriv}$ (\resp $\sizem{\tderiv} = \sizem{\tderivtwo} = \sizes{\solvnf} + 2\sizem{\deriv'} = \sizes{\solvnf} + 2\sizem{\deriv}$),  since $\sizem{\deriv} = \sizem{\deriv'}$.
		\qedhere
	\end{itemize} 
\end{proof}

\paragraph*{Completeness}

\begin{lemma}[Precisely solvable typability of solvable fireballs]
	\label{propappendix:precise-solvable-typability-nf}
	\NoteState{prop:precise-solvable-typability-nf}
	If $\tm$ is a solvable fireball, then there is a derivation 
	$\concl{\tderiv}{\typctx}{\tm}{\mtype}$ with $\typctx$ inert and $\mtype$ precisely solvable.
\end{lemma}

\begin{proof}		
		By induction on the definition of solvable fireball. 
		Cases:
		\begin{itemize}
			\item \emph{Inert}, \ie $\tm$ is an inert term. According to \Cref{prop:precise-open-typability-nf}.\ref{p:precise-open-typability-nf-inert}, since the precisely solvable multi type $\mset{\ground}$ is also inert, there is a derivation 
			$\concl{\tderiv}{\typctx}{\tm}{\mset{\ground}}$ for some inert type context $\typctx$.
			
			\item \emph{Abstraction}, \ie $\tm = \la{\var}{\solvnf}$ for some solvable fireball $\solvnf$.
			By \ih, there is a derivation $\concl{\tderivtwo}{\typctx, \var \hastype \imtype}{\solvnf}{\psmtypetwo}$ for some unitary precisely solvable multi type $\psmtypetwo$, inert multi type $\imtype$ and inert type context $\typctx$. 
			We can then build the derivation 
			\begin{equation*}
			\tderiv = 
			\begin{prooftree}
			\hypo{}
			\ellipsis{$\tderivtwo$}{\typctx, \var \hastype \imtype \vdash \solvnf \hastype \psmtypetwo}
			\infer1[\footnotesize$\lambda$]{\typctx \vdash \la{\var}\solvnf \hastype \larrow{\imtype}{\psmtypetwo}}
			\infer1[\footnotesize$\ruleManyVal$]{\typctx \vdash \la{\var}\solvnf \hastype \mset{\larrow{\imtype}{\psmtypetwo}}}
			\end{prooftree}
			\end{equation*}
			where $\mtype = \mset{\larrow{\imtype}{\psmtypetwo}}$ is a precisely solvable multi type.
			
			\item \emph{Explicit substitution}, \ie $\tm = {\solvnf}\esub{\var}{\pitm}$ for some solvable fireball $\solvnf$ and proper inert term $\pitm$.
			By \ih, there is a derivation $\concl{\tderivtwo}{\typctxtwo, \var \hastype \imtype}{\solvnf}{\psmtypetwo}$ for some inert type context $\typctxtwo$, inert multi type $\imtype$ and unitary precisely solvable multi type $\smtypetwo$.
			By \Cref{prop:precise-open-typability-nf}.\ref{p:precise-open-typability-nf-inert}, there is a derivation $\concl{\tderivthree}{\typctxthree}{\pitm}{\imtype}$ for some inert type context $\typctxthree$.
			We can build the derivation 
			\begin{equation*}
			\tderiv = 
			\begin{prooftree}
			\hypo{}
			\ellipsis{$\tderivtwo$}{\typctxtwo, \var \hastype \imtype \vdash \solvnf \hastype \smtype}
			\hypo{}
			\ellipsis{$\tderivthree$}{\typctxthree \vdash \pitm \hastype \imtype}
			\infer2[\footnotesize$\ruleES$]{\typctxtwo \uplus \typctxthree \vdash \solvnf \esub{\var}{\pitm} \hastype \psmtype}
			\end{prooftree}
			\end{equation*}
			where $\typctx = \typctxtwo \mplus \typctxthree$ is a inert type context, by \Cref{rmk:merge-split-inert}.
			\qedhere
		\end{itemize}
\end{proof}

\begin{proposition}[Solvable quantitative subject expansion]
	\label{propappendix:solvable-subject-expansion}
	\NoteState{prop:solvable-subject-expansion}
	Assume $\concl{\tderiv'\!}{\typctx}{\tm'\!}{\mtypetwo}$ with $\mtypetwo$ solvable (\resp unitary solvable).
	\begin{enumerate}
		\item \emph{Multiplicative step:} if $\tm \tosolvm \tm'$ then there is a derivation 
$\concl{\tderiv}{\typctx}{\tm}{\mtypetwo}$ with $\sizem{\tderiv'} \leq \sizem{\tderiv}-2$ and $\size{\tderiv'} < \size{\tderiv}$
		(\resp $\sizem{\tderiv'} = \sizem{\tderiv}-2$ and $\size{\tderiv'} = \size{\tderiv} - 1$);
		\item \emph{Exponential step:} if $\tm \tosolve \tm'$ then there is a derivation 
$\concl{\tderiv}{\typctx}{\tm}{\mtypetwo}$ such that
		$\sizem{\tderiv'} = \sizem{\tderiv}$ and $\size{\tderiv'} < \size{\tderiv}$.
	\end{enumerate}
\end{proposition}

\begin{proof}
	First, we prove the unitary solvable version of the claim (\ie under the hypothesis that $\mtypetwo$ is unitary solvable).
	The proof is by induction on the solvable context $\solvctx$ such that $\tm = \solvctxp{\tmtwo} \tosolv \solvctxp{\tmtwo'} = \tm'$ with $\tmtwo \tomo \tmtwo'$ or $\tmtwo' \toeo \tmtwo'$. 
	Cases for $\solvctx$:
	\begin{itemize}
		\item \emph{Hole context}, \ie{} $\solvctx = \ctxhole$ and $\tm \Rew{\wsym a} \tm'$ with $a \in \{\msym, \esym\}$.
		According to subject expansion for $\tow$ (\Cref{prop:weak-subject-expansion}),
		\begin{itemize}
			\item if $\tm \towm \tm'$ then there is a derivation $\concl{\tderiv}{\typctx}{\tm}{\mtypetwo}$ such that
			$\sizem{\tderiv'} = \sizem{\tderiv}-2$ and $\size{\tderiv'} = \size{\tderiv} - 1$;
			\item if $\tm \towe \tm'$ then there is a derivation $\concl{\tderiv}{\typctx}{\tm}{\mtypetwo}$ such that
			$\sizem{\tderiv'} = \sizem{\tderiv}$ and $\size{\tderiv'} < \size{\tderiv}$.
		\end{itemize}
		Note that in this case the hypothesis that $\mtypetwo$ is a unitary solvable multi type is not used.
		
		\item \emph{Abstraction}, \ie $\solvctx = \la{\var}{\solvctxtwo}$. 
		Then, $\tm = \solvctxp{\tmtwo} = \la{\var}{\solvctxtwop{\tmtwo}} \Rew{\solvredsym a} \la{\var}{\solvctxtwop{\tmtwo'}} = \solvctxp{\tmtwo'} = \tm'$ with $\tmtwo \Rew{\wsym a} \tmtwo'$ and $a \in \{\msym, \esym\}$.
		Since $\mtypetwo$ is unitary solvable by hypothesis, it has the form $\mtypetwo = 
\mset{\larrow{\mtypethree}{\mtype}}$ where $\mtype$ is unitary solvable.
		Therefore, the derivation $\tderiv'$ is necessarily
		\begin{equation*}
		\tderiv' = 
		\begin{prooftree}
		\hypo{}
		\ellipsis{$\tderivtwo'$}{\typctx, \var \hastype \mtypethree \vdash \solvctxtwop{\tmtwo'} \hastype \mtype}
		\infer1[\footnotesize$\lambda$]{\typctx \vdash \la{\var}\solvctxtwop{\tmtwo'} \hastype \larrow{\mtypethree}{\mtype}}
		\infer1[\footnotesize$\ruleManyVal$]{\typctx \vdash \la{\var}\solvctxtwop{\tmtwo'} \hastype 
\mset{\larrow{\mtypethree}{\mtype}}}
		\end{prooftree}
		\end{equation*}
		By \ih, there is a derivation $\concl{\tderivtwo}{\typctx, \var \hastype \mtypethree}{\solvctxtwop{\tmtwo}}{\mtype}$ 
with: 
		\begin{enumerate}
			\item $\sizem{\tderivtwo'} = \sizem{\tderivtwo} - 2$ and $\size{\tderivtwo'} = \size{\tderivtwo} - 1$ if $\tmtwo \tomo \tmtwo'$; 
			\item $\size{\tderivtwo'} = \size{\tderivtwo}$ and $\size{\tderivtwo'} < \size{\tderivtwo}$ if $\tmtwo \toeo 
\tmtwo'$.
		\end{enumerate}
		We can then build the derivation 
		\begin{equation*}
		\tderiv = 
		\begin{prooftree}
		\hypo{}
		\ellipsis{$\tderivtwo$}{\typctx, \var \hastype \mtypethree \vdash \solvctxtwop{\tmtwo} \hastype \mtype}
		\infer1[\footnotesize$\lambda$]{\typctx \vdash \la{\var}\solvctxtwop{\tmtwo} \hastype \larrow{\mtypethree}{\mtype}}
		\infer1[\footnotesize$\ruleManyVal$]{\typctx \vdash \la{\var}\solvctxtwop{\tmtwo} \hastype 
\mset{\larrow{\mtypethree}{\mtype}}}
		\end{prooftree}
		\end{equation*}
		where
		\begin{enumerate}
			\item $\sizem{\tderiv'} = \sizem{\tderivtwo'} +1 = \sizem{\tderivtwo} - 2 +1 = \sizem{\tderiv} - 2$ and 
$\size{\tderiv'} = \size{\tderivtwo'} +1 = \size{\tderivtwo} -1 +1 = \size{\tderiv} - 1$ if $\tmtwo \tomo \tmtwo'$ (\ie if $\tm \tosolvm \tm'$); 
			\item $\sizem{\tderiv'} = \sizem{\tderivtwo'} +1 = \sizem{\tderivtwo} + 1 = \sizem{\tderiv}$ and $\size{\tderiv'} 
= \size{\tderivtwo'} +1 < \size{\tderivtwo} + 1 = \size{\tderiv}$ if $\tmtwo \toeo \tmtwo'$ (\ie if $\tm \tosolve \tm'$).
		\end{enumerate}

		\item \emph{Explicit substitution}, \ie $\solvctx = \solvctxtwo\esub{\var}{\tmthree}$. 
		Then, $\tm = \solvctxp{\tmtwo} = \solvctxtwop{\tmtwo} \esub{\var}{\tmthree} \Rew{\solvredsym a} 
\solvctxtwop{\tmtwo'}\esub{\var}{\tmthree} = \solvctxp{\tmtwo'} = \tm'$ with $\tmtwo \Rew{\wsym a} \tmtwo'$ and $a \in 
\{\msym, \esym\}$.
		The derivation $\tderiv'$ is necessarily
		\begin{equation*}
			\tderiv' = 
			\begin{prooftree}
			\hypo{}
			\ellipsis{$\tderivtwo'$}{\typctxtwo, \var \hastype \mtype \vdash \solvctxtwop{\tmtwo'} \hastype \mtypetwo}
			\hypo{}
			\ellipsis{$\tderivthree$}{\typctxthree \vdash \tmthree \hastype \mtype}
			\infer2[\footnotesize$\Es$]{\typctxtwo \uplus \typctxthree \vdash \solvctxtwop{\tmtwo'} \esub{\var}{\tmthree} \hastype \mtypetwo}
			\end{prooftree}
		\end{equation*}
		where $\typctx = \typctxtwo \uplus \typctxthree$.
		By \ih applied to $\tderivtwo'$, there is a derivation $\concl{\tderivtwo}{\typctxtwo, \var \hastype \mtypetwo}{\solvctxtwop{\tmtwo}}{\mtypetwo}$ with: 
		\begin{enumerate}
			\item $\sizem{\tderivtwo'} = \sizem{\tderivtwo} - 2$ and $\size{\tderivtwo'} = \size{\tderivtwo} - 1$ if $\tmtwo \tomo \tmtwo'$; 
			\item $\sizem{\tderivtwo'} = \sizem{\tderivtwo}$ and $\size{\tderivtwo'} < \size{\tderivtwo}$ if $\tmtwo \toeo \tmtwo'$.
		\end{enumerate}
		We can then build the derivation 
		\begin{equation*}
			\tderiv = 
			\begin{prooftree}
				\hypo{}
				\ellipsis{$\tderivtwo$}{\typctxtwo, \var \hastype \mtype \vdash \solvctxtwop{\tmtwo} \hastype \mtypetwo}
				\hypo{}
				\ellipsis{$\tderivthree$}{\typctxthree \vdash \tmthree \hastype \mtype}
				\infer2[\footnotesize$\Es$]{\typctxtwo \uplus \typctxthree \vdash \solvctxtwop{\tmtwo} \esub{\var}{\tmthree} \hastype \mtypetwo}
			\end{prooftree}
		\end{equation*}
		where $\typctx = \typctxtwo \uplus \typctxthree$ and
		\begin{enumerate}
			\item $\sizem{\tderiv'} = \sizem{\tderivtwo'} + \sizem{\tderivthree} = \sizem{\tderivtwo} + \sizem{\tderivthree} - 2 = \sizem{\tderiv} - 2$ and $\size{\tderiv'} = \size{\tderivtwo'} + \size{\tderivthree} +1 = \size{\tderivtwo} -1 + \size{\tderivthree} +1 = \size{\tderiv} - 1$ if $\tmtwo \tomo \tmtwo'$ (\ie if $\tm \tosolvm \tm'$); 
			\item $\sizem{\tderiv'} = \sizem{\tderivtwo'} + \sizem{\tderivthree} = \sizem{\tderivtwo} + \sizem{\tderivthree} = \sizem{\tderiv}$ and $\size{\tderiv'} = \size{\tderivtwo'} + \size{\tderivthree} +1 < \size{\tderivtwo} + \size{\tderivthree} +1 = \size{\tderiv}$ if $\tmtwo \toeo \tmtwo'$ (\ie if $\tm \tosolve \tm'$).
		\end{enumerate}

		\item \emph{Application}, \ie $\solvctx = \solvctxtwo\tmthree$. 
		Then, $\tm = \solvctxp{\tmtwo} = \solvctxtwop{\tmtwo} \tmthree \Rew{\solvsym a} \solvctxtwop{\tmtwo'} \tmthree = \solvctxp{\tmtwo'} = \tm'$ with $\tmtwo \Rew{\wsym a} \tmtwo'$ and $a \in \{\msym, \esym\}$.
		The derivation $\tderiv'$ is necessarily
		\begin{equation*}
			\tderiv' = 
			\begin{prooftree}
			\hypo{}
			\ellipsis{$\tderivtwo'$}{\typctxtwo \vdash \solvctxtwop{\tmtwo'} \hastype \mset{\larrow{\mtype}{\mtypetwo}}}
			\hypo{}
			\ellipsis{$\tderivthree$}{\typctxthree \vdash \tmthree \hastype \mtype}
			\infer2[\footnotesize$\ruleAp$]{\typctxtwo \uplus \typctxthree \vdash \solvctxtwop{\tmtwo'} \tmthree \hastype \mtypetwo}
			\end{prooftree}
		\end{equation*}
		where $\typctx = \typctxtwo \uplus \typctxthree$.
		By \ih  applied to $\tderivtwo$ (since $\mset{\larrow{\mtype}{\mtypetwo}}$ is a unitary solvable multi type), 
there 
is a derivation $\concl{\tderivtwo}{\typctxtwo}{\solvctxtwop{\tmtwo}}{\mset{\larrow{\mtype}{\mtypetwo}}}$ with: 
		\begin{enumerate}
			\item $\sizem{\tderivtwo'} = \sizem{\tderivtwo} - 2$ and $\size{\tderivtwo'} = \size{\tderivtwo} - 1$ if $\tmtwo \tomo \tmtwo'$; 
			\item $\sizem{\tderivtwo'} = \sizem{\tderivtwo}$ and $\size{\tderivtwo'} < \size{\tderivtwo}$ if $\tmtwo \toeo 
\tmtwo'$.
		\end{enumerate}
		We can then build the derivation 
		\begin{equation*}
		\tderiv = 
		\begin{prooftree}
		\hypo{}
		\ellipsis{$\tderivtwo$}{\typctxtwo \vdash \solvctxtwop{\tmtwo} \hastype \mset{\larrow{\mtype}{\mtypetwo}}}
		\hypo{}
		\ellipsis{$\tderivthree$}{\typctxthree \vdash \tmthree \hastype \mtype}
		\infer2[\footnotesize$\ruleAp$]{\typctxtwo \uplus \typctxthree \vdash \solvctxtwop{\tmtwo} \tmthree \hastype \mtypetwo}
		\end{prooftree}
		\end{equation*}
		where $\typctx = \typctxtwo \uplus \typctxthree$ and
		\begin{enumerate}
			\item $\sizem{\tderiv'} = \sizem{\tderivtwo'} + \sizem{\tderivthree} +1 = \sizem{\tderivtwo} -2 + \sizem{\tderivthree} +1 = \sizem{\tderiv} - 2$ 
			and $\size{\tderiv'} = \size{\tderivtwo'} + \size{\tderivthree} +1 = \size{\tderivtwo} - 1 + \size{\tderivthree} +1 = \size{\tderiv} - 1$ if $\tmtwo \tomo \tmtwo'$ (\ie if $\tm \tosolvm \tm'$); 
			\item $\size{\tderiv'} = \size{\tderivtwo'} + \size{\tderivthree} +1 = \size{\tderivtwo} + \size{\tderivthree} +1 
= \size{\tderiv}$ and $\size{\tderiv'} = \size{\tderivtwo'} + \size{\tderivthree} +1 < \size{\tderivtwo} + 
\size{\tderivthree} +1 = \size{\tderiv}$ if $\tmtwo \toeo \tmtwo'$ (\ie if $\tm \tosolve \tm'$).
		\end{enumerate}
	\end{itemize}

This completes the proof for the unitary solvable case.

In the solvable case (\ie under the weaker hypothesis that  $\mtypetwo$ is a solvable multi type), the proof is 
analogous to the one for unitary solvable, except for the \emph{Abstraction} case.
Indeed, in the base case (\emph{Hole context}) unitary solvability does not play any role, and the other cases follow 
from the \ih in a way analogous to unitary solvable.
Let us see the only substantially different case: 
\begin{itemize}
	\item \emph{Abstraction}, \ie $\solvctx = \la{\var}{\solvctxtwo}$. 
	So, $\tm = \solvctxp{\tmtwo} = \la{\var}{\solvctxtwop{\tmtwo}} \Rew{\solvredsym a} \la{\var}{\solvctxtwop{\tmtwo'}} = 
	\solvctxp{\tmtwo'} = \tm'$ with $\tmtwo \Rew{\wsym a} \tmtwo'$ and $a \in \{\msym, \esym\}$.
	Since $\mtypetwo$ is a solvable multi type by hypothesis, it has the form $\mtypetwo = \mset{\larrow{\mtype_1}{\mtypetwo_1}, \dots, \larrow{\mtype_n}{\mtypetwo_n}}$ for some $n > 0$, where $\mtypetwo_j$ is solvable for all $1 \leq j \leq n$.
	Thus, the derivation $\tderiv$ is necessarily
	\begin{equation*}
	\tderiv' = 
	\begin{prooftree}[separation=1em]
	\hypo{}
	\ellipsis{$\tderivtwo_j'$}{\typctx_j, \var \hastype \mtype_j \vdash \solvctxtwop{\tmtwo'} \hastype \mtypetwo_j}
	\infer1[\footnotesize$\lambda$]{\typctx_j \vdash \la{\var}\solvctxtwop{\tmtwo'} \hastype 
		\larrow{\mtype_j}{\mtypetwo_j}}
	\delims{ \left( }{ \right)_{1 \leq j \leq n} }
	\infer1[\footnotesize$\ruleManyVal$]{ \bigmplus_{j=1}^n \typctx_{j} \vdash \la{\var}\solvctxtwop{\tmtwo'} \hastype  
		\bigmplus_{j=1}^n \mset{\larrow{\mtype_j}{\mtypetwo_j}}}
	\end{prooftree}
	\end{equation*}
	For all $1 \leq j \leq n$, by \ih, there is a derivation $\concl{\tderivtwo_j}{\typctx_j, \var \hastype \mtype_j}{\solvctxtwop{\tmtwo}}{\mtypetwo_j}$ with: 
	\begin{enumerate}
		\item $\sizem{\tderivtwo_j'} \leq \sizem{\tderivtwo_j} - 2$ and $\size{\tderivtwo_j'} < \size{\tderivtwo_j}$ 
		if $\tmtwo \tomo \tmtwo'$ ; 
		\item $\sizem{\tderivtwo_j'} = \sizem{\tderivtwo_j}$ and $\size{\tderivtwo_j'} < \size{\tderivtwo_j}$ if $\tmtwo \toeo \tmtwo'$.
	\end{enumerate}
	We can then build the derivation 
	\begin{equation*}
	\tderiv = 
	\begin{prooftree}[separation=1em]
	\hypo{}
	\ellipsis{$\tderivtwo_j$}{\typctx_j, \var \hastype \mtype_j \vdash \solvctxtwop{\tmtwo} \hastype \mtypetwo_i}
	\infer1[\footnotesize$\lambda$]{\typctx_j \vdash \la{\var}\solvctxtwop{\tmtwo} \hastype 
		\larrow{\mtype_j}{\mtypetwo_j}}
	\delims{ \left( }{ \right)_{1 \leq j \leq n} }
	\infer1[\footnotesize$\ruleManyVal$]{ \bigmplus_{j=1}^n \typctx_{j} \vdash \la{\var}\solvctxtwop{\tmtwo} \hastype  
		\bigmplus_{j=1}^n \mset{\larrow{\mtype_j}{\mtypetwo_j}}}
	\end{prooftree}
	\end{equation*}
	where
	\begin{enumerate}
		\item $\sizem{\tderiv'} = \sum_{j=1}^n(\sizem{\tderivtwo_j'} +1) \leq \sum_{j=1}^n(\sizem{\tderivtwo_j} + 1 - 2) = \sizem{\tderiv} - 2n \leq \sizem{\tderiv} - 2$ (where the last inequality holds because $n >0$) and $\size{\tderiv'} =  \sum_{j=1}^n(\size{\tderivtwo_j'} +1) =  \sum_{j=1}^n(\size{\tderivtwo_j} + 1 - 1) \leq \size{\tderiv} - n < \size{\tderiv}$ (the last inequality holds because $n >0$) if $\tmtwo \tomo \tmtwo'$, \ie $\tm \tosolvm \tm'$; 
		\item $\sizem{\tderiv'} = \sum_{j=1}^n(\sizem{\tderivtwo_j'} +1) = \sum_{j=1}^n(\sizem{\tderivtwo_j} + 1) = \sizem{\tderiv}$ and $\size{\tderiv'} = \sum_{j=1}^n(\size{\tderivtwo_j'} +1) < \sum_{j=1}^n(\size{\tderivtwo_j} + 1) = \size{\tderiv}$ if $\tmtwo \toeo \tmtwo'$, \ie $\tm \tosolve \tm'$ (the inequality holds because $n > 0$).
		\qedhere
	\end{enumerate}
\end{itemize}	
\end{proof}

\begin{theorem}[Solvable completeness]
	\label{thmappendix:solvable-completeness}
	\NoteState{thm:solvable-completeness}
	Let $\deriv \colon \tm \tosolv^* \tmtwo$ be an $\solvredsym$-normalizing evaluation. 
	Then there is a derivation $\concl{\tderiv}{\typctx}{\tm}{\mtypetwo}$ with $\typctx$ inert, $\mtypetwo$ precisely solvable and $2\sizem{\deriv} + \sizes{\tmtwo} = \sizem{\tderiv}$.
\end{theorem}

\begin{proof}
	By induction on the length $\size{\deriv}$ of the $\solvredsym$-evaluation $\deriv$.
	
	If $\size{\deriv} = 0$ then $\sizem{\deriv} = 0$ and $\tm = \tmtwo$ is $\solvredsym$-normal and so a solvable fireball.
	By typability of solvable fireballs (\Cref{prop:precise-solvable-typability-nf}), there is a derivation $\concl{\tderiv}{\typctx}{\tm}{\mtypetwo}$ with $\typctx$ inert and $\mtypetwo$ precisely solvable.
	Thus, $\sizem{\tderiv} = \sizes{\tmtwo} = \sizes{\tmtwo} + 2\sizem{\deriv}$ by \Cref{l:size-solvable-nf}.
	
	Otherwise, $\size{\deriv} > 0$ and $\deriv$ is the concatenation of a first step $\tm \tosolv \tmthree$ and an evaluation $\deriv' \colon \tmthree \tosolv^* \tmtwo$, with $\size{\deriv} = 1 + \size{\deriv'}$.
	By \ih, there is a derivation $\concl{\tderivtwo}{\typctx}{\tmthree}{\mtypetwo}$ with $\typctx$ inert, $\mtypetwo$ precisely solvable and $\sizem{\tderivtwo} = \sizes{\tmtwo} + 2\sizem{\deriv'}$. 
	By solvable subject expansion (\Cref{prop:solvable-subject-expansion}), there is a derivation 	$\concl{\tderiv}{\typctx}{\tm}{\mtypetwo}$ with 
	\begin{itemize}
		\item $\sizem{\tderiv} = \sizem{\tderivtwo} + 2 = \sizes{\tmtwo} + 2\sizem{\deriv'} + 2 = \sizes{\tmtwo} + 2\sizem{\deriv}$ 
		if $\tm \tosolvm \tmthree$, since $\sizem{\deriv} = \sizem{\deriv'} + 1$;
		\item $\sizem{\tderiv} = \sizem{\tderivtwo} = \sizes{\tmtwo} + 2\sizem{\deriv'} = \sizes{\tmtwo} + 2\sizem{\deriv}$ if $\tm \tosolve \tmthree$, since $\sizem{\deriv} = \sizem{\deriv'}$.
		\qedhere
	\end{itemize} 
\end{proof}

%
%




\end{document}